\documentclass[12pt]{amsart}
\usepackage[margin=1in]{geometry}
\usepackage{lmodern}
\usepackage{mathptmx}
\usepackage{mathtools}
\usepackage{amsthm}
\usepackage{amssymb}
\usepackage{amsrefs}
\usepackage{graphicx}
\usepackage{tikz}
	\usetikzlibrary{decorations.pathreplacing}
	\usetikzlibrary{patterns}
\usepackage{caption}
\usepackage{csquotes}
\usepackage[usestackEOL]{stackengine}
\usepackage{scalerel}
\usepackage{calc}
\usepackage{accents}
\usepackage[new]{old-arrows}
\usepackage{subcaption}

\MakeOuterQuote{"}
\newcommand{\cd}{\cdot}
\newcommand{\ra}{\rightarrow}
\newcommand{\pr}{\prime}
\newcommand{\de}{\partial}
\newcommand{\te}{\theta}
\newcommand{\C}{\mathbb{C}}
\newcommand{\Q}{\mathbb{Q}}
\newcommand{\R}{\mathbb{R}}
\newcommand{\N}{\mathbb{N}}
\newcommand{\Z}{\mathbb{Z}}
\newcommand{\A}{\mathcal{A}}
\newcommand{\B}{\mathcal{B}}
\newcommand{\OO}{\mathcal{O}}
\newcommand{\co}{\mathcal{O}}
\DeclareMathOperator{\Id}{Id}
\DeclareMathOperator{\Img}{Im}

\DeclareMathOperator{\flux}{flux}

\newcommand{\be}{\begin{equation}}
\newcommand{\ee}{\end{equation}}

\newcommand{\cC}{{\mathcal C}}
\newcommand{\cB}{{\mathcal B}}
\newcommand{\cD}{{\mathcal D}}
\newcommand{\cE}{{\mathcal E}}
\newcommand{\cF}{{\mathcal F}}
\newcommand{\Int}{{\rm Int}}
\newcommand{\Ext}{{\rm Ext}}
\newcommand{\Bd}{{\rm Bd}}
\newcommand{\cPalphaF}{{\mathcal P}(\alpha,F)}
\newcommand{\Ind}{{\rm flux}}
\newcommand{\Swap}{{\rm Swap}}

\DeclareMathAlphabet{\mathcal}{OMS}{cmsy}{m}{n}

\renewcommand{\coprod}{\rotatebox[origin = c]{180}{$\prod$}}
\newcommand{\dbtilde}[1]{\accentset{\approx}{#1}}
\newcommand{\hookuparrow}{\mathrel{\rotatebox[origin=c]{90}{$\hookrightarrow$}}}
\def\myupbracefill#1{\rotatebox{90}{\stretchto{\{}{#1}}}
\def\rlwd{.5pt}
\newcommand\notate[4][B]{
  \if B#1\else\def\myupbracefill##1{}\fi
  \def\useanchorwidth{T}
  \setbox0=\hbox{$\displaystyle#2$}
  \def\stackalignment{c}\stackunder[-6pt]{
    \def\stackalignment{c}\stackunder[-1.5pt]{
      \stackunder[2pt]{\strut $\displaystyle#2$}{\myupbracefill{\wd0}}}{
    \rule{\rlwd}{#3\baselineskip}}}{
  \strut\kern13pt$\rightarrow$\smash{\rlap{$~\displaystyle#4$}}}
}
\newcommand{\verteq}{\rotatebox{90}{$\,=$}}
\newcommand{\equalto}[2]{\underset{\overset{\mkern4mu\verteq}{#2}}{#1}}

\newtheorem{thm}{Theorem}[section]
\newtheorem{innercustomthm}{Theorem}
\newenvironment{customthm}[1]
  {\renewcommand\theinnercustomthm{#1}\innercustomthm}
  {\endinnercustomthm}

\theoremstyle{defin}
\newtheorem{defin}{Definition}[section]
\newtheorem*{defin*}{Definition}
\newtheorem{claim}{Claim}
\newtheorem*{note}{Note}
\newtheorem*{fact}{Fact}
\newtheorem*{remark}{Remark}
\theoremstyle{plain}
\newtheorem{lemma}{Lemma}[section]

\begin{document}

\title{Classification of Quantum Cellular Automata}

\author{Michael Freedman}
\address{\hskip-\parindent
	Michael Freedman\\
    Microsoft Research, Station Q, and Department of Mathematics\\
    University of California, Santa Barbara\\
    Santa Barbara, CA 93106\\}
\email{mfreedman@math.ucsb.edu}

\author{Matthew B.~Hastings}
\address{\hskip-\parindent
	Matthew Hastings\\
    Microsoft Research, Station Q\\
    University of California, Santa Barbara\\
    Santa Barbara, CA 93106\\}
\email{mfreedman@math.ucsb.edu}

\begin{abstract}
\end{abstract}

\begin{abstract}
There exists an index theory to classify strictly local quantum cellular automata in one dimension\cite{Gross2012,fermionGNVW1,fermionGNVW2}.  We consider two classification questions.  First, we study to what extent this index theory can be applied in higher dimensions via dimensional reduction, finding a classification by the first homology group of the manifold modulo torsion.  Second, in two dimensions, we show that an extension of this index theory (including torsion)
fully classifies quantum cellular automata, at least in the absence of fermionic degrees of freedom.  This complete classification in one and two dimensions by index theory is not expected to extend to higher dimensions due to recent evidence of a nontrivial automaton in three dimensions\cite{FHH}.
Finally, we discuss some group theoretical aspects of the classification of quantum cellular automata and consider these automata on higher dimensional real projective spaces.
\end{abstract}
\maketitle

\section{Introduction}
Quantum cellular automata (QCA) in one-dimension have been considered by several authors, in particular in Ref.~\cite{Gross2012}, and the existence of a nontrivial index has been shown.
In this paper, we first use topological arguments to push this index theory as far as we can in higher dimensions.  The result is, perhaps not surprisingly, a classification using the first homology group of the manifold; however, the classification does not detect torsion in this group.
These arguments begin with applications of covering space theory to the one dimensional index, but then require a long digression into algebra to establish a key dimension bound.
In a second part of this paper, we push algebraic arguments further for two dimensional QCA and show that the first homology group (including torsion) gives a full classification, at least for QCA without fermionic degrees of freedom.

Before describing our results, let us describe QCA in a general setting, as well as mentioning some possible generalizations.
A QCA is an automorphism of the $*$-algebra of operators acting on the Hilbert space of some system, subject to certain
locality constraints on the map described later.  For a finite system, this map can be described by
by conjugation by a unitary: for a QCA $\alpha$ and an operator $O$, we map $\alpha(O)=U^\dagger O U$ for some unitary $U$.
For this paper, we will restrict to this case of finite systems, giving bounds uniform in system size.

One example of QCA is a {\it finite depth quantum circuit} (fdqc), where the corresponding unitary is obtained by composing several individual unitaries, called {\it gates} such that
each gate acts on a set of some bounded diameter and such that the circuit has a bounded depth.
However, QCA that cannot be described by an fdqc also exist.  More precisely, since on any finite system any QCA can be represented by some quantum circuit, we mean that families of QCA exist on systems of increasing size with uniform bounds on the range of the QCA (defined below) such that the family of QCA cannot be described by a family of quantum circuits with uniformly bounded depth and range.  We describe several different ways of classifying possible QCA later.

For any graph $G$ with vertex set $V$, we can define a QCA (using some additional data) as follows.
We call the vertices of the graph {\it sites}, and 
each site has finite-dimensional Hilbert space (not necessarily the same on all sites) and the Hilbert space of
the whole system is the tensor product of these Hilbert spaces.
We call these finite-dimensional Hilbert spaces "degrees of freedom".
We use the graph metric on this graph to define a distance.
For us, a QCA of range $R$ will be an automorphism $\alpha$ of the algebra of operators on the Hilbert space of the system, such that for any operator $O_x$ supported on a site $x$,
we have that $\alpha(O_x)$ is supported on the set of sites within distance $R$ of site $x$.
Note that every QCA with range $R$ is a QCA with range $R'$ for all $R'\geq R$.
More generally, rather than a graph, we can consider degrees of freedom corresponding to sites  $\{x_i\}$ in a space $X$, called the "control space," where $X$ is a metric space. The simplest cases are when $X$ is a line, circle $S^1$, or graph $G$; we then define the range of the QCA using the metric of the control space.

If the graph is the graph of a $d$-dimensional hypercubic lattice with periodic boundary conditions (i.e., 
sites are labelled by a string of $d$ integers, identifying sites modulo some integer $L$, which is the side length of the cube)
we will call this a QCA in $d$-dimensions.  More generally, if the control space is a $d$-dimensional manifold then this
also is termed a QCA in $d$-dimensions.  
The original case of a one-dimensional QCA is the case that sites of the graph are labelled by a single integer, as studied previously\cite{Gross2012}.

The graph definition of a QCA with a sharp radius $R$ is the definition that we will use.  However, it is interesting to note that there are at least two possible ways that the definition can be generalized.  First, one might consider the case that the locality is not strict: one might require that for any $R$, $\alpha(O_x)$ can be approximated by an operator supported on the sets of sites within distance $R$ of $x$, up to an error that decays rapidly (perhaps exponentially) with $R$.  This case was called a "locality preserving unitary" in Ref.~\cite{Hastings2013}.
Even in the simplest 1D case, with $x_i$ located at the integer site $i \in \R$, it is not yet proven that the GNVW index, discussed below, is well-defined with such an exponential decay assumption. We will shortly see that Wedderburn's basic theorems on the structure of finite dimensional $C^\ast$-algebras and their subalgebras play a key role in defining the index. It appears that the most natural way to deal with exponential tails is to simply truncate them and accept that post-truncation what previously was a $\ast$-subalgebra now is only approximately closed under the operations $+$ and $\cd$. Repairing the proof seems to require a form of Hyers-Ulam-Rassias stability in the context of Wedderburn theory. For example, one would like to know if an approximate $\ast$-subalgebra of a full matrix algebra is in fact near an exact $\ast$-subalgebra.

Another possible generalization is to consider fermionic operators which anti-commute on distinct sites rather than commuting; this case was considered in one-dimension in Ref.~\cite{fermionGNVW1,fermionGNVW2}.
One might also combine these two possibilities, so at least four interesting cases exist, depending on whether one requires strict locality or not and depending on whether one considers fermions or not.  Here we consider the definition above, which is the case of strict locality without fermions.  We expect that the arguments in the topological section can be extend to the fermionic case without difficulty.

There are at least three interesting classification questions one might consider.  The first question is whether a given QCA can be described by a quantum circuit of bounded depth and range.
Two other questions are "path equivalence", i.e., whether there is a continuous path of QCA connecting two given QCA $\alpha,\beta$, and "blending", i.e., given two QCA $\alpha,\beta$ and two vertex sets $S,T$, whether there is a QCA with the same action as $\alpha$ for operators supported on $S$ and the same action as $\beta$ for operators supported on $T$ ("blending" is a generalization of "crossover" from Ref.~\cite{Gross2012}; we use the term blending to match terminology from topology).  We then outline the paper and sketch some of the results.

Let us recall that in Ref.~\cite{Gross2012} all of these classification questions were completely answered for one-dimensional QCA.  The key to their result was defining a certain index that was (in that paper) defined to be a positive rational, with this index being the only obstruction to path equivalence and blending.  It will be slightly more convenient for us to consider the logarithm of the quantity that they defined.  We call this logarithm the GNVW index after the authors of Ref.~\cite{Gross2012}, i.e., the GNVW index is equal to $\log(p/q)$ for integer $p,q$.
Sometimes we also use the physically suggestive word "flux" for this logarithm.

We use $\Id$ to represent the identity QCA and $I$ to represent an identity operator.
All algebras that we consider are $*$-algebras.

\subsection{Path Equivalence}
The second interesting classification question is whether or not two distinct QCA $\alpha_0,\alpha_1$ can be connected by a continuous path $\alpha_s$ while keeping the {\it range} of the QCA $\alpha_s$
bounded by some $R'$ which is a constant times the range $R$ of $\alpha_0$, with the constant independent of system size, local Hilbert space dimension, and all other data.
Since we study finite systems, by a "continuous path" we simply mean a continuous path of unitaries acting in some finite dimensional unitary group, i.e., we use a subspace topology where we restrict to unitaries which obey the locality conditions; for an infinite system it must be defined more carefully.
However, we will consider a slightly different question.  Similarly to what is done when studying phases of free fermion Hamiltonians\cite{Kitaev_2009,Hastings2013},
we consider two modifications of this question.

First we define equivalence of QCA by paths:
\begin{defin}
Given two QCA $\alpha,\beta$ with the same graph $G$ for each and the same Hilbert space on each site for each, we say that they are $R'$-path equivalent if
there exists a continuous path of QCA $\alpha_t$ with range $R'$ with $\alpha_0=\alpha,\alpha_1=\beta$.
\end{defin}

 Our first modification of the question is to modify the notion of equivalence by considering
a broader notion of equivalence that allows
 {\it stabilization} by {\it tensoring with additional degrees of freedom}, just as is done when considering free fermionic phases.
For each site, we will define a way to increase the Hilbert space dimension on that site.  If site $i$ had some Hilbert space dimension $d_i$, we allow one to 
increase the Hilbert space dimension by increasing the dimension to $d_i d'_i$ for any positive integer $d'_i$.  This may be done for any number of sites; we refer to the
tensor factor of dimension $d_i$ as the "original degree of freedom" on that site and the tensor factor of dimension $d'_i$ as the "ancilla degree of freedom on that site".  We replace the QCA $\alpha$ by $\alpha \otimes \Id$, where 
$\alpha$ acts on the original degrees of freedom and $\Id$ acts on the ancilla degrees of freedom.
This motivatives our next definition:
\begin{defin}
Given two QCA $\alpha,\beta$ with the same graph $G$ for each, we say that $\alpha,\beta$ are stably  $R'$-equivalent if
one may tensor with additional degrees of freedom so that
$\alpha \otimes \Id$ is $R'$-path equivalent to $\beta \otimes \Id$.

Note that $\alpha,\beta$ may have different Hilbert space dimensions on each site from each other, and so we may tensor with degrees of freedom with certain dimensions (for example, $\alpha$ may have dimension $3$ for some site, $\beta$ may have dimension $5$, and perhaps for both we increase the dimension to $15$) so that the tensor factor $\Id$ may differ in $\alpha \otimes \Id$
from that in $\beta \otimes \Id$.
\end{defin}

\begin{remark}
For brevity, we will sometimes just say "stably equivalent", dropping the word "path" and the $R'$ if one can pick $R'=O(R)$ independent of system size and Hilbert space dimension.
\end{remark}

\begin{remark}
Later, we will sometimes describe the tensor process in several steps.  For example, for a single site, we may describe tensoring in some additional degree of freedom with dimension $D'$, then another with dimension $D''$, and so on; the net effect is just to tensor in one degree of freedom with dimension $D' \cdot D'' \cdot \ldots$.
\end{remark}

The second way in which we modify the question can be explained by an example.  Suppose we have a system with the topology of a two torus.  For example, the sites are labelled by a pair of integers, $(i,j)$, in the range $0,L-1$ for some $L$ and we have an edge between sites $(i,j)$ and $(k,l)$ if $i=k \pm 1$ mod $L$ and $j=l$ or $i=k$ and $j=l\pm1$ mod $L$.
Since we will refer to this graph several times later, we call it $H(L)$.
Suppose we have the same dimension $d_i$ on each site and suppose $\alpha$ acts trivially on all sites, except the sites with $i=0$, where it acts as a shift: it maps an operator on $(i,j)$ to the corresponding operator on $(i,j+1)$.  This QCA $\alpha$ is not stably$R'$-equivalent to the identity for $R'=O(1)$; this can be seen by ignoring the $i$ coordinate and regarding this QCA as a one-dimensional QCA with sites labelled by $j$.  Then, this one-dimensional QCA has a nontrivial GNVW index, and hence it is not path equivalent to $\Id$.

This nontrivial index for the one-dimensional QCA is a simple example showing the application of the index theory to higher dimensions; section \ref{topologysection} considers this in more generality.
Indeed, two independent GNVW indices can be described for any QCA on a torus as one may ignore either of two coordinates (and other indices can be defined via Eq.~(\ref{fluxadds} later).

However, at the same time, we certainly would not like to consider this QCA as describing some interesting, truly two-dimensional behavior since it is simply an example of a one-dimensional shift QCA.
The way we will deal with this question is that we will show that some QCA is stably equivalent to another
QCA which acts trivially on all sites except those on some lower dimensional subset; for example, in this case, the QCA acts trivially everywhere except on one line while in general on a torus it may act trivially everywhere except on two lines.  We will explain this more precisely later.

Let us describe an interesting modification of this example. Consider the same graph $H(L)$.  Let $\alpha$ act trivially on all sites except those with $i=0$ and those with $i=L/2$.
On the sites with $i=0$, it acts as a shift: it maps an operator on $(i,j)$ to the corresponding operator on $(i,j+1)$, while on the sites with $i=L/2$ it acts as a shift in the opposite
direction, mapping an operator on $(i,j)$ to the corresponding operator on $(i,j-1)$; let us call these "positive" and "negative" flows respectively. It is not hard to show that this is $O(1)$-path equivalent to $\Id$ of QCA, as one may simply deform the two different flows towards each other along the path (i.e., describe a path where first we deform it to a QCA with a positive flow along the line $i=1$ and negative flow along the line $i=L/2-1$, then deform to shifts along lines $i=2$ and $i=L/2-2$, and so on, until the lines meet and we can cancel the shifts).
However, this does not give a quantum circuit of bounded depth and range: the depth diverges with $L$.

However, it is possible to give a quantum circuit of bounded depth and range using a trick reminiscent of the "Eilenberg swindle" in topology as follows. First, find a quantum circuit of bounded depth and range that has a positive flow for $i=0$ and a negative flow for $i=-1$; this follows from results of GNVW as we can treat the two lines $i=0$ and $i=1$ as a single one-dimensional QCA and the flow for that QCA vanishes.
In parallel, create positive flow for $i=2$ and negative flow for $i=3$ and so on, creating positive flow for $2k$ and negative flow for $2k+1$ for $2k=0,\ldots,L/2-1$.  Then, cancel the negative flow for $i=1$ against the positive flow for $i=2$, cancel the negative flow for $i=3$ against the positive flow for $i=4$ and so on, until all that is left is the flows for $i=0$ and $i=L/2$.
We will give a more detailed treatment of this later, once we define the incompressible flow in two dimensions.

Also, another subtlety occurs when we ask whether path equivalence to $\Id$ implies that something can be described by a quantum circuit.  On a finite system, if the path is smooth enough, we can describe the QCA by conjugation by some unitary $U$ written as
$$U={\mathcal S} \exp(\int_0^1 \eta_s {\rm d}s),$$
where the notation ${\mathcal S}$ means that we are computing an $s$-ordered exponential and $\eta_s$ is some anti-Hermitian matrix.
One question is whether $\eta_s$ is a sum of uniformly (in system size) bounded local operators (in the example above, the path where we moved the two paths each a distance of order $L$ would not allow $\eta_s$ to be a sum of uniformly bounded local operators as we deform the paths by a distance $~L$ over the same interval of $s$; however, the other path where we created and cancelled pairs of flows does allow $\eta_s$ to be written as a sum of uniformly bounded local operators).
Another problem though is whether or not this evolution can be described by a quantum circuit; certainly it can be approximated by a quantum circuit, but that only approximates the final unitary.
We do not consider these questions further; for the particular case of two-dimensions, we are fully able to understand when a QCA can be described by a quantum circuit and we do not yet consider higher dimensional cases beyond the discussion in section \ref{topologysection}.

\subsection{Blending}
The third classification question is blending.
We define:
\begin{defin}
Let $\alpha,\beta$ be QCA.
We say that $\alpha,\beta$`agree on $S$ if $\alpha(O)=\beta(O)$ for all $O$ supported on $S$.
\end{defin}

\begin{defin}
Given two QCA $\alpha,\beta$ and given two sets $S,T$, we say that a QCA $\gamma$ blends between $\alpha$ on $S$ and $\beta$ on $T$ if
$\gamma$ agrees with $\alpha$ on $S$ and $\gamma$ agrees with $\beta$ on $T$.
\end{defin}

Now, an interesting question is: given two QCA,
$\alpha,\beta$, both of range $R$, and which are disjoint with large distance between $S$ and $T$, does there exist another QCA $\gamma$ of range $O(R)$ which blends between $\alpha$ on $S$ and $\beta$ on $T$?

\subsection{Families of QCA and Circuits}
As noted at the start, we actually consider {\it families} of QCA, which we will take to be sequences $\{\alpha_i\}$, $0\leq i \leq \infty$.  Our use of the big-O notation $O(R)$ above, for example, implicitly refers to a family of QCA all with the same range.  Families are mostly implicit, rather than explicit, in this paper, but we now discuss this point further.

As an example, consider a QCA with range $R$ on a circle.  So long as $R$ is much smaller than the circumference $C$ of the circle, the GNVW index is well-defined.  However, composition of QCAs will increase the range, so that for any finite system eventually the range will become comparable to $C$ and the index will not be well-defined.  Indeed, for any finite number of degrees of freedom on a circle, a shift (even a shift by a single site) will have finite order. 
Families of QCAs give a mathematically precise solution to this problem: for a control space which is some manifold, each family consists of a sequence of QCAs on the same control space, with the radius of all QCA being fixed at $O(1)$ and with the metric on the control space being rescaled in the family so that the injectivity radius of the control space diverges.  Then, any finite composition will give an index which is well-defined for all but finitely many QCA.

Note that, given an arbitrary family of QCA on the same control space,
it is possible that different QCA in a family have different values of the indices that we define.  We show, however, that given two families $F_1,F_2$ of QCA, such that the indices of the QCA in $F_1$ differ from the corresponding indices in $F_2$, it is not possible to relate $F_1$ to $F_2$ by a family of fdqc.  

Nevertheless, the reader may wonder if there is a good way to define a family of QCA so that all QCA in the family are related by an fdqc (and as a result all but finitely many will have the same value of the indices we define).  We propose the following definition of such a {\it coherent} family.  Before giving the definition, let us rescale the metric for each QCA in the family, so that all QCA in the family have the same metric on the same control space so now the sequence of QCA $\alpha_0,\alpha_1,\ldots$ in the family have ranges $R_0,R_1,\ldots$ with $R_i\rightarrow 0$ as $i\rightarrow \infty$.  Let us fix $R_i=2^{-i}$.  Then, we define:
 
 \begin{defin}
Such a family of QCA is "coherent" if there exists a constant $c$ such that for all $i$, QCA $\alpha_i$ is stably $cR_i$-equivalent to $\alpha_{i+1}$.
\end{defin}

In a subsequent paper\cite{cohf}, we will show how to construct a coherent family starting from any given $\alpha_0$ with the range of $\alpha_0$ sufficiently small compared to the scale of the control space, and we will prove uniqueness of this family (up to paths), allowing us to dispense with families and return to the naive definition of QCAs.
 
Remark: We have chosen to focus only on the equivalence of $\alpha_i$ to $\alpha_{i+1}$.  If we instead consider, for example, the equivalence of $\alpha_i$ to $\alpha_{i+100}$, this involves two QCA with very different scales and so the stable path equivalence involves QCA that are very nonlocal compared to $\alpha_{i+100}$.  However, the definition above implies that $\alpha_i$ is stably $cR_i$ equivalent to $\alpha_j$ for all $j>i$.  Also, we have fixed $R_i=2^{-i}$ so that $R_{i+1}$ is not much smaller than $R_i$.
 
Remark: as an example consider a family where $\alpha_i$ is a shift by one on a circle with $2^i$ sites.  This family is coherent because $\alpha_i$ is stably equivalent to $\alpha_{i+1}$ as follows.  Tensor in an additional $2^i$ sites in between the existing sites so that the result is a shift by two on half the sites.  Then this is connected by an fdqc to $\alpha_{i+1}$.

\subsection{Outline and Results}
In section \ref{topologysection}, we use apply the one dimensional index theory to higher dimensional QCA.  One of the results (on additivity of flux when summing homology classes) will require applying the {\it two} dimensional theory that we develop in section \ref{algebraicsection}.
In this section, we do not refer specifically to the graph when defining a QCA, but simply consider some manifold as a control space, and assume that degrees of freedom correspond to points within the control space, with the QCA being short range with respect to some metric in the control space.
In this section, the existence of a nontrivial index (or indices) for some QCA will imply that it is not stably path equivalent to $\Id$.  However, in this section we do {\it not} consider questions of blending or of whether a QCA is path equivalent to some QCA which acts trivially except on some lower dimensional subspace.
This is because one can generate QCA with nontrivial indices using QCA which are shifts on essential cycles of the manifold.

In section \ref{algebraicsection}, we push the algebraic methods further in two dimensions.  Here we consider questions of blending.  Roughly, we show that every QCA in two dimensions can be blended to the identity QCA outside any disc.  This allows us to give a full classification in terms of the first homology group.

Finally, in section \ref{RPn}, we consider the real projective space $RP^n$ as an interesting examples of a control space.  Here there is an interplay of the topological and algebraic methods.

\section{GNVW Index in One Dimension and Higher}
\label{topologysection}
We now use topological methods to apply the index theory in higher dimensions.
In this section,
all QCA will be considered up to stabilization, whereby we tensor with the identity on any additional, locally finite, degrees of freedom called \textit{ancilla}.
In this section, we describe the locality of the QCA using a more general control space $X$, rather than just a graph $G$.

In subsection \ref{GNVWreview}, we review the GNVW index.
We modify slightly the GNVW discussion and also use the more physical word \textit{flux} for the index.
In subsection \ref{multiplicativesection}, we also establish some multiplicativity properties for flux under finite covers which can be interpretted as naturality for $f_\alpha$ under certain finite covers (Theorems \ref{2dcover}, \ref{3dcover}, and \ref{4dcover}).
In subsection \ref{reductionsection} we will prove (Theorem 2.4) that on any higher dimensional manifold (or simplical complex) $X$, for $R > 0$, but sufficiently small for any $R$ QCA $\alpha$, there is a well-defined flux $f_\alpha \in H_1^{\text{lf}}(X;M)/$torsion. The group structure of QCA are also studied.
Subsection \ref{productsection} gives information on QCAs defined on products $X \times R$, also from topological sources. It shows how to replace a general QCA with a periodic one, retaining a "germ."

Finally we note that although not treated here, \cites{fermionGNVW1,fermionGNVW2} gave a quite satisfactory fermionic generalization of the GNVW index (or flux). The new feature is that the values are now in the $\Z$-module:
\[
M^\pr = M \cup \frac{1}{2} M
\]
The Majorana fermion with quantum dimension $\sqrt{2}$ can now also flow about leading to halves of logarithms. Presumably all the theorems of this section can be put through their fermionic machinery to yield natural generalizations.

\subsection{Review of GNVW Index}
\label{GNVWreview}

We now review the GNVW index.
We will work initially with the case where the \textit{control space} $X = \R$ (reals), $\{x_i\} = \Z$ (integers), and $R = 1$. The other 1D cases with sharp cut-off $R > 1$ readily reduce to this one by clustering sites. Later we progress beyond these 1-dimensional cases but much is still unknown in higher dimensions.

\begin{note}
When $\{x_i\}$ is infinite the $\otimes_i$ symbol stands for the direct limit over the net of finite tensor products. Correctly scaled, this becomes the hyperfinite type II factor. However, the reader may freely assume $\{x_i\}$ have a large period and that $X$ is a large circle, and later that all graphs $G$ are also compact. Infinities play no essential role in our discussion.
\end{note}

We work in the category of finite dimensional, unital $\ast$-algebras over $\C$. Suppose $\mathcal{A} \subset \mathcal{B}_1 \otimes \mathcal{B}_2$.

\begin{defin}
	The \textit{support algebra} $S(\mathcal{A}, \mathcal{B}_1)$ is

	\begin{enumerate}
		\item The smallest $S$ s.t. $\mathcal{A} \subset S \otimes \mathcal{B}_2$, or
		\item The intersection of all $S^\pr$ s.t. $\mathcal{A} \subset S^\pr \otimes \mathcal{B}_2$, or
		\item The algebra generated by $\{A_{v,\mu}\}$, where $\{A_\nu\}$ is a basis for $\mathcal{A}$ and $\{E_\mu\}$ a basis for $\mathcal{B}_2$ and $A_v = \sum_\mu A_{v,\mu} \otimes E_\mu$. Note that diagramatically

		\begin{tikzpicture}
			\node at (-3,0) {$A_{v,\mu} =$};
			\draw  (-2,0.5) rectangle (-0.5,-0.5);
			\node at (-1.25,0) {$A_\mu$};
			\draw (-1.8,0.5) -- (-1.8, 1.4);
			\node at (-1.8,1.7) {$J_1$};
			\draw (-1.8,-0.5) -- (-1.8, -1.4);
			\node at (-1.8,-1.7) {$J_1$};

			\draw (0.5, 0) arc (0:-146:0.8 and 0.9);
			\draw (0.5, 0) arc (0:146:0.8 and 0.9);
			\node at (1, 0) {$E_\mu$};
			\draw [fill = black] (0.5,0) circle (0.3ex);
			\node at (-0.9,1) {$J_2$};
			\node at (-0.9,-0.9) {$J_2$};
			\node at (-1.2,-2.5) {where \hspace{2em} $= \text{Id}_{B_2}$ and \hspace{2em} $= \text{tr}_{B_2}$};
			\draw (-2.25, -2.4) arc (0:-180:0.4 and 0.2);
			\draw (-0.15, -2.5) arc (180:0:0.4 and 0.2);
		\end{tikzpicture}
	\end{enumerate}
\end{defin}

$S(\mathcal{A}, \mathcal{B}_2)$ is similarly defined switching $1 \xleftrightarrow{} 2$.

\begin{lemma}
\label{supportcommute}
	Inside $\mathcal{B}_1 \otimes \mathcal{B}_2 \otimes \mathcal{B}_3$, assume $\mathcal{A} \subset \mathcal{B}_1 \otimes \mathcal{B}_2 \otimes \Id$, and $\mathcal{A}^\pr \subset \Id \otimes \mathcal{B}_2 \otimes \mathcal{B}_3$. If $\A$ and $\A^\pr$ commute, $[\A, \A^\pr] = 0$, then $[S(\A, \B_2), S(\A^\pr, \B_2)] = 0$.
\end{lemma}

\begin{proof}
	Write $A_v = \sum_\mu E_\mu \otimes A_{v,\mu}$, \hspace{0.25em} $A^\pr_{v^\pr} = \sum_{\mu^\pr} A^\pr_{v^\pr,\mu^\pr} \otimes E^\pr_{\mu^\pr}$, then

	\[
	0 = [A_v, A^\pr_{v^\pr}] = \sum_{\mu, \mu^\pr} E_\mu \otimes [A_{v,\mu}, A^\pr_{v^\pr,\mu^\pr}] \otimes E^\pr_{\mu^\pr}
	\]

	Since $\{E_\mu \otimes E^\pr_{\mu^\pr}\}$ are linearly independent, in fact a basis of $\B_1 \otimes B_3$, all $[A_{v,\mu}, A^\pr_{v^\pr, \mu^\pr}] = 0$ implying the two support algebras commute.

	Alternative diagramatic proof:

	\begin{center}
	\begin{tikzpicture}
		\node at (-3.2,0) {$A_{v, \mu} A^\prime_{v^\prime, \mu^\prime} =$};
		\draw  (-1,-0.1) rectangle (0.1,-0.5);
		\draw  (-0.3,0.3) rectangle (0.8,0.7);
		\draw (-0.2,1.3) -- (-0.2,0.7);
		\draw (-0.2, 0.3) -- (-0.2, -0.1);
		\draw (-0.2,-0.5) -- (-0.2,-1.1);
		\draw (-1.6, -0.3) arc (180:19:0.4 and 0.6);
		\draw (-1.6, -0.3) arc (-180:-19:0.4 and 0.6);
		\draw [fill = black] (-1.6, -0.3) circle (0.3ex);
		\node at (-1.9,-0.6) {$E_\mu$};
		\node at (-0.45,-0.32) {$A_v$};
		\node at (0.25,0.48) {$A^\prime_{v^\prime}$};
		\draw (1.4, 0.5) arc (0:161:0.4 and 0.6);
		\draw (1.4, 0.5) arc (0:-161:0.4 and 0.6);
		\draw [fill = black] (1.4, 0.5) circle (0.3ex);
		\node at (1.7,0.2) {$E_{\mu^\prime}$};

		\node at (2.4,0) {$=$};
		\draw  (4,-0.1) rectangle (5.1,-0.5);
		\draw  (4.7,0.3) rectangle (5.8,0.7);
		\draw (4.8,1.3) -- (4.8,0.7);
		\draw (4.8, 0.3) -- (4.8, -0.1);
		\draw (4.8,-0.5) -- (4.8,-1.6);
		\draw (3.4, 0) arc (180:-4:0.5 and 1.5);
		\draw (3.4, 0) arc (-180:-19.5:0.5 and 1.5);
		\draw [fill = black] (3.4,0.1) circle (0.3ex);
		\node at (3.1,0.3) {$E_\mu$};
		\node at (4.55,-0.32) {$A_v$};
		\node at (5.25,0.5) {$A^\prime_{v^\prime}$};
		\draw (6.4, 0) arc (0:152:0.5 and 1.5);
		\draw (6.4, 0) arc (0:-191:0.5 and 1.5);
		\draw [fill = black] (6.4, 0.1) circle (0.3ex);
		\node at (6.8,0) {$E_{\mu^\prime}$};

		\node at (7.7,0) {$=$};
		\draw [dashed] (3.8,1.1) rectangle (6,-1.3);
		\node at (7.7,-0.4) {apply};
		\node at (7.7,-0.9) {hypothesis};
		\node at (7.7,-1.4) {in dotted box};

		\draw  (-1.7,-2.3) rectangle (0.1,-2.7);
		\draw  (-0.5,-4) rectangle (1.3,-4.4);
		\draw (-0.2,-1.9) -- (-0.2, -2.3);
		\draw (-0.2,-2.7) -- (-0.2,-4);
		\draw (-0.2,-4.4) -- (-0.2, -4.9);
		\draw (-2.2, -3.4) arc (180:47:0.5 and 1.5);
		\draw (-2.2, -3.4) arc (180:387.5:0.5 and 1.5);
		\node at (-2.7,-3.4) {$=$};
		\draw [fill = black] (-2.1,-2.5) circle (0.3ex);
		\node at (-2.5,-2.5) {$E_\mu$};
		\draw (1.8, -3.4) arc (0:203.5:0.5 and 1.5);
		\draw (1.8, -3.4) arc (0:-138.5:0.5 and 1.5);
		\draw [fill = black] (1.7,-4.2) circle (0.3ex);
		\node at (2.1,-4.1) {$E^\prime_{\mu^\prime}$};
		\node at (-0.8,-2.5) {$A_v$};
		\node at (0.5,-4.2) {$A^\prime_{v^\prime}$};

		\node at (2.5,-3.4) {$=$};
		\node at (-3.2,0) {$A_{v, \mu} A^\prime_{v^\prime, \mu^\prime} =$};
		\draw  (4,-3.2) rectangle (5.1,-2.8);
		\draw  (4.7,-3.6) rectangle (5.8,-4);
		\draw (4.8,-2.2) -- (4.8,-2.8);
		\draw (4.8,-3.2) -- (4.8,-3.6);
		\draw (4.8,-4) -- (4.8,-4.6);
		\draw (3.4,-3) arc (180:19:0.4 and 0.6);
		\draw (3.4,-3) arc (-180:-19:0.4 and 0.6);
		\draw [fill = black] (3.4,-3) circle (0.3ex);
		\node at (3.1,-3.3) {$E_\mu$};
		\node at (4.6,-3) {$A_v$};
		\node at (5.2,-3.82) {$A^\prime_{v^\prime}$};
		\draw (6.4,-3.8) arc (0:161:0.4 and 0.6);
		\draw (6.4,-3.8) arc (0:-161:0.4 and 0.6);
		\draw [fill = black] (6.4,-3.8) circle (0.3ex);
		\node at (6.7,-4.1) {$E_{\mu^\prime}$};

		\node at (7.9,-4.8) {$= A^\prime_{v^\prime, \mu^\prime}A_{v,\mu}$};
	\end{tikzpicture}
	\end{center}
	\vspace{-2em}
\end{proof}

Recall that we have grouped sites if necessary so that $R = 1$. Next we group these sites $\{x_i\} = \Z$ in pairs along the control space $X = \R$. Let us draw a diagram indicating the support of certain source and target operator algebras. Let $\mathcal{O}_i = \mathrm{End}(\mathcal{H}_i)$. The support algebras $S$ lie in $\mathrm{End}(\mathcal{H}_{2i+1}) \otimes \mathrm{End}(H_{2i+2})$. We denoted $\mathrm{End}(\mathcal{H}_i)$ by $\OO_i$.

\begin{figure}[ht]
\centering
\begin{tikzpicture}[scale = 1.1]
	\draw[fill=black] (0,-0.1) circle (0.3ex);
	\draw[fill=black] (0,-0.1) circle (0.3ex);
	\node at (0,-0.4) {$\mathcal{O}_{2i-1}$};
	\draw  (-0.2,0.6) rectangle (1.5,0.2);
	\draw  (-0.2,1.2) rectangle (1.5,0.8);
	\node at (0.7,0.4) {$S_{2i-1}$};
	\node at (0.7,1) {$S_{2i}$};
	\draw[fill=black] (1.3,-0.1) circle (0.3ex);
	\draw[fill=black] (2.6,-0.1) circle (0.3ex);
	\draw[fill=black] (3.9,-0.1) circle (0.3ex);
	\draw[fill=black] (5.2,-0.1) circle (0.3ex);
	\draw[fill=black] (6.5,-0.1) circle (0.3ex);
	\draw  (2.4,1.2) rectangle (4.1,0.8);
	\draw  (2.4,1.8) rectangle (4.1,1.4);
	\draw  (5,2.4) rectangle (6.7,2);
	\draw  (5,1.8) rectangle (6.7,1.4);
	\node at (3.3,1) {$S_{2i+1}$};
	\node at (3.3,1.6) {$S_{2i+2}$};
	\node at (5.9,1.6) {$S_{2i+3}$};
	\node at (5.9,2.2) {$S_{2i+4}$};
	\draw (1.95, -0.1) ellipse (1 and 0.2);
	\draw (4.55, -0.1) ellipse (1 and 0.2);
	\node at (1.3,-0.5) {$\mathcal{O}_{2i}$};
	\node at (2.6,-0.5) {$\mathcal{O}_{2i+1}$};
	\node at (3.9,-0.5) {$\mathcal{O}_{2i+2}$};
	\node at (5.2,-0.5) {$\mathcal{O}_{2i+3}$};
	\node at (6.5,-0.5) {$\mathcal{O}_{2i+4}$};
\end{tikzpicture}
\caption{}
\end{figure}

Legend: $S_{2i+1} = S(\alpha(\OO_{2i} \otimes \OO_{2i+1}), \OO_{2i+1} \otimes \OO_{2i+2})$ and $S_{2i+2} = S(\alpha(\OO_{2i+2} \otimes \OO_{2i+3}), \OO_{2i+1} \otimes \OO_{2i+2})$.

\begin{lemma}["Columns Commute"]
\label{columnscommute}
	$[S_{2i+1}, S_{2i+2}] = 0$.
\end{lemma}

\begin{proof}
	Apply Lemma \ref{supportcommute} to the two subalgebras $\alpha(\OO_{2i} \otimes \OO_{2i+1})$ and $\alpha(\OO_{2i+2} \otimes \OO_{2i+3})$ of $(\OO_{2i-1} \otimes \OO_{2i}) \otimes (\OO_{2i+1} \otimes \OO_{2i+2}) \otimes (\OO_{2i+3} \otimes \OO_{2i+4})$, using the fact that $\OO_{2i} \otimes \OO_{2i+1}$ and $\OO_{2i+2} \otimes \OO_{2i+3}$ commute by disjointness.
\end{proof}

\begin{lemma}
	All distinct support algebras $S_x$, $x \in \Z$ commute.
\end{lemma}

\begin{proof}
	Those not in the same column clearly commute since their supports are disjoint.
\end{proof}

\begin{lemma}
\label{full}
	$\OO_{2i+1} \otimes \OO_{2i+2} = S_{2i+1} \otimes S_{2i+2}$. By Wedderburn's classification of simple $\ast$-algebras this implies that each of the support algebras $S_y$ is a full matrix algebra.
\end{lemma}

\begin{proof}
	Clearly $S_{2i+1} \otimes S_{2i+2} \subset \OO_{2i+1} \otimes \OO_{2i+2}$. If the inclusion is proper, then there is a $z \in \OO_{2i+1} \otimes \OO_{2i+2}$, not a multiple of $\Id$, which commutes with $S_{2i+1} \otimes S_{2i+2}$. But again by disjointness $z$ would commute with all the $S$ and therefore commute with all of $\Img(\alpha) = \A$, contradicting Wedderburn's Theorem that full matrix algebras are simple.
\end{proof}

\begin{lemma}
\label{isproduct}
	$\alpha(\OO_{2i} \otimes \OO_{2i+1}) = S_{2i} \otimes S_{2i+1}$
\end{lemma}

\begin{proof}
	From the definition of support algebra,
	\begin{equation}
	\label{isproducteq}
		\alpha(\OO_{2i} \otimes \OO_{2i+1}) \subset S_{2i} \otimes \OO_{2i+1} \otimes \OO_{2i+2} \cap \OO_{2i-1} \otimes \OO_{2i} \otimes S_{2i+1} = S_{2i} \otimes S_{2i+1},
	\end{equation}
	the last equality follows by expanding elements on a tensor basis.

	If the inclusion were proper, there would be by Wedderburn's Theorem a $z$ in the full matrix algebra $S_{2i} \otimes S_{2i+1}$, not a scalar multiple of the identity, such that $[\alpha(\OO_{2i} \otimes \OO_{2i+1}), z] = 0$.

	But by Lemma \ref{columnscommute}, $z$ commutes with $S_{2i+2}$ and $S_{2i-1}$. By disjointness $z$ also commutes with the more distant columns, the other $S$. From Eq.~(\ref{isproducteq}), $z$ commutes with $\alpha(\OO_{2j} \otimes \OO_{2j+1})$ for all $j$. So $z$ commutes with $\Img(\alpha)$. But $\alpha$ is an automorphism of $\A$ so $
	\Img(\alpha) = \A$, implying a non-scalar element $\alpha^{-1}(z)$ of Center$(\A)$, a contradiction.
\end{proof}

Now we take dimensions ("decategorify") the isomorphism. It is actually more convenient to take the $\frac{1}{2} \log \dim$. The $\log$ makes things additive (to agree with our intuition of flux) and the conventional $\frac{1}{2}$ means we are taking $\log$(Hilbert space dimension) rather than $\log$(Endomorphism algebra dimension). Let

\[
o_x = \frac{1}{2} \log(\dim \OO_x), \hspace{1em} s_x = \frac{1}{2} \log(\dim S_x)
\]

From Lemmas \ref{full} and \ref{isproduct} we have
\begin{equation}
	o_{2x} + o_{2x+1} = s_{2x} + s_{2x+1} \text{ and } o_{2x+1} + o_{2x + 2} = s_{2x+1} + s_{2x+2},
\end{equation}
which implies a well-defined \textit{flux}---independent of $x$---$f(\alpha)$:
\begin{equation}
\label{flux}
	f(\alpha) = s_{2x+1} - o_{2x+1} = -(s_{2x} - o_{2x}) = -(s_{2x+2} - o_{2x+2})
\end{equation}
for all $x$.

If we choose base 2 logarithms and all our degrees of freedom are qubits (or tensor products of qubits) then flux is integer valued, and for the generalization to higher dimensions which we soon treat it is sufficient to focus on this case to understand the arguments.
Essentially, the different primes do not interact.

However, in general flux takes values in $M = \{\log \Q\}$ logarithms (say base 2) of rationals, which we regard as a module over $\Z$, with the integers acting via multiplication on the logarithms. What we have just derived is that for a QCA $\alpha$ on $R$, $f(\alpha) \in M$. With very little extra work we have:

\begin{thm}
\label{graphthm}
	Let $G$ be any locally finite graph whose edge lengths all exceed $2R$, then to a $R$-QCA $\alpha$ on $G$ there is a well-defined flux $f(\alpha) \in H^{\text{lf}}_1(G;M)$ in the first homology of $G$ with coefficients in $M$. The correct homology is the one based on "locally finite" chains. If the degrees of freedom are all qubits we may replace $M$ by the integers $\Z$.
	If the graph is finite we may drop the "lf" superscript.
\end{thm}

\begin{proof}
	The key principle was established in Eq.~(\ref{flux}): the computation of flux is local and the same answer is obtained by doing the computation in different places. To check that flux, when locally computed on (temporarily oriented) edges of sufficient length ($>2R$), obeys the cycle condition, we must see that oriented sum around any vertext vanish as in Figure 2:

	\begin{figure}
		\centering
		\begin{tikzpicture}[scale = 1.2]
		\draw (-1,0) -- (1,0);
		\draw [fill=black] (1,0) circle (0.3ex);
		\draw (1,0) -- (3,0.6);
		\draw (1,0) -- (3, -0.6);
		\node at (0,0.4) {$f_1$};
		\node at (2,0.6) {$f_2$};
		\node at (2,-0.6) {$f_3$};
		\node at (1,-1.1) {$f_1 = f_2 + f_3$};
		\draw (-0.15, -0.15) -- (0,0) -- (-0.15, 0.15);
		\draw (1.93, 0.12) -- (2, 0.3) -- (1.85, 0.4);
		\draw (1.93, -0.12) -- (2, -0.3) -- (1.85, -0.4);
		\end{tikzpicture}
		\caption{}
	\end{figure}

	This is proven by locally projecting the degrees of freedom to an oriented line---in the case above identify the two segments labeled $f_1$ and $f_2$ and then applying (3).

	\textbf{Note on stability:} The quantity $\flux(\alpha$) is stable for two reasons, with potentially different scope of application. First, since $\flux(\alpha$) can be calculated near any point on $R$ and the same value will be obtained at different points, any deformation which can be expressed as a composition of deformations $\alpha_t$ with support smaller than $2R$ must leave $\flux(\alpha_t$) unchanged. Also if there is a QCA $\overline{\alpha}$ which agrees with $\alpha_0$ in one interval in $R$ and $\alpha_1$ in another disjoint interval, then $\flux(\alpha_0$) = $\flux(\alpha_1$); QCAs of different fluxes cannot be blended together. Local compatibility is a very general principle and one could imagine applying it beyond the case of finite dimensional degrees of freedom treated above, e.g. to type II factors.

	An even sharper source of stability, but one dependent on the integral dimensions of the local Hilbert spaces, is the total disconnectedness of the set $M = \log(\Q)$ in which possible flux values lie. If $\alpha_t$ changes continuously the values of $\flux(\alpha_t$) must stay constant.
\end{proof}

\subsection{Coverings, Immersions, and Homology}
\label{multiplicativesection}
In this subsection we employ topological tricks to extract as much information as possible from the 1D definition (GNVW) of index or "flux." The most fundamental trick is \textit{dimensional reduction} given a control space $X$ and a map $g: X \ra G$, to a 1D space, a graph $G$, if all point preimages of $g$ are compact we may replace degrees of freedom in $X$ located at $\{x_i\}$ with their images $\{g(x_i)\}$ located in $G$. (It is immaterial if $g(x_i) = g(x_j), i \neq j$.) Control in $X$ will imply control in $G$. Then in $G$ we may apply the 1D theory recalled in subsection \ref{GNVWreview}. A space such as a torus can be reduced to a 1D circle in many ways, and this will be important for us.

Another topological idea is that of \textit{germ}. The word is from function theory, if $f: \R \ra \R$, $f(0) = 0$, is any real function fixing the origin, its germ is the equivalence class of $f|_\textrm{U}$, U any open set containing 0, where equivalent means equal on $\textrm{U} \cap \textrm{U}^\pr$. Intuatively germ means a tiny, but important, piece of something larger.

In that spirit, suppose we have QCA on a $d$-manifold $Y$ containing a $(d-1)$-submanifold $W$ with trivial normal bundle neighborhood $W \times (-1, 1) \subset Y$ and the interaction radius $R$ of a QCA $\alpha$ on $Y$ is small w.r.t. the thickness of the normal bundle, say, $R << 0.1$, then $\alpha|_{W \times (-1/2, 1/2)}$, the \textit{germ} of $\alpha$ near $W$, is not strictly speaking always a QCA because the edges may be ragged: d.o.f. hopping on and off under $\alpha$. Instead of an automorphism of algebras we have embedding ($\ast$-isometries into):

\[
\A_{W \times (-0.5, 0.5)} \xrightarrow{\alpha|} \A_{W \times (-0.6, 0.6)} \xrightarrow{\alpha^{-1}|} \A_{W \times (-0.7, 0.7)}
\]

\noindent
so that the composition is $\Id_{\A_{W \times (-0.5, 0.5)}}$ composed with the inclusion Inc: $\A_{W \times (-0.5, 0.5)}$ $\ra$ $\A_{W \times (-0.7, 0.7)}$.

\begin{note}
If we dimensionally reduce $W \times (-0.5, 0.5)$ to $(-0.5, 0.5)$ the discussion of subsection \ref{GNVWreview} produces a well-defined GNVW Index (flux) near 0 despite the fact that the system has edges. Thus we may speak of the flux of $\alpha$ crossing $W$ when $R$ is sufficiently small.
\end{note}

We will use the fact\cite{Hastings2013} that QCAs can be pulled back under covering space projections and, more generally, immersions. The latter case is conceptually similar to the germ concept in that there is also a ragged edge so the QCA is not pulled back precisely to an automorphism but rather a germ on a somewhat reduced region as in the previous discussion.

As the methods will be different, we divide this subection into subsubsections \ref{2dsubsection},\ref{3dsubsection},\ref{4dsubsection},
dealing respectively with QCAs of manifolds of dimensions 2, 3, and 4, respectively. However, in each subsection the result is the same homological multiplicativity formula that would arise if the flux $f$ of quantum information were merely a classical flow. Since we know the situation to be richer than that\cite{FHH} these results supply a "lower bound" on the classical behavior.

\begin{thm}
\label{cover}
	Let $\alpha$ on $W \times [-\frac{1}{2}, \frac{1}{2}]$ be the germ of a QCA on a closed orientable hypersurface $W$ of dimension $d-1$ for $d\leq 4$. Let $\widetilde{W} \xrightarrow{\pi} W$ be an $n$-sheeted covering space, $n$ a natural number, and $\widetilde{\alpha}$ on $\widetilde{W} \times [-\frac{1}{2}, \frac{1}{2}]$ be the pulled back QCA, $\widetilde{\alpha} = \pi^\ast(\alpha)$. Then the flux across $W$ and $\widetilde{W}$ satisfy $\flux(\widetilde{\alpha}$) = $n \flux(\alpha$). Said homologically $\flux(\alpha$) $\in H_1^{\text{lf}}(W \times [-\frac{1}{2}, \frac{1}{2}]; \Z)$ and $\flux(\widetilde{\alpha}$) $\in H_1^{\text{lf}}(\widetilde{W} \times [-\frac{1}{2}, \frac{1}{2}]; \Z)$ satisfy:

	\[
	\pi_\ast(\flux(\widetilde{\alpha})) = \mathrm{deg}(\pi)(\flux(\alpha))
	\]
\end{thm}
The three cases $d=2,3,4$ are theorems \ref{2dcover}, \ref{3dcover}, \ref{4dcover}.

\subsubsection{Proof of Theorem \ref{cover} for d = 2}
\label{2dsubsection}
Following standard 2D notation, we write $\gamma$ for a simple closed curve (scc) on $W$, $C$ for the collar, $\gamma \times [-0.5, 0.5]$, and $\Sigma$ for an ambient surface containing the germ.

So we consider a QCA with some small range $R$ on a surface $\Sigma$, with scc $\gamma \subset \Sigma$ and $C$ a cylindrical tubular neighborhood around $\gamma$ with injectivity radius much larger than R.

\begin{figure}[ht]
	\centering
	\begin{tikzpicture}[scale = 0.6]
		\draw (8.65,0) arc (90:-90:0.4 and 1);
		\draw [dashed] (8.65,0) arc (90:270:0.4 and 1);
		\draw (4.6, -0.8) arc (-180:0:1.3 and 0.6);
		\draw (4.9, -1.2) arc (180:0:1 and 0.5);
		\draw (1.6, -0.7) arc (180:10:8 and 6);
		\draw (3.24, -0.32) arc(185:20:6.4 and 4);
		\draw (3.24, -0.32) to [out=60, in = 160] (8.65,0);
		\draw (1.6, -0.7) to [out = -90, in = -160] (8.65, -2);

		\draw (15.63,1.4) to [out = -60, in = 10] (13.5, -0.3);
		\draw (17.48, 0.35) to [out = -75, in = 0] (13.5, -2.3);
		\draw [dashed] (13.5, -0.3) arc(90:270:0.3 and 1);
		\draw (13.5, -0.3) arc(90:-90:0.3 and 1);
		\draw (8.65, 0) to [out = -20, in = 180] (11.5,-0.7) to [out = 0, in = 200] (13.5, -0.3);
		\draw (8.65, -2) to [out = 20, in = 180] (10.8,-1.8) to [out =0, in = 180] (13.5, -2.3);

		\draw (11.5, -0.7) arc (90:-90:0.3 and 0.58);
		\draw [dashed] (11.5, -0.7) arc (90:270:0.4 and 0.58);
		\draw (14.2, -0.8) arc (180:360:1 and 0.5);
		\draw (14.41, -1.1) arc (180:0:0.8 and 0.3);
		\draw [rotate=-45] (10.55, 11.5) arc (180:-0:1.15 and 0.45);
		\draw [rotate=-45, dashed] (10.55, 11.5) arc (-180:0:1.15 and 0.45);

		\node at (5.7,0.1) {$\Sigma$};
		\node at (12.2,-1.2) {$\gamma$};
		\draw [decorate,decoration={brace,amplitude=10pt}](13.5,-2.6) -- (8.7,-2.6)node [black,midway, yshift = -17pt] {$C$};
		\draw [decorate,decoration={brace,amplitude=20pt}](17,-3.6) -- (8.7,-3.6)node [black,midway, yshift = -27pt] {$C^\prime$};
		\draw (18.3,-2.5) -- (18.3,-2.9);
		\draw (18.3,-2.7) -- (18.8,-2.7);
		\draw(18.8,-2.5) -- (18.8,-2.9);
		\node at (18.5,-3.1) {$R$};
	\end{tikzpicture}
	\caption{$\alpha$ acts on operator $R$, almost preserving locality}
\end{figure}

\begin{lemma}
	A germ of $\alpha$ near $\gamma$ may be imported into a twice-punctured torus $T^{--}$ with $\gamma$ parallel to one of the punctures.
\end{lemma}

Note that for $\Sigma$ as drawn in Fig. 3, we may just restrict $\alpha$ to (most of) $C^\pr \cong T^{--}$. The work comes when, for example, $\Sigma$ is itself a torus where there is no "extra" genus.

But the "work" is quite easy given the Kirby torus trick, introduced in this context in \cite{Hastings2013}. One merely immerses $T^{--}$ in $C$, sending the (say) left puncture to a parallel copy of $\gamma$ and then pull back the restriction of $\alpha$ to $T^{--}$. Figure 4 shows explicitly how to do this.

\begin{figure}[ht]
	\centering
	\begin{tikzpicture}[scale = 0.7]
		\draw (0,0) ellipse (0.6 and 2);
		\draw (2.5,2) arc(90:-90:0.6 and 2);
		\draw [dashed] (2.5,2) arc(90:270:0.6 and 2);
		\draw (0,2) -- (11,2);
		\draw (0,-2) -- (11,-2);
		\draw (6.5, 2) arc(90:-90:0.6 and 2);
		\draw [dashed] (6.5, 2) arc(90:270:0.6 and 2);
		\node at (3.7,0) {$\gamma$};
		\draw (11,2) arc(90:-90:0.6 and 2);
		\draw [dashed] (11, 2) arc(90:270:0.6 and 2);

		\draw (7.1, 0.4) arc (100:-102.5:1.35 and 0.6);
		\draw (7.05, 0.8) arc (100:-102.5:1.75 and 1);

		\draw (7,1) to [out = 15, in = 90] (9.35,-0.1) to [out = -90, in = -30] (8.15,-1.07);
		\draw (7.84, -0.75) to [out = 155, in = -45] (7.1, -0.23);
		\draw (6.9,1.5) to [out = 20, in = 90] (9.8,-0.1) to [out = -90, in = 45] (9.3,-1.2) to [out = -135, in = -35] (7.6,-1.18);
		\draw (7.3,-0.8) to [out = 155, in = -40] (7.08, -0.67);

		\draw [decorate, decoration={brace,amplitude=10pt}] (2.5, 2.2) -- (6.5, 2.2) node [black, midway, yshift = 18pt] {$C$};
		\draw [decorate, decoration={brace,amplitude=10pt}] (6.5, -2.2) -- (0, -2.2) node [black, midway, yshift = -18pt] {Immersed copy of $T^{--}$ (bands give the genus)};
		\draw [->] (9.3, 2.9) -- (8.9,1.4);
		\node at (9.3,3.9) {bands should be much};
		\node at (9.3,3.2) {thicker than $R$};
	\end{tikzpicture}
	\caption{}
\end{figure}

\begin{note}
Locality of the calculation of flux in a one-dimensional reduction shows that the flux of "pulled back $\alpha$" through the two ends is the same as $f(\alpha)(\gamma)$.
\end{note}

\begin{lemma}
\label{3fold}
	$T^{--}$ admits an irregular 3-fold cover schematically drawn in Figure 5. The preimage of one puncture is a circle; the preimage of the other puncture is three circles.
\end{lemma}

\begin{figure}[ht]
	\centering
	\hspace{2em}
	\begin{tikzpicture}[scale = 0.75]
		\draw (0,0) ellipse (1.25 and 2);
		\draw (0, 2) to [out = 0, in = 180] (3,1.4) to [out = 0, in = 180] (9.1,3.2);
		\draw (0, -2) to [out = 0, in = 180] (3,-1.4) to [out = 0, in = 180] (9.1,-3.2);
		\draw (9.1,2.5) ellipse (0.6 and 0.7);
		\draw (9.1,-2.5) ellipse (0.6 and 0.7);
		\draw (9.1,0) ellipse (0.6 and 0.7);
		\draw (9.1, 1.8) arc(90:270:1.7 and 0.55);
		\draw (9.1,-0.7) arc(90:270:1.7 and 0.55);
		\draw (2.3, -0.1) arc (180:0:0.9 and 0.35);
		\draw (2.45, 0.08) arc (180:360:0.76 and 0.3);
		\draw (4.8, -0.1) arc (180:0:0.9 and 0.35);
		\draw (4.95, 0.08) arc (180:360:0.76 and 0.3);

		\draw [->] (4.3,-2.2) -- (4.3, -3.2);

		\draw (0,-5.5) ellipse (0.5 and 1);
		\draw (9,-5.5) ellipse (0.5 and 1);
		\draw (0, -4.5) to [out = 0, in = 180] (4.5,-3.8) to [out = 0, in = 180] (9, -4.5);
		\draw (0, -6.5) to [out = 0, in = 180] (4.5,-7.2) to [out = 0, in = 180] (9, -6.5);
		\draw (7.6, -4.36) arc (90:270:0.5 and 1.14);
		\draw [dashed] (7.6, -4.36) arc (90:-90:0.5 and 1.14);
		\draw (1.4, -4.36) arc(90:-90:0.5 and 1.14);
		\draw [dashed] (1.4, -4.36) arc (90:270:0.5 and 1.14);

		\draw (1.82,-6.1) to [out = 0, in = 180] (4.7, -6.7) to [out = 0, in = 200] (5.2,-6.3) to [out = -20, in = 180] (6.2,-6.5) to [out = 0, in = 200] (7.2, -6.2);
		\draw (4.5, -5.5) ellipse (1.5 and 0.9);
		\draw (3.6,-5.65) arc (180:0:0.9 and 0.35);
		\draw (3.73,-5.47) arc (180:360:0.76 and 0.3);
		\draw (4.8, -5.75) arc (90:-90:0.4 and 0.72);
		\draw [dashed] (4.8, -5.75) arc (90:270:0.4 and 0.72);

		\node at (5.3,-5.9) {$m$};
		\node at (5.6,-4.6) {$l$};
		\node at (2.1,-4.8) {$x$};
		\node at (6.9,-4.9) {$y$};
		\draw (1.88, -5.13) -- (2.05, -5.25);
		\draw (1.88, -5.13) -- (1.7, -5.25);
		\draw (4.5, -4.6) -- (4.65, -4.45);
		\draw (4.5, -4.6) -- (4.65, -4.75);
		\draw (7.1, -5.4) -- (6.95, -5.55);
		\draw (7.1, -5.4) -- (7.25, -5.55);
		\node at (10.5,-5.5) {$C^\prime$};
	\end{tikzpicture}
	\caption{}
\end{figure}

\begin{proof}
	$\pi_1(T^{--}) \cong \mathrm{Free}(l,m,x)$; define a homomorphism $\te$ from this group to the permutation group $S(3)$ by:

	\vspace{0.5em}
	\noindent
	$\te(l) = (1,3,2)$\\
	$\te(m) = (2,3)$\\
	$\te(x) = (1,2,3)$

	Now the element $y$ corresponding to the right puncture is sent to the identity $e \in S(3)$:
	\begin{equation}
		y = [m,l]x, \hspace{1em} \te(y) = \te[m,l]\te(x) = (2,3)(1,3,2)(2,3)(1,2,3)(1,2,3) = e
	\end{equation}

	Let the irregular cover correspond to the pulled back group $H$:

	\begin{figure}[ht]
		\centering
		\hspace{5em}
		\begin{tikzpicture}
			\node at (0,0) {Free($l,m,x$)};
			\node at (0,-1.6) {$H$};
			\draw [->] (1.2, 0) -- (2.7, 0);
			\node at (3.3,0) {$S(3)$};
			\node at (3.3,-1.6) {$Z_2$};
			\draw [->] (0.5, -1.6) -- (2.7, -1.6);
			\node at (4.6,-1.6) {$= \{e, (1,2)\}$};
			\node at (5.1,-2.2) {(or any other $Z_2$ subgroup)};
			\node at (1.9,0.25) {$\theta$};
			\node at (0,-0.8) {\huge $\hookuparrow$};
			\node at (3.3,-0.8) {\huge $\hookuparrow$};
		\end{tikzpicture}
	\end{figure}

	Since $\te(x)$, $\te(x^2) \notin Z_2$, $x$ and $x^2 \notin H$, explaining why the left puncture is covered by a single circle, whereas $y \in \ker(\te)$ and in particular $y \in H$, so the right puncture is covered by three circles.
\end{proof}

There is a straightforward generalization to $n$-fold irregular covers replacing $S(3)$ with the dihedral group $D(n)$ of order $2n$.

\[
D(n) = \{r,t | r^2 = e, t^n = e, rtr = t^{-1}\}
\]

Define $\te_n:\pi_1(T^{--}) \ra D(n)$ by:

\vspace{0.5em}
\noindent
$\te_n(l) = t^{-1}$\\
$\te_n(m) = r$\\
$\te_n(x) = t^{-2}$

$\te_n(y) = \te_n([m,l]x) = rt^{-1}rtt^{-2} = e$.

Now construct the irregular $n$-fold cover corresponding to the pulled back subgroup $H_n$:

\begin{figure}[ht]
	\centering
	\begin{tikzpicture}
		\node at (0,0) {Free($l,m,x$)};
		\node at (0,-1.6) {$H_n$};
		\draw [->] (1.2, 0) -- (2.7, 0);
		\node at (3.3,0) {$D(n)$};
		\node at (3.3,-1.6) {$Z_2$};
		\draw [->] (0.5, -1.6) -- (2.7, -1.6);
		\node at (4.3,-1.6) {$= \{e, r\}$};
		\node at (1.9,0.25) {$\theta$};
		\node at (0,-0.8) {\huge $\hookuparrow$};
		\node at (3.3,-0.8) {\huge $\hookuparrow$};
	\end{tikzpicture}
\end{figure}

There are two cases. If $n$ is odd $\te(x)$ has order $n$, if $n$ is even $\te(x)$ has order $\frac{n}{2}$. Also the normalizer $N(Z_2)$ of $Z_2$ in $D(n)$ is $Z_2$ for $n$ odd and the Klein group $\{e,r,t^{n/2},rt^{n/2} \cong Z_2 \oplus Z_2\}$ for $n$ even. The \textit{covering group} of this irregular cover is, of course,
\begin{equation}
	N(H_n) / H_n \cong N(Z_2) / Z_2 \cong \left\{
	\begin{array}{lr}
		\{e\}, & n \text{ odd} \\
		\{e, t^{n/2}\} \cong Z_2, & n \text{ even}
	\end{array}\right.
\end{equation}

Consequently, for $n$ odd the preimage of $x$ is one circle of "length" $n$, and for $n$ even is two circles of length $\frac{n}{2}$ permuted by the covering group. Since $\te(y) = e$, $y \in H_n$, for all $n$, the preimage of $y$ consists of $n$ circles. This establishes:

\begin{lemma}
	$T^{--}$ admits an irregular degree $n$ covering space, with preimage $\gamma_n$ circles and preimage $x$ one (two) circle(s) $n$ odd (even).
\end{lemma}

\begin{proof}
	The total space for $n$ odd is a genus $\frac{n+1}{2}$ surface with one puncture covering $x$ and $n$ punctures covering $y$. For $n$ even the total space is a genus $\frac{n}{2}$ surface with two punctures covering $x$ and $n$ punctures covering $y$ The genus is computed using multiplicity of the Euler characteristic under cover.
\end{proof}

\begin{customthm}{2.1A}
\label{2dcover}
	Given a germ of a 2D QCA $\alpha$ on a surface containing a scc $\gamma$, denote the GNVW flux of $\alpha$ across $\gamma$ by $f_\gamma$. Let $C$ be a cylindrical neighborhood of $\gamma$ (large w.r.t. the interaction radius $R$ of $\alpha$) and $C_n$, $\gamma_n$, and $\alpha_n$ be the n-fold cylic covers. In the cover the GNVW flux $f_{\alpha_n}$ satisfies:
	\[
	f_{\alpha_n} = nf_\alpha
	\]
\end{customthm}

\begin{figure}[ht]
	\centering
	\begin{tikzpicture}[scale = 0.6]
		\draw (0,0) ellipse (0.8 and 2);
		\draw (0, 2) to [out = 0, in = 180] (3,1.8) to [out = 0, in = 180] (9.1,3.2);
		\draw (0, -2) to [out = 0, in = 180] (3,-1.8) to [out = 0, in = 180] (9.1,-3.2);
		\draw (9.1,2.6) ellipse (0.5 and 0.6);
		\draw (9.1,-2.6) ellipse (0.5 and 0.6);
		\draw (9.1,0.5) ellipse (0.5 and 0.6);
		\draw (9.1,2) arc(90:270:1.7 and 0.45);
		\draw (9.1,-0.1) arc(90:180:1.7 and 0.45);
		\draw (7.4,-1.55) arc(180:270:1.7 and 0.45);

		\draw (1.5,1.1) arc (180:0:0.9 and 0.35);
		\draw (1.65,1.3) arc (180:360:0.76 and 0.3);
		\draw (2.1,0.3) arc (180:0:0.9 and 0.35);
		\draw (2.25,0.5) arc (180:360:0.76 and 0.3);
		\draw (3.9,-1.1) arc (180:0:0.9 and 0.35);
		\draw (4.03,-0.9) arc (180:360:0.76 and 0.3);

		\draw [->] (4.3,-2.2) -- (4.3, -3.2);
		\draw (0,-5.5) ellipse (0.5 and 1);
		\draw (9,-5.5) ellipse (0.5 and 1);
		\draw (0, -4.5) to [out = 0, in = 180] (4.5,-3.8) to [out = 0, in = 180] (9, -4.5);
		\draw (0, -6.5) to [out = 0, in = 180] (4.5,-7.2) to [out = 0, in = 180] (9, -6.5);
		\draw (3.6,-5.65) arc (180:0:0.9 and 0.35);
		\draw (3.73,-5.47) arc (180:360:0.76 and 0.3);

		\node at (3.7,-0.3) {\Large{$\ddots$}};
		\node at (1,-1.1) {\small{$\gamma_n$}};
		\node at (5.5,0.6) {\small{genus = $\frac{n+1}{2}$}};
		\node at (4.4,2.8) {\large{$\widetilde{C^\prime}$}};
		\node at (4.5,-7.7) {\large{$C^\prime$}};
		\node at (9.1,-0.9) {\Large{$\vdots$}};
		\node at (2.2,-2.8) {\small{degree $n$, odd}};

		\draw (12,-5.5) ellipse (0.5 and 1);
		\draw (21,-5.5) ellipse (0.5 and 1);
		\draw (12, -4.5) to [out = 0, in = 180] (16.5,-3.8) to [out = 0, in = 180] (21, -4.5);
		\draw (12, -6.5) to [out = 0, in = 180] (16.5,-7.2) to [out = 0, in = 180] (21, -6.5);
		\node at (16.5,-7.7) {\large{$C^\prime$}};
		\draw (15.6,-5.65) arc (180:0:0.9 and 0.35);
		\draw (15.73,-5.47) arc (180:360:0.76 and 0.3);
		\draw [->] (16.3,-2.2) -- (16.3, -3.2);
		\node at (14.2,-2.8) {\small{degree $n$, even}};
		\node at (16.4,2.8) {\large{$\widetilde{C^\prime}$}};

		\draw (12, 2) to [out = 0, in = 180] (15,1.9) to [out = 0, in = 180] (21.1,3.2);
		\draw (12, -2) to [out = 0, in = 180] (15,-1.9) to [out = 0, in = 180] (21.1,-3.2);
		\draw (12,1.3) ellipse (0.5 and 0.7);
		\draw (12,-1.3) ellipse (0.5 and 0.7);
		\draw (12, 0.6) arc (90:-90:1.2 and 0.6);
		\node at (13,1.3) {\small{$\gamma_{n/2}$}};
		\node at (13,-1.3) {\small{$\gamma_{n/2}$}};

		\draw (13.5,1.1) arc (180:0:0.9 and 0.35);
		\draw (13.65,1.3) arc (180:360:0.76 and 0.3);
		\draw (14.1,0.3) arc (180:0:0.9 and 0.35);
		\draw (14.25,0.5) arc (180:360:0.76 and 0.3);
		\draw (15.9,-1.1) arc (180:0:0.9 and 0.35);
		\draw (16.03,-0.9) arc (180:360:0.76 and 0.3);
		\node at (15.7,-0.3) {\Large{$\ddots$}};
		\node at (17.7,0.8) {\small{genus = $\frac{n}{2}$}};

		\draw (21.1,2.6) ellipse (0.5 and 0.6);
		\draw (21.1,-2.6) ellipse (0.5 and 0.6);
		\draw (21.1,0.5) ellipse (0.5 and 0.6);
		\draw (21.1,2) arc(90:270:1.7 and 0.45);
		\draw (21.1,-0.1) arc(90:180:1.7 and 0.45);
		\draw (19.4,-1.55) arc(180:270:1.7 and 0.45);
		\node at (21.1,-0.9) {\Large{$\vdots$}};
	\end{tikzpicture}
	\caption{}
\end{figure}

\begin{proof}
	The proof is to apply Lemmas \ref{3fold}, \ref{welldefined}, and dimensional reduction (i.e. aggregation of sites). We do not decorate $\alpha$ as we modify it. The local nature of computation of flux (not local in the transverse direction but in the linear direction after aggregation of sites) implies:

	\[
	f_{\gamma,C}(\alpha) = f_{\gamma,C^\pr}(\alpha)
	\]

	\noindent
	and

	\[
		nf_\gamma(\alpha) = nf_{\gamma,C^\pr}(\alpha) = f^{\text{right}}_{(y)^{-1},\widetilde{C^\pr}} (\alpha) \notate[X]{{}={}}{2}{\text{well-definedness of GNVW index in 1D}} f^{\text{left}}_{(x)^{-1},\widetilde{C^\pr}} (\alpha) = \left\{
		\begin{array}{lr}
			f_{\gamma_n}(\alpha), & n \text{ odd} \\
			2f_{\gamma_{n/2}}(\alpha), & n \text{ even}
		\end{array}\right.
	\]
\end{proof}

Theorem \ref{2dcover} follows directly for $n$ odd. For the general integer $k$, odd or even divide the equation by 2 to obtain the theorem:

\[
2k f_\gamma(\alpha) = 2f_{\gamma_{n/2}}(\alpha)
\]
\qed

\subsubsection{Proof of Theorem \ref{cover} for d = 3}
\label{3dsubsection}
In the previous subsection \ref{2dcover} we gave a topological proof of Theorem \ref{cover} that for a QCA on a surface the GNVW-flux across any circle is multiplicative under unwrapping the circle in a finite cover. We now give related, but conceptually distinct, topological proof that GNVW flux of a QCA across a hypersurface is again multiplicative in finite covers when the QCA is acting with sufficient locality near the germ of hypersurface $Y$ in a manifold of dimension 3 (the hypersurface is a closed 2D surface). We consider here only oriented manifolds and oriented hypersurfaces. This leaves out an interesting case we would like to understand $RP^2 \times (-0.5, 0.5)$.

We first consider a germ of a QCA of small interaction radius $R$ on $Y \times (-0.5, 0.5)$, $Y$ a closed oriented surface.

\begin{lemma}
\label{welldefined}
	Let $Y$ be a closed oriented $(d-1)$-manifold. Let $Z \overset{f_t}{\longhookrightarrow} Y \times \mathbb{R}$ be a regular homotopy of immersion of a closed oriented $(d-1)$-manifold in $Y \times (-0.04, 0.04)$, where $R << 0.1$, $t \in [0,1]$. For any QCA $\alpha$ with range $R$ sufficiently small (w.r.t. $f_t$) the GNVW flux of $\alpha$ across $f_t(Z)$ is well-defined and independent of $t$.
\end{lemma}

\begin{proof}
	The orientation trivializes the normal bundle to $Z$ as $Z \times \mathbb{R}$. The pulling back the QCA $\alpha$ to $Z \times \mathbb{R}$ we obtain a new QCA on $Z \times \mathbb{R}$ and may determine its flux. The regular homotopy $f_t$ may be partitioned so that for any $t$ the homotopy is supported in a small ball whose width in the $R$ direction is $<< 0.1$. This means that at any time $t$ we may locate a representative hypersurface $Z_t$ across which the GNVW flux is not changing at time $t$, because $Z_t \cap (\text{support } f_t) = \varnothing$. This implies that the flux does not vary as a function of $t$.
\end{proof}

It is instructive to see an example of two regularly homotopic immersions (in this case of a circle) which do not co-bound an immersion of a manifold of one higher dimension (in this case a surface). Figure 7 shows two circles $\gamma_1$ and $\gamma_2$, one embedded and one immersed in an annulus. They are regularly homotopic (the dotted line suggests an intermediate stage of the regular homotopy) but do not cobound any immersion of a surface.

\begin{figure}[ht]
	\centering
	\hspace{2em}
	\begin{tikzpicture} [scale = 1.1]
		\draw  (0,0) circle (0.6);
		\draw  (0,0) circle (1.1);
		\draw  (0,0) circle (2.7);
		\draw (-1.9,0) arc (180:340:1.9);
		\draw (-1.9, 0) to [out = 90, in = 220] (-1.2,1.3) to [out = 40, in = 270] (0,2.1);
		\draw (0,2.1) arc (0:220:0.3 and 0.3);
		\draw (-0.54, 1.92) to [out = -58, in = 180] (1.1,1.5);
		\draw (1.1,1.5) arc (90:-180:0.3 and 0.35);
		\draw (0.8,1.15) to [out = 90, in = 180] (1.4, 1.7) to [out = 0, in = 90] (2.1,0.9) to [out = -90, in = 70] (1.78,-0.66);
		\node at (0.6,-0.7) {$\gamma_1$};
		\node at (-1.4,1.5) {$\gamma_2$};
		\draw [dashed] (1.15, 1.65) to [out = 145, in = 0] (-0.3, 2.5) to [out = 180, in = 60] (-0.8, 2.2) to [out = 240, in = 165] (0.8,1.1);
		\draw [decorate,decoration={brace,amplitude=10pt}](3.2,2.7) -- (3.2,-2.7) node [black,midway, xshift = 30pt] {annulus};
	\end{tikzpicture}
	\caption{}
\end{figure}

The reason this is possible is that "regular homotopy" implies a 1-parameter family of maps whose differential has rank one, whereas an immersed surface implies a differential of rank two, a stronger condition.

To prove there is no immersed surface $\Sigma$ connecting $\gamma_1$ and $\gamma_2$ use the maximal principal to show that image($\Sigma$) would be contained within the outer contour of $\gamma_2$. Then inspecting image($\Sigma$) near the leftmost kink yields a contradiction.

\begin{customthm}{2.1B}
\label{3dcover}
	Let $Y$ be a closed surface and $c: T \ra Y$ an $n$-sheeted covering. Let $\alpha$ be a QCA on $Y \times \mathbb{R}$, or germ of a QCA on $Y \times (-0.5, 0.5)$, with GNVW-flux $f(\alpha) = a$. Then the pull back of $\alpha$, $\alpha_n$, on $T \times R$ has GNVW-flux $F(\alpha_n) = na$.
\end{customthm}

\begin{proof}
	Consider the map $C: T \coprod (-Y_1 \coprod \cdots \coprod -Y_n) \coloneqq T^+ \ra Y \times \mathbb{R}$ defined by $C(T) = c(T) \times 0$, $C(-Y_i) = c(Y_i) \times i$. The $-$ sign refers to reversed orientation. Image($C$) is a null homologous cycle so it bounds an integral 3-chain in $Y \times \mathbb{R}$. A three-dimensional integral chain bounding a surface such as $T^+$ can always be assumed to be a map $g$ from an oriented 3-manifold with boundary equal $T^+$. We would like, ideally, to approximate $g$, relative to the boundary $T^+$, by an immersion. Unfortunately immersion theory is only operative in the range where the relative handles to be immersed, the handles of $(X, T^+)$, have index less than the ambient dimension, in this case 3, into which they are to be immersed. Beyond this dimension requirement, the only other hypothesis of the Smale-Hirsch theory is that the map to be approximated by an immersion must be covered by a tangent bundle injection. Since all oriented 3-manifolds are parallelizable, the latter is for free, and the dimension hypothesis is achieved by puncturing $X$, removing the interior of a 3-ball, to obtain $X^-$ with $\de X^- = \de X \cup S^2 = T^+ \cup S^2$, $S^2$ a 2-sphere, now the restricted map $g^-: X^- \ra Y \times \mathbb{R}$ may be approximated by an immersion $\widetilde{g}: X^- \ra Y \times \mathbb{R}$.

	Pull $\alpha$ back via $\widetilde{g}$ and dimensionally reduce $X^-$ to a tree joining $n+2$ vertices, one for each boundary component: $T$, $Y_1$, ..., $Y_n$, $S^2$, to a base point. Apply additivity of GNVW flux to trees and high girth graphs (Theorem \ref{graphthm}) to obtain:

	\[
	f_\alpha(\alpha_n) = nf(\alpha) + f_\alpha(\widetilde{g}(S^2))
	\]

	This is the desired formula except we must show that the correction term $f_\alpha(S^2)$ representing GNVW flux escaping through the immersed 2-sphere $\widetilde{g}(S^2)$ is zero. For this we use Lemma \ref{welldefined}. Since $\widetilde{g} \simeq g^-$ are homotopic, $\widetilde{g}|_{S^2}$ is homotopic to a small round 2-sphere $h$ imbedded in $Y \times \mathbb{R}$. Because $\pi_2(SO(3)) \cong 0$ the choice of tangential data normally associated with an immersion is unique. Consequently $\widetilde{g}|_{S^2}$ is regularly homotopic to the small round 2-sphere $h$. Again by dimensional reduction, there is a "Gauss Law" $f_\alpha(h(S^2)) = 0$ and by Lemma \ref{welldefined}
	\begin{equation}
		f_\alpha(\widetilde{g}(S^2)) = 0,
	\end{equation}
	completing the proof.
\end{proof}

Note that as in Figure 7 one should not expect that there is actually an immersion of $X^-$ into $Y \times \mathbb{R}$ which embeds $\de(X^-)$, but using Lemma \ref{welldefined}, we see that this is not required to complete the argument.

\begin{note}
The proof of Theorem \ref{3dcover} cannot be adapted to prove Theorem \ref{2dcover}. The reason is that $\pi_1(SO(2)) \cong \Z_2 \neq 0$. If we followed the logic of Theorem \ref{3dcover} to construct an immersion of say a punctured "pair of pants" joining one degree 2 circle to two degree 1 circles we would find the framing obstruction around the puncture to be = 2, so the puncture circle would not be regularly homotopic to a trivial embedded circle having framing obstruction $= 1$. The model for this is to take the 2-fold branched cover $b$ of a cylinder $C$ by a pair of points $P$ and notice that if the branch point $p \in P$ is deleted to form $P^-$ then $b|_{P^-}: P^- \ra C$ maps the puncture to a degree 2 loop. See Figure 8.
\end{note}

\begin{figure}[ht]
	\centering
	\hspace{7em}
	\begin{tikzpicture}[scale = 0.7]
		\draw (0,0) ellipse (0.6 and 2);
		\draw (6,1.3) ellipse (0.4 and 0.7);
		\draw (6,-1.3) ellipse (0.4 and 0.7);
		\draw (0, 2) -- (6, 2);
		\draw (0, -2) -- (6, -2);

		\draw [->] (2.8,-2.3) -- (2.8,-3.3);

		\draw (0,-4.8) ellipse (0.5 and 1);
		\draw (6,-4.8) ellipse (0.5 and 1);
		\node at (7.5,-4.8) {$C$};

		\draw (4,0) arc (180:90:2 and 0.6);
		\draw (4,0) arc (180:270:2 and 0.6);
		\node at (7.5,0) {$P$};
		\draw (4,0) [fill=black] circle (0.3ex);
		\node at (3.2,0) {$P$};
		\draw (3.5, 0) arc (180:43:0.5 and 0.5);
		\draw (3.5, 0) arc (-180:-43:0.5 and 0.5);
		\draw [dashed] (3.7, 0) arc (180:55:0.4 and 0.38);
		\draw [dashed] (3.7, 0) arc (-180:-55:0.4 and 0.38);
		\node at (3.6,-2.8) {$b/P^-$};

		\draw (0,-3.8) -- (6,-3.8);
		\draw (0, -5.8) -- (6, -5.8);
		\draw (4, -4.4) [out = 145, in = 0]  to (3.8, -4.35) to [out = 180, in = 90] (3.5, -4.8) to [out = -90, in = 180] (4, -5.3) to [out = 0, in = -90] (4.5, -4.8) to [out = 90, in = 0] (4.2, -4.35) to [out = 180, in = 35] (4, -4.4);
		\draw (4, -4.4) [out = -35, in = 90] to (4.25, -4.8) to [out = -90, in = 0] (4, -5.1);
		\draw (4, -4.4) [out = 215, in = 90] to (3.75, -4.8) to [out = -90, in = 180] (4, -5.1);
		\draw [->] (5.5, -3) -- (4.4, -4.3);
		\node at (9,-2.7) {image of loop around puncture $P$};

		\draw (-1.6, 1) arc (90:-90:0.5 and 1);
		\node at (-1.9,0) {$180^\circ$};
		\draw [->] (-1.59, -1) -- (-1.6, -1);
		\draw [fill=black] (4,-4.8) circle (0.3ex);
	\end{tikzpicture}
	\caption{}
\end{figure}

\subsubsection{Proof of Theorem \ref{cover} for d = 4}
\label{4dsubsection}
\begin{customthm}{2.1C}
\label{4dcover}
	Let $Y$ be a closed oriented 3-manifold and $c:T \ra Y$ an n-sheeted covering. Let $\alpha$ be a QCA on $Y \times \mathbb{R}$, or a germ of a QCA on $Y \times (-0.5, 0.5)$, with GNVW-flux $f(\alpha) = a$. Then the pull back of $\alpha$, $\alpha_n$, on $T \times R$ has GNVW-flux $f(\alpha_n) = na$.
\end{customthm}

\begin{proof}
	Just as before construct a null homologous (integral) cycle $C: T \coprod (Y_1, ..., Y_n) \coloneqq T^+ \ra Y \times \mathbb{R}$. Dimension 4 is the last dimension where a bounding integral chain can always be assumed to be the image $g(X)$ of a connected manifold $X$. In fact, choosing a spin structure on $Y$, $T^+$ inherits a spin structure and there is no loss of generality in assuming $X$ has a spin structure extending the one on $\de X = T^+$. (Technically this can be proven from the Aliyah-Bott spectral sequence which computes spin bordism from (oriented bordism) $\otimes$ spin bordism (pt.) as its $E_1$ term. But it is geometrically rather direct, since spin bordism of a point vanishes through dimension three.)

	Again to apply immersion theory, puncture $X$ to obtain $X^-$ and an immersion $\widetilde{g}(X^-)$ bounding $T^+$ $\coprod$ $\widetilde{g}(S^3)$, and obtain:
	\begin{equation}
	\label{punctured}
		f(\alpha_n) = nf(\alpha) + f_\alpha (\widetilde{g}(S^3)),
	\end{equation}
	where $f(\alpha_n)$ is the flux in through $T$, $f(\alpha)$ the flux out through each copy of $Y$ and the third term any GNVW flux that might pass through $\widetilde{g}(S^3)$.

	Again $\widetilde{g}(S^3)$ is homotopically trivial in $Y \times \mathbb{R}$ but now the framing choices and therefore potential regular homotopy classes for the 3-sphere are innumerated by $\pi_3(SO(4)) \cong Z \oplus Z$. According to Milnor \cite{milnor} the two factors are generated by left and right quaternion multiplication, $l$ and $r$ respectively, and in these coordinates the two four-dimensional characteristic classes of the tangent bundle of a spin 4-manifold are given:
	\begin{equation}
		\begin{split}
			& \text{Pontryagin class: } p_1 = 2(l - r) \text{, and} \\
			& \text{Euler class: } \chi = l + r
		\end{split}
	\end{equation}

	We have the freedom to rechoose the bounding spin manifold $X$ by connect summing to it any closed spin 4-manfiold $M$. Consider three possibilities.

	If $M$ is the K3 surface $K$, we have
	\begin{equation}
	\label{Ms1}
		\begin{split}
			& p_1(K) = 48 \\
			& \chi(K) = 24
		\end{split}
	\end{equation}
	If $M$ is $J = S^1 \times S^3$ $\#$ $S^1 \times S^3$, then
	\begin{equation}
	\label{Ms2}
		\begin{split}
			& p_1(J) = 0 \\
			& \chi(J) = -2
		\end{split}
	\end{equation}
	and if $M$ is $L=S^2 \times S^2$, then
	\begin{equation}
		\label{Ms3}
\begin{split}
			& p_1(L) = 0 \\
			& \chi(L) = 4
		\end{split}
	\end{equation}
	Also note the symmetries:
	\begin{equation}
		\begin{split}
			& p_1(M) = -p_1(-M) \text{, and} \\
			& \chi(M) = \chi(-M)
		\end{split}
	\end{equation}

	Take 48 copies of $T^+$ bounded by 48 copies of $X^-$. Now connect sum the 48 copies of $X^-$ together and tube together the 48 punctures to produce a connected manifold $V^-$ with one puncture bounding $48(T^+)$. Immersing $V^-$ into $Y \times \R$, we find Eq.~(\ref{punctured}) becomes:
	\begin{equation}
		48 f(\alpha_n) = 48nf(\alpha) + f_\alpha(\bar{g}(S^3)),
	\end{equation}
	where $\bar{g}$ is the immersion of $S^3$ made by tubing together the 48 individual copies of $\widetilde{g}(S^3)$.

	From lines \ref{Ms1} through \ref{Ms3}, we can add vectors of the form $(48s, 48t) \in Z \oplus Z$ to the framing obstruction around $\bar{g}(S^3)$ by adding copies of $\pm K$, $L$, and $J$ to $\underset{\scriptscriptstyle{\text{48 copies}}}{\#}(X^-)$. Since $\bar{g}(S^2)$ is made by tubing (connected sum) 48 disjoint copies of $\widetilde{g}(S^2)$ together its framing obstruction vector is already divisible by 48. Consequently this freedom allows us to change $\underset{\scriptscriptstyle{\text{48 copies}}}{\#}(X)$ to bring the framing obstruction around $\bar{g}(S^3)$ to $(0,0)$. Now as in Theorem \ref{3dcover}, $\bar{g}(S^3)$ is regularly homotopic to a small round sphere and the Gauss law now applies to show the correction term in lines (7) and (12) vanishes, proving Theorem \ref{4dcover}.
\end{proof}

If one pushes on to still higher dimensions, the line of reasoning we present runs into difficult questions in algebraic topology, for QCA on manifolds of dimension $\geq 5_i$ multiplicativity in finite covers of GNVW flux across hypersurfaces remains an open question.

While the behavior of flux under finite cover is generally unknown when $\dim(Y) \ d - 1 \geq 4$, some special cases are readily handled by dimensional reduction. For example:

The case of hypersurfaces diffeomorphic to $d-1$ torus $T^{d-1}$, all finite covers can be produced by unwrapping one direction at a time. This is because a $n$-fold cover has a order $n$, abelian group $G$ of covering translations, $G = Z^{d-1}/H$, $H$ a rank $d-1$ subgroup of $Z^{d-1}$ of index $n$. Necessarily $H$ is a normal subgroup and the quotient $G$ abelian. $G$ may be factored as a product $\prod_i^{d-1} C_i \cong G$ of cyclic group, with the sequence of "unwrappings" given by the subgroups $\pi^{-1}(\prod_{i=j}^{d-1}C_1)$, $j = 1, \dots, d-1$, where $\pi: Z^{d-1} \ra G$ is the quotient map.

Then, by repeated use of dimensional reduction one may show that for any $T^\pr \ra T^{d-1}$ which is a degree $n$ cover of the $(d-1)$ torus:
\begin{equation}
	nf_T(\alpha) = f_{T^\pr}(\widetilde{\alpha})
\end{equation}

\subsection{Group structure and additivity of flux}
\label{reductionsection}
As also noted in \cite{FHH} we have the following two results:

\begin{lemma}
	Over any acceptable control space $X$, the finite depth quantum circuits (fdqc) constitute a normal subgroup of the QCA and the quotient group QCA/fdqc, which we denote $Q(X)$, is abelian. As noted in subsection \ref{GNVWreview}, all fdqc and QCA are considered up to stabilization, $\otimes \Id$ with a locally bounded number of degrees of freedom.
\end{lemma}

\textbf{Comment on acceptable control spaces $X$}: In the proof below we will need to parse the points $\{x_i\}$ in the control space $X$ into a bounded number of groups of $k(L)$ "colors" so that no two $x_i$ of the same color are closer than some constant distance $L$ apart. Let us rescale so that no two distinct $x_i$ are closer than distance $1$. (To make introduction of ancilla convenient in the proof of Theorem \ref{abelian} we should relax this to allow two distinct - or any bounded number- $x_i$ to map to the same the same location in $X$, but this only affects the constant $k(L)$ by that multiple.) To avoid pathological metrics where no such constant $k(L)$ exists we assume our control space is a subset of Euclidean $n$-space or hyperbolic $n$-space for some n, with the metric induced by inclusion.  By the Nash embedding theorem, the first case allows approximate realization of all metrics on compact Riemannian manifolds, and compact real semi-algebraic varieties. We include the hyperbolic case to allow a very broad range of noncompact control spaces, such as hyperbolic $n$-space, itself.
       To sketch the construction of the $k$-coloring in the hyperbolic case ( the Euclidean case is easier since there the domains discussed  can all be taken to be large  $n$-cubes.) Let $D \subset H^n$ be the fundamental domain for a compact hyperbolic $n$-manifold $M$; copies of $D$ tile $H^n$. By the separation requirement, there are a bounded number $b$ of control points in each copy of $D$.  Give each such point a different color. Now since the fundamental group of a hyperbolic manifolds is residually finite\cite{selberg1960discontinuous}, there is a much larger fundamental domain $D'$ for a large injectivity radius ($>>L$) , finite cover $M'$ of $M$, where $i$ copies of $D$ tile $D'$, $i$ the index of the cover. Now divide the original colors, each into $i$ varieties according to the position of $D$ in $D'$. The tiling of $H^n$ by $D'$ now induces an $i*b$ coloring with the property that points with the same color are always distance $>L$ apart. For embeddings the coloring is induced by inclusion from the coloring constructed in the ambient space.

\vspace{0.5em}
\noindent
\textbf{Notation:} We write QCA, an algebra map, in caps, and fdqc lowercase since it is a unitary, and becomes a QCA only when acting by conjugation.

\begin{proof}
	Let $c$ be a fdqc and $\alpha$ a QCA. We may write $c = \cdots g_2 \circ g_1$ of local gates (when a substring is disjointly represented, they can be simultaneously). Conjugating $g_1$ by $\alpha$ expands its range but $g_1^\alpha$ can still be written as a finite composition of local gates of any desired range (e.g. nearest neighbor) so $g_1^\alpha$ is also fdqc of depth at most a constant $c$ depending on the geometry of the control space $\{x_i\} \in X$ and the range $R$ of $\alpha$. Similarly, all $g_j^\alpha$ are depth $c$ fdqc.

	These circuits $g_i^\alpha$ have depth $c$ and width $2R + r$ where $r$ is an upper bound to the range of the original $g_i$. The increase in width causes a small technical problem: How can these wider circuits be fitted together (without overlapping) to build a new fdqc, $c^\alpha$? Depending on the geometry of $\{x_i\} \in X$ and the range $2R + r$, there is a constant $k$ so that we may $k$-partition or
"$k$-color" 
the original gate supports $\{S_j\}$ ($S_j$ might be a pair of nearby sites $x_i \cup x_k$) so that no supports of the same color are closer than $2R + r$. Then the gates $\{g_j^\alpha\}$ can be applied in a partial order which refines the original partial order on $\{g_j\}$ so that each
"layer"
 within $\{g_j\}$ instead of being applied simultaneously is applied in $k$ disjoint layers. Thus if $c$ had depth $d$,
	\begin{equation}
		(\text{depth } c^\alpha) \leq ckd
	\end{equation}
	In particular conjugating a fdqc by $\alpha$ yields a fdqc, establishing normality.
\end{proof}

\begin{thm}
\label{abelian}
	$Q \coloneqq$ \textup{QCA/fdqc} is an abelian group.
\end{thm}

\begin{proof}
	We show that the group structure under composition agrees with the group structure ($\otimes$-product) under disjoint union, which manifestly is abelian. The argument is summarized in Figure 9 below:

	\begin{figure}[ht]
		\centering
		\begin{tikzpicture}[scale = 0.8]
		\draw (3, 3) -- (3, -3);
		\node at (3.9,0) {$\equiv$};
		\draw  (0.6,1.7) rectangle (2,1);
		\draw  (0.6,-1) rectangle (2,-1.7);
		\draw (1.3, 3) -- (1.3, 1.7);
		\draw (1.3, 1) -- (1.3, -1);
		\draw (1.3, -1.7) -- (1.3, -3);

		\draw  (4.6,1.7) rectangle (6,1);
		\draw  (4.6,-1) rectangle (6,-1.7);
		\draw (4.9,3) to [out = -90, in = 135] (5.3,1.7);
		\draw (5.7,1) to [out = -45, in = 90] (7.2,-0.6) to [out = -90, in = 90] (5.3,-3);
		\draw (7,3) to [out = -90, in = 45] (6.3, .8);
		\draw (5.9, 0.5) to [out = 225, in = 90] (5,-0.5) to [out = -90, in = 125] (5.2,-1);
		\draw (5.7, -1.7) to [out = -45, in = 155] (5.9, -1.9);
		\draw (6.1, -2.1) to [out = -45, in = 90] (6.8, -3);

		\draw (9, 3) -- (9, 1.7);
		\node at (8.1,0) {$=$};
		\draw  (8.3,1.7) rectangle (9.7,1);
		\draw (9, 1) -- (9, -3);
		\draw (10.7, 3) -- (10.7, -1);
		\draw  (10,-1) rectangle (11.4,-1.7);
		\draw (9, 1) -- (9, -3);
		\draw (10.7, -1.7) -- (10.7, -3);
		\node at (11.7,0.1) {$=$};

		\draw (13.2,3) -- (13.2,0.4);
		\draw (13.2,-0.3) -- (13.2,-3);
		\draw (15,3) -- (15,0.4);
		\draw (15,-0.3) -- (15,-3);
		\draw  (12.5,0.4) rectangle (13.9,-0.3);
		\draw  (14.3,0.4) rectangle (15.7,-0.3);

		\node at (1.3,1.3) {$\alpha_1$};
		\node at (1.3,-1.4) {$\alpha_2$};
		\node at (5.3,1.3) {$\alpha_1$};
		\node at (5.3,-1.4) {$\alpha_2$};
		\node at (9,1.3) {$\alpha_1$};
		\node at (10.7,-1.4) {$\alpha_2$};
		\node at (13.2,0) {$\alpha_1$};
		\node at (15,0) {$\alpha_2$};
		\end{tikzpicture}
		\caption{}
	\end{figure}

	The left vertical line represents $\mathcal{H} = \otimes_i \mathcal{H}$ over $\{x_i\} \subset X$ (that is, $X$ is imagined normal to the figure) and the stacking is composition of QCA, $\alpha_2 \circ \alpha_1$. The line to its right represents doubling the degrees of freedom by introducing for every $\mathcal{H}_i$ at $x_i$ an identical ancilla at the same location. The $\equiv$ sign represents equality in the quotient group $Q$. The crossings (under/over has no significance here) are our notation for the swap operators which at $x_i$ \textit{swaps} the local Hilbert space $\mathcal{H}_i$ with its ancillary partner. Notice that a swap is a fdqc of depth one, so may be added at will without changing the element in $Q$. At the far right we have the manifestly abelian group structure.
\end{proof}

It is a natural question which we hope to study later: is $Q(X)$ a grade of some generalized homology theory of the control space $X$? In any case we are at least able to identify a map from $Q(X)$ to a bit of ordinary homology inside $Q$.

Let $H = H_1(X;M)$ when $X$ is compact and $H = H_1^\text{lf}(X;M)$ in the more general case where $X$ is permitted to be noncompact. The essential features of this paper can all be cast in the compact setting where only finitely many degrees of freedom arise but at this juncture we clarify the topological language required to treat noncompact control spaces $X$ such as the Reals, or, for example, an infinite $n$-valent tree.

There is really no loss in assuming $X$ is a manifold, for example, the tree above could be thickened to be 3D and then we could use the boundary of this thickening $\overline{X}$, instead of the original $X$ as our control space. The advantage is that manifolds exhibit Poincar\'{e} duality which is a conceptual convenience, so from here on $X$ will be a $d$-manifold. For any coefficient module $N$ over $\Z$ Poincar\'{e} duality takes these forms:

\begin{align*}
	H_i(X;N) \cong H^{d-i}(X;N), \hspace{4pt} & X \text{ compact} \\
	\begin{rcases}
		H_i^\text{lf}(X;N) \cong H^{d-i}(X;N) \\
		H_i(X;N) \cong H^{d-i}_\text{cs}(X;N) \\
	\end{rcases} & X \text { possibly noncompact}
\end{align*}

\noindent
where lf denotes \emph{locally finite} chains meaning formal chains are allowed to be \emph{infinite} sums as long as no compact set meets more than finitely many formal simplicies, and cs denotes \emph{compact support}, meaning each co-chain is a compactly supported function.

It is known that $H_i^\text{lf}(X;N) \cong \xleftarrow[\scriptscriptstyle{K-\text{compact}}]{\lim} H_i(X, X \backslash K;N)$ and $H^i_\text{cs}(X;N) \cong \xrightarrow[\scriptscriptstyle{K-\text{compact}}]{\lim} H^1(X, X\backslash K; N)$.

Using shifts, it is easy to build a map $s: H(X) \ra Q(X)$.

\begin{thm}
	There is a four term exact sequence where $T$ is some unknown torsion subgroup, possibly trivial.

	\[
	0 \ra T \ra H \xrightarrow{s} Q \ra P \ra 0
	\]

	After dividing by \emph{Tor}, all the torsion in $H$, we get a spit short exact sequence:

	\begin{center}
	\begin{tikzpicture}
		\node at (0,0) {$0 \rightarrow H /$Tor};
		\draw [->] (1.1, 0) -- (1.6, 0);
		\node at (1.35,0.15) {$\scriptstyle \overline{s}$};
		\node at (2.8,0) {$Q \rightarrow P \rightarrow 0$};
		\draw [->] (1.9, -0.2) to [out = 225, in = 0] (1.3, -0.4) to [out = 180, in = -45] (0.7, -0.2);
		\node at (1.3,-0.6) {$\scriptstyle t$};
	\end{tikzpicture}
	\end{center}

	We denote the quotient by $P$, $s$ is for "shift," and $t$ is the splitting.
\end{thm}

\vspace{0.5em}
\noindent
\textbf{Note 1.} If we restrict the degrees of freedom to be qubits, then the coefficients $M$ may be replaced by the integers $\Z$. In some cases, for simplicity we will argue just in that case if the generalization is easy.

\vspace{0.5em}
\noindent
\textbf{Note 2.} When $X$ is a surface we will see in section \ref{algebrasection} that $H \cong Q(X)$.

\begin{proof}
	Take graph $G \subset X$ carrying $H$ and on $G$ implement a local permutation of sites $\{x_i\}$ and Hilbert spaces $\{\mathcal{H}_i\}$carrying degree of freedom $H_i$ so as to induce a local algebra endormoprhism with any desired flux in $H_1(G;M)$; then under inclusion, this maps to the general element in $H^\text{lf}_1(X;M)$. This construction turns out to be unique. To explain why, let us restrict, for simplicity, to the qubit case (integer coefficients) and $X$ compact, where it is sufficient to prove the following lemma.

	\begin{lemma}
	\label{localtranspose}
		Let $G$ be a graph embedded in a compact manifold (or finite simplical complex $X$). Let $\{x_i\}$ be finitely many points distributed on $G$ and $P$ a local permutation with $\flux(P) \in H_1(G;Z)$ so that $\textup{Inc}_\ast(\flux(P)) = O \in H_1(X;Z)$. Then we may stabilize $\{x_i\}$ with additional points in $X$ (still denoted $\{x_i\}$) and extend $P$ to be the identity on these. Then $P$ is a \textit{bounded depth} composition of local transposition of $\{x_i\}$. Here "bounded depth" means that if we assume the spacing of the $\{x_i\}$ along $G$ is small enough and the stabilization also dense enough that the diameter of every site-track under this circuit defining $P$ is of arbitrarily small diameter. The diameter will be a small multiple of the maximum distance that $P$ moves a point.
	\end{lemma}

	The well-definedness of the splitting $t$ is proven by applying Lemma \ref{localtranspose} to $G = G_1 \cup G_2$ and $\{x_i\} = \{x_i^1\} \cup \{x_i^2\}$, and $P = P_1 \cup P_2^{-1}$, where $(G_j, \{x_i^j\}, P_j)$ constitute two constructions of $t$ for $j = 1,2$.

	\vspace{0.5em}
	\noindent
	\textit{Proof (sketch) for Lemma \ref{localtranspose}.}
	The main idea is a kind of swindle that allows a bounded sequence of local permutations to build a cyclic permutation along a null homologous cycle. Figure 10 illustrates the principle for a cycle bounding a disk. Further details are left to the reader.

	\begin{figure}[p!]
		\centering
		\begin{tikzpicture}[scale = 0.6]
			\draw (0,0) circle (3);
			\draw (0,0) circle (3.5);
			\foreach \a in {1,2,...,18}{
			\draw [fill=black] (\a*360/18: 3) circle (0.3ex);
			}
			\foreach \a in {1,2,...,18}{
			\draw [fill=black] (\a*360/18: 3.5) circle (0.3ex);
			}
			\draw [<-] (360/18:2.6) arc (360/18:720/18-3:2.6);
			\draw [->] (1080/18:2.6) arc (1080/18:720/18+3:2.6);
			\draw [->] (1440/18:2.6) arc (1440/18:1080/18+3:2.6);
			\draw [->] (360/18:3.8) arc (360/18:720/18-3:3.8);
			\draw [<-] (1080/18:3.8) arc (1080/18:720/18:3.8);
			\draw [<-] (1440/18:3.8) arc (1440/18:1080/18+3:3.8);
		\end{tikzpicture}

		\begin{flushleft}
		\vspace{0.5em}
		Build two counter-rotating rings by a bound sequence of transpositions. To do this note:
		\vspace{0.5em}
		\end{flushleft}

		\begin{tikzpicture}
			\node at (-0.2,0) {\Large$\cdots$};
			\draw (0.5,0.5) rectangle (1.5,-0.5);
			\draw [->] (0.99, 0.5) -- (1,0.5);
			\draw [->] (1.01, -0.5) -- (1, -0.5);
			\draw [->] (0.5, -0.01) -- (0.5, 0);
			\draw [->] (1.5, 0.01) -- (1.5, 0);

			\draw (2,0.5) rectangle (3,-0.5);
			\draw [->] (2.49, 0.5) -- (2.5,0.5);
			\draw [->] (2.51, -0.5) -- (2.5, -0.5);
			\draw [->] (2, -0.01) -- (2, 0);
			\draw [->] (3, 0.01) -- (3, 0);

			\draw (3.5,0.5) rectangle (4.5,-0.5);
			\draw [->] (3.99, 0.5) -- (4,0.5);
			\draw [->] (4.01, -0.5) -- (4, -0.5);
			\draw [->] (3.5, -0.01) -- (3.5, 0);
			\draw [->] (4.5, 0.01) -- (4.5, 0);

			\draw (5,0.5) rectangle (6,-0.5);
			\draw [->] (5.49, 0.5) -- (5.5,0.5);
			\draw [->] (5.51, -0.5) -- (5.5, -0.5);
			\draw [->] (5, -0.01) -- (5, 0);
			\draw [->] (6, 0.01) -- (6, 0);

			\draw (6.5,0.5) rectangle (7.5,-0.5);
			\draw [->] (6.99, 0.5) -- (7,0.5);
			\draw [->] (7.01, -0.5) -- (7, -0.5);
			\draw [->] (6.5, -0.01) -- (6.5, 0);
			\draw [->] (7.5, 0.01) -- (7.5, 0);
			\node at (8.3,0) {\Large$\cdots$};

			\foreach \a in {1,2,3,4,5} {
			\draw [fill = black] (\a*1.5-1, 0.5) circle (0.3ex);
			\draw [fill = black] (\a*1.5, 0.5) circle (0.3ex);
			\draw [fill = black] (\a*1.5-1, -0.5) circle (0.3ex);
			\draw [fill = black] (\a*1.5, -0.5) circle (0.3ex);
			}

			\node at (-0.2,-2) {\Large$\cdots$};
			\node at (4,-1.05) {\small followed by};

			\foreach \a in {1,2,3,4,5} {
			\draw [<->] (\a*1.5-.95, -2.45) -- (\a*1.5-0.05, -1.55);
			\draw [->] (\a*1.5-0.55, -1.95) -- (\a*1.5-.95, -1.55);
			\draw [->] (\a*1.5-0.45, -2.05) -- (\a*1.5-0.05, -2.45);
			\draw [fill = black] (\a*1.5-1, -2.5) circle (0.3ex);
			\draw [fill = black] (\a*1.5, -1.5) circle (0.3ex);
			\draw [fill = black] (\a*1.5-1, -1.5) circle (0.3ex);
			\draw [fill = black] (\a*1.5, -2.5) circle (0.3ex);
			}
			\node at (8.3,-2) {\Large $\cdots$};

			\node at (4,-3.05) {\small is equal};
			\node at (-0.2,-3.8) {\Large$\cdots$};
			\draw (0.5,-3.6) -- (7.5,-3.6);
			\draw (0.5,-4) -- (7.5,-4);
			\foreach \a in {1,2,...,10}{
			\draw [fill = black] (\a*7/9-0.3, -3.6) circle (0.3ex);
			\draw [fill = black] (\a*7/9-0.3, -4) circle (0.3ex);
			}
			\foreach \a in {1,2,...,9}{
			\draw [->] (\a*7/9+0.14, -3.6) -- (\a*7/9+0.15, -3.6);
			\draw [<-] (\a*7/9+0.14, -4) -- (\a*7/9+0.15, -4);
			}
			\node at (8.3,-3.8) {\Large$\cdots$};
		\end{tikzpicture}

		\begin{flushleft}
		\vspace{0.5em}
		This can be made periodic and a difference in number of sites in these adjacent rows can also be accomodated locally a replacement of the form:	
		\vspace{0.5em}	
		\end{flushleft}

		\begin{subfigure}{\textwidth}
		\centering
		\begin{tikzpicture}
			\draw (0.5,0.5) rectangle (1.5,-0.5);
			\draw [->] (0.99, 0.5) -- (1,0.5);
			\draw [->] (1.01, -0.5) -- (1, -0.5);
			\draw [->] (0.5, -0.01) -- (0.5, 0);
			\draw [->] (1.5, 0.01) -- (1.5, 0);

			\draw [->] (1.9,0) -- (2.4,0);

			\draw (2.8, 0.5) -- (4.3, 0.5) -- (3.9, -0.5) -- (3.2, -0.5) -- (2.8, 0.5);

			\draw [->] (3.19, 0.5) -- (3.2,0.5);
			\draw [->] (3.89, 0.5) -- (3.9,0.5);
			\draw [->] (3.51, -0.5) -- (3.5, -0.5);
			\draw [->] (3.005, -0.01) -- (3, 0);
			\draw [->] (4.105, 0.01) -- (4.1, 0);

			\draw [fill = black] (0.5, 0.5) circle (0.3ex);
			\draw [fill = black] (0.5, -0.5) circle (0.3ex);
			\draw [fill = black] (1.5, 0.5) circle (0.3ex);
			\draw [fill = black] (1.5, -0.5) circle (0.3ex);
			\draw [fill = black] (2.8, 0.5) circle (0.3ex);
			\draw [fill = black] (3.5, 0.5) circle (0.3ex);
			\draw [fill = black] (4.3, 0.5) circle (0.3ex);
			\draw [fill = black] (3.2, -0.5) circle (0.3ex);
			\draw [fill = black] (3.9, -0.5) circle (0.3ex);
		\end{tikzpicture}
		\caption{}
		\end{subfigure}

		\begin{flushleft}
		\vspace{0.5em}
		Build $n$ such ring pairs and then cancel the inner one locally. Finally cancel $R_{2n-1}$ with $R_{2n-2}$, ..., $R_3$ with $R_2$ leaving only $R_1$.
		\vspace{0.5em}
		\end{flushleft}

		\begin{subfigure}{\textwidth}
		\centering
		\begin{tikzpicture}[scale = 0.85]
			\draw (0,0) circle (3);
			\draw (0,0) circle (2.3);
			\draw (0,0) circle (1.6);
			\draw (0,0) circle (0.5);
			\draw (0,0) circle (0.2);
			\node at (-2.1,2.6) {\small $R_1$};
			\node at (-1.7,2) {\small $R_2$};
			\node at (-1.3,1.4) {\small $R_3$};
			\node at (-0.6,0.9) {$\ddots$};
			\node at (0,-0.3) {\small $R_{2n-1}$};

			\draw (.1, 3.1) -- (0, 3) -- (.1, 2.9);
			\draw (-0.1, 2.4) -- (0, 2.3) -- (-0.1, 2.2);
			\draw (.1, 1.7) -- (0, 1.6) -- (.1, 1.5);
			\draw (.1, 0.6) -- (0, 0.5) -- (.1, 0.4);
			\draw (-0.1, 0.27) -- (0, 0.2) -- (-0.07, 0.07);
		\end{tikzpicture}
		\caption{}
		\end{subfigure}
		\caption{}
	\end{figure}

	(Note: All symmetric groups $S_n$ are generated by transpositions, so for $n$ small we do not need to bother explicitly expressing local permutations as products of transpositions.)
\end{proof}

It remains to construct the splitting $t$. We will also need the universal coefficient exact sequence (UCT):

\[
0 \ra \mathrm{Ext}_\Z^1(H_{i-1}(W;\Z),N) \ra H^i(W;N) \ra \mathrm{Hom}_\Z(H_i(W,\Z),N) \ra 0
\]

When $i = 1$ the Ext-term vanishes. Set $N = M$. If $X$ is compact set $W = X$, if $X$ is noncompact set $W = (X, X\backslash K)$ for $K$ compact in $X$. We get

\[
H^1(X;M) \xrightarrow{\cong} \mathrm{Hom}_\Z(H_1(X;\Z),M)
\]

\noindent
or in the noncompact case:

\begin{align*}
& H^1(X, X\backslash K;M) \xrightarrow{\cong} \mathrm{Hom}_\Z(H_1(X, X\backslash K; \Z),M) \\
\xrightarrow[K]{\lim} ( & \equalto{H^1(X, X \backslash K; M)}{H^1_\text{cs}(X;M)} \xrightarrow{\cong} \equalto{\mathrm{Hom}_\Z(\xleftarrow[K]{\lim}H_1(X, X\backslash K; \Z), K))}{\mathrm{Hom}_\Z(H_1^\text{lf}(X;\Z),M)}
\end{align*}

The lesson is if $X$ is noncompact we take cohomology with compact supports and locally finite homology. If $X$ is compact one may \emph{ignore} these decorations. Now taking $\mathrm{Hom}_\Z$ we find natural isomorphisms:

\[
\mathrm{Hom}_\Z(H_\text{cs}^1(X;\Z),M) \cong \mathrm{Hom}_\Z(H_\text{cs}^1(X;M),\Z) \ra \mathrm{Hom}_\Z(\mathrm{Hom}_\Z(H_1^\text{lf}(X;Z),M)) \cong H_1^\text{lf}(X;M)/\textup{Torsion}
\]

The first and last isomorphism since the module $M$ is torsion free.

On the right we have our proposed home for flux; on the left we have linear functions from hypersurfaces to $M$. To see this latter fact note that $H^1(X;\Z) \cong [X,S^1]$ homotopy classes of maps to the circle and, more generally, $H^1_\text{cs}(X;\Z) \cong [\hat{X}; S^1]$ where $\hat{X}$ is the one-point-compactification of $X$. By Poincar\'{e} duality $H^1_\text{cs}(X;\Z) \cong H_{d-1}(X;\Z)$, the linearized hypersurfaces. Corresponding to a class $a \in H^1_\text{cs}(X;\Z)$ is the hypersurface $Y = f^{-1}(\ast \in S^1)$ where $[f] \in [\hat{X}, S^1]$ corresponds to $a$.

Given a QCA $\alpha$ with sufficiently small range on $X$, then using dimensional reduction, flux can be defined as a function $f(\alpha): H_{d-1}(X;Z) \ra M$. We can now state:

\begin{thm}
\label{thmcohomology}
	For a QCA $\alpha$ on $X$, the flux $f(\alpha)$ is a linear map $f(\alpha):H_{d-1}(X;\Z) \ra M$. Hence (by the UCT) $\alpha$ determines a cohomology class $[\alpha] \in H^{d-1}(X;M)/\textup{Torsion}$, or equivalently via Poincar\'{e} duality a class $[\alpha]^\wedge \in H_1^\textup{lf}(X;M)/\textup{Torsion}$
\end{thm}

\begin{remark}
We do not know if $\alpha$ in some sense determines a divergenceless vector field (or equivalently a closed $d-1$ form), but at least it contains the corresponding cohomological information.
\end{remark}

\begin{proof}
	Curiously it is completely general to give the argument when $X$ is $T^2$, a 2-torus. We need to show that for two embedded $(d-1)$ submanifolds $Y_1$ and $Y_2 \subset X$ that
	\begin{equation}
	 	f_\alpha(Y_1) + f_\alpha(Y_2) = f_\alpha(Y_3),
	\end{equation}
	where $Y_3$ is an embedded $d-1$ submanifold in the sum of the two homology classes:
	\begin{equation}
		[Y_1] + [Y_2] = [Y_3]
	\end{equation}

	A prototypical example of this is $Y_1 = \text{meridian} = m$ and $Y_2 = \text{longitude} = l$ on the torus $T^2$. But more importantly, using the classification property of $S^1$, every example $(X; Y_1; Y_2)$ may be written as the transverse inverse image of $(T^2;m,l)$. Let $\te_i:X \ra S^1$ classify $Y_i, i = 1,2$, so that $Y_i = \te_i^{-1}(\ast) = $ the preimage of the base point. Now setting $\te_1 \times \te_2: X \xrightarrow{\te} S^1 \times S^1 = T^2$, we see that $Y_1 = \te^{-1}(\ast \times S^1)$ and $Y_2 = \te^{-1}(S^1 \times \ast)$.

	In the special case where $(X; Y_1, Y_2) = (T^2, m, l)$ we need to argue that the $\alpha$-flux through the diagonal $\Delta = m + l$ is the sum of the two fluxes through $m$ and $l$ respectively.
In fact, this follows immediately from theorem \ref{applic} later, as the QCA is stably equivalent to a composition of two shift QCA, one shifting on the meridian and one shifting on the longitude, and one can directly compute the fluxes in this case.

However, it is also interesting to see that one can prove the additivity of fluxes using topological methods, without using the full machinery of section \ref{algebraicsection}.  We will need some input from that section below, but not as much.
 We do this via a large but finite covering space $T^2 \ra T^2$ (degree $n^2$). The general case then follows by pulling back all constructions over $\te$ and computing using dimensional reduction (to 1D).

	The 2-torus $T$ may be dimensionally reduced either to a horizontal circle $m$ or a vertical circle $l$. See Figure 11(A) and 11(B).

	\begin{figure}[ht]
		\centering
		\begin{subfigure}{0.3\textwidth}
		\begin{tikzpicture}[scale = 0.7]
	 		\draw (1.8,0) -- (1.8,5) -- (7.8,5) -- (7.8,0) -- (1.8,0);
			\draw (2.8,0) -- (2.8,5);
			\draw (3.8,0) -- (3.8,5);
			\draw (4.8,0) -- (4.8,5);
			\draw (5.8,0) -- (5.8,5);
			\draw (6.8,0) -- (6.8,5);
			\node at (4.9,5.4) {$T^2$};
			\node at (4.8,-0.4) {$\ast$};
			\draw [->] (4.8, -0.7) -- (4.8, -1.5);
			\draw (1.8,-1.8) -- (7.8,-1.8);
			\node at (4.8,-2.1) {$m$};
		\end{tikzpicture}
		\caption{}
		\end{subfigure}
		\hspace{1em}
		\begin{subfigure}{0.3\textwidth}
		\begin{tikzpicture}[scale = 0.7]
			\draw (11.5,-1) -- (11.5,5) -- (17.5,5) -- (17.5,-1) -- (11.5,-1);
			\node at (14.6,5.4) {$T^2$};
			\draw (11.5, 1) -- (17.5, 1);
			\draw (11.5, 2) -- (17.5, 2);
			\draw (11.5, 3) -- (17.5, 3);
			\draw (11.5, 4) -- (17.5, 4);
			\draw (11.5, 0) -- (17.5, 0);
			\node at (11.1,2) {$\ast$};
			\draw [->] (17.7, 2) -- (18.3, 2);
			\draw (18.5,5) -- (18.5,-1);
			\node at (18.9,2) {$l$};
	 	\end{tikzpicture}
	 	\caption{}
	 	\end{subfigure}

		\begin{flushleft}
		We splice these together at the base point of each family to produce 11(C), a decomposition of $T^2$ whose quotient space is a "tadpole," a circle glued to the end point of an arc.
		\vspace{1em}
		\end{flushleft}

	 	\centering
	 	\begin{subfigure}{0.65\textwidth}
	 	\hspace{5em}
	 	\begin{tikzpicture}
	 		\draw (0.5,0) |- (0, 0.5);
			\draw (1,0) |- (0, 1);
			\draw (1.5,0) |- (0, 1.5);
			\draw (1.5, 1.5) -- (3.5,3.5);
			\draw (3.5,5) |- (5,3.5);
			\draw (4,5) |- (5,4);
			\draw (4.5,5) |- (5,4.5);

			\draw (0.5,5) |- (0,4.5);
			\draw (1,5) |- (0,4);
			\draw (1.5, 5) |- (0, 3.5);
			\draw (0,3) -- (1.3,3) -- (2,3.7) -- (2,5);
			\draw (0,2.5) -- (1.4,2.5) -- (2.5,3.6) -- (2.5,5);
			\draw (0, 2) -- (1.5, 2) -- (3, 3.5) -- (3, 5);

			\draw (4.5, 0) |- (5, 0.5);
			\draw (4, 0) |- (5, 1);
			\draw (3.5, 0) |- (5, 1.5);
			\draw (3,0) -- (3,1.3) -- (3.7,2) -- (5,2);
			\draw (2.5,0) -- (2.5, 1.4) -- (3.6, 2.5) -- (5, 2.5);
			\draw (2, 0) -- (2, 1.5) -- (3.5, 3) -- (5, 3);
			\node at (0.5,5.2) {$a$};
			\node at (1,5.2) {$b$};
			\node at (1.5,5.2) {$c$};
			\node at (2,5.2) {$d$};
			\node at (2.5,5.2) {$e$};
			\node at (3,5.2) {$f$};
			\node at (3.5,5.2) {$c$};
			\node at (4,5.2) {$b$};
			\node at (4.5,5.2) {$a$};
			\node at (5.2,3) {$d$};
			\node at (-0.2,3) {$d$};
			\node at (-0.2,2.5) {$e$};
			\node at (2,-0.2) {$d$};
			\node at (2.5,-0.2) {$e$};
			\node at (3,-0.2) {$f$};
			\draw [->] (5.9,2.5) -- (7.8,2.5);
			\node at (6.8,2.7) {quotient};
		\end{tikzpicture}
		\caption{}
		\end{subfigure}
		\begin{subfigure}{0.3\textwidth}
		\hspace{1em}
		\begin{tikzpicture}
			\draw  (10.7,1.2) circle (1);
			\draw (10.7,2.2) -- (10.7,4.6);

			\fill (11.7,1.2) circle[radius=2pt];
			\fill (9.7,1.2) circle[radius=2pt];
			\fill (10.7,2.2) circle[radius=2pt];
			\fill (10.7,0.2) circle[radius=2pt];
			\fill (10.7,3.4) circle[radius=2pt];
			\fill (10.7,4.6) circle[radius=2pt];
			\node at (11,4.6) {$a$};
			\node at (11,3.4) {$b$};
			\node at (10.9,2.3) {$c$};
			\node at (12,1.2) {$d$};
			\node at (10.7,-0.1) {$e$};
			\node at (9.4,1.2) {$f$};
	 	\end{tikzpicture}
	 	\caption{}
	 	\end{subfigure}
	\caption{}
	\end{figure}

	\begin{claim}
	$|f_m(\alpha) + f_l(\alpha) - f_\Delta(\alpha)| < c$, a constant depending on the range $R$ and the local Hilbert space dimensions, where a diagonal representative $\Delta$ may be identified with the preimage of point on the loop $e$ in Figure 11(D), and the constant $c$ can be taken to be the product of the Hilbert space dimensions within radius $2R$ of $\ast = l \cap m$.
	\end{claim}

	This claim is established by comparing the computation of support algebra dimensions for $\alpha$ near $l$, $m$, and $\Delta$, the preimage of $e$ in Figure 11(D), respectively, where we note that $\Delta$ is in the class of the diagonal. We know nothing detailed of the contribution to support algebras near $\ast$ and this ignorance yields the inequality. The derivation of this estimate is given in subection \ref{estimate} and it requires some input from the theory of section \ref{algebraicsection} later.

	Now let us pull back $\alpha$ on $T = T^2$ to $\widetilde{\alpha}$ on an $n \times n$ square covering $\widetilde{T} \ra T$. By Theorem \ref{2dcover}:

	\[
	f_{\widetilde{l}} = nf_l \hspace{1em} \text{and} \hspace{1em} f_{\widetilde{m}} = nf_m
	\]

	Let $\widetilde{\Delta}$ be the (small scale) resolution of $\widetilde{l} \cup \widetilde{m}$ to a scc on $\widetilde{T}$ (we use the same local resolution used to obtain $\Delta$ as the preimage of $e$ in Figure 11(D)). Over $l$ and $m$ the \textasciitilde \hspace{2pt} indicates a connected component of the inverse image in $\widetilde{T}$ under the covering map. We use $\dbtilde{\Delta}$ to denote any component of the preimage of $\Delta$ in $\widetilde{T}$. Although $\dbtilde{\Delta}$ has undergone $n$ local resolutions it is clearly in the same homology class as $\widetilde{\Delta}$, which has a single resolution. Since the range $R$ of $\alpha$ and the local Hilbert space dimensions are unchanged by pullback, from that single resolution we have:
	\begin{equation}\label{constant}
	 	|f_{\widetilde{l}} + f_{\widetilde{m}} - f_{\widetilde{\Delta}}| \leq c \text{ (same const. as in the claim)}
	\end{equation}

	Directly from Theorem \ref{2dcover},
	\begin{equation}
	\label{nmtimes}
	 	f_{\widetilde{l}} = nf_l, \hspace{1em} f_{\widetilde{m}} = n f_m, \hspace{1em} and \hspace{1em} f_{\dbtilde{\Delta}} = nf_\Delta
	\end{equation}

	\begin{claim}
	Since $\dbtilde{\Delta}$ and $\widetilde{\Delta}$ are homologous, $f_{\dbtilde{\Delta}} = f_{\widetilde{\Delta}}$
	\end{claim}

	\begin{proof}
		$\dbtilde{\Delta}$ and $\widetilde{\Delta}$ are isotopic scc so the proof is essentially the same as in Lemma \ref{welldefined}. We may deform $\dbtilde{\Delta}$ to $\widetilde{\Delta}$ by a sequence of locally supported isotopies, none of which can alter the transverse flux since at each step it may be alternatively computed as the flux across a distant parallel copy of that scc. (The argument only requires the weaker relation: homology.)
	\end{proof}

	From Eq.~(\ref{nmtimes}) we now have
	\begin{equation}
	 	f_{\widetilde{\Delta}} = nf_\Delta
	\end{equation}
	From (\ref{constant}) and (\ref{nmtimes}):
	\begin{eqnarray}
	\label{fluxadds}\nonumber
	 	|nf_l + nf_m - nf_\Delta| & \leq c \\ \nonumber
	 	|f_l + f_m - f_\Delta| & \leq \frac{c}{n} \text{ , for all $n$} \\
	 	f_l + f_m & = f_\Delta \text{, as desired}
	\end{eqnarray}

	\begin{figure}[ht]
		\centering
		\begin{tikzpicture}[scale = 1.5]
			\draw  (0,3) rectangle (3,0);
			\draw [very thin, gray] (0,2) -- (3,2);
			\draw [very thin, gray] (0,1) -- (3,1);
			\draw [very thin, gray] (1,3) -- (1,0);
			\draw [very thin, gray] (2,3) -- (2,0);

			\draw (0, 1.5) -- (1,1.5);
			\draw (1.5, 3) -- (1.5, 2);
			\draw (1.5,2) arc (0:-90:0.5);
			\draw (3, 1.5) -- (2, 1.5);
			\draw (1.5, 0) -- (1.5, 1);
			\draw (1.5,1) arc (180:90:0.5);
			\draw [dashed] (1.5,1.5) circle (0.4);
			\node at (1.7,2.5) {$\widetilde{e}$};
			\draw [<-] (1.2,1.1) -- (0.6,-0.3);
			\node at (1.5,-0.4) {error $c$ occurs here and does};
			\node at (1.5,-0.7) {not scale with $n$};
		\end{tikzpicture}
		\caption{$\widetilde{\Delta} \subset \widetilde{T}$ for $n = 3$}
		\label{figerrorregion}
	\end{figure}

	As noted at the start of the proof, in the general case, sites in $X$ are aggregated by first pushing forward to $T^2$ via $\te$ and then applying the explicit decomposition of $T^2$ shown in Figure 11(c). All equalities and estimates obtained in $T^2$ now hold without change in $X$ via pullback.
\end{proof}

\begin{remark}
The splitting $t$ shows that something like an incompressible flow; actually just its homology class, a closed $(d-1)$-form $\in H^{d-1}(X,\R) \cong \mathrm{Hom}(H_{d-1}(X;\R),\R)$ is associated to a QCA regardless of the dimension $d$. What we actually map to is $\mathrm{Hom}(H_{d-1}(X;Z),M)$. Recent work of \cite{FHH} show that although QCAs contain this information they also contain more mysterious torsion.
\end{remark}

\begin{remark}
We did not need, in proving Claim 2, a strong fact that we record here for future application.
\end{remark}

\begin{fact}
For two disjoint oriented 1-manifolds $\alpha, \beta$ on a surface $\Sigma$ the equivalence relation $\equiv$ generated by $\alpha \sim \beta$ if $\alpha \cap \beta = \varnothing$ and $\alpha \cup -\beta$ cobound a surface $\Sigma_0 \subset \Sigma$ is the same as $\equiv_H$, $\alpha$ and $\beta$ satisfy $\alpha \equiv_H \beta$ if $[\alpha] = [\beta] \in H_1(\Sigma;Z)$, i.e. the homological equivalence.
\end{fact}

It would take us too far afield to prove this topological fact here; we only note that it implies that the flux through a scc $\gamma$ of a sufficiently local QCA on a surface $\Sigma$ depends only on the homology class of $\gamma$.

\subsubsection{Derivation of Estimate}
\label{estimate}
We show the existence of a constant $c(\alpha)$ in Claim 1.  See Fig.~\ref{figerrorregion}.
Recall that the flux is defined by the difference between the $\frac{1}{2}\log$ dimension of a support algebra and the $\frac{1}{2}\log$ dimension
of an algebra on local sites.
Consider the curve $\overline \Delta$ shown in the figure.  The relevant support algebra is as follows.  Choose one side of the curve and call this side the "right side".  Take the full algebra on  sites near the curve, apply QCA $\alpha$ to this algebra, and consider its support algebra on the sites to the right of the curve.  Call this support algebra $A_\Delta$.  Similary, for curves $l$ or $m$, we define support algebra $A_l,A_m$.

We will first show that
$ f_\Delta(\alpha)  - f_m(\alpha)- f_l(\alpha) < c$.  Then, by a slight modification of the argument we will show that
$f_m(\alpha) + f_l(\alpha) - f_\Delta(\alpha) < c$, establishing the claim.
To show the first inequality, we use a result, which follows from theorem \ref{lltoy}, that algebra $A_{\Delta}$ is generated by two subalgebras, which we call
${\mathcal C}_1,{\mathcal C}_2$, such that these subalgebras commute with each other
and such that
${\mathcal C}_2$ is supported a distance greater than $R$ form the sites inside the dashed circle in Fig.~\ref{figerrorregion}, while ${\mathcal C}_1$ is supported near (within distance $O(R)$) of sites inside dashed circle.
Thus, the dimension of ${\mathcal C}_1$ is bounded by some constant.
However, we claim that ${\mathcal C}_2$ is a subalgebra of the algebra generated by $A_l,A_m$.  To see this, for any operator $O$ in ${\mathcal C}_2$, note that $\alpha^{-1}(O)$ is supported near $l \cup m$, and hence is in the site algebra used to define $A_l,A_m$, i.e., away from the dashed circle, $\Delta$ and $l \cup m$ contain the same set of sites.
This establishes the bound on the difference in dimensions.

To show that 
$f_m(\alpha) + f_l(\alpha) - f_\Delta(\alpha) < c$, we use a similar argument.  For $A_l$ (and similarly for $A_m$)
we show that it is generated by subalgebras,
${\mathcal C}_1,{\mathcal C}_2$, such that these subalgebras that commute with each other
and such that
${\mathcal C}_2$ is supported a distance greater than $R$ form the sites inside the dashed circle in Fig.~\ref{figerrorregion}, while ${\mathcal C}_1$ is supported near (within distance $O(R)$) of sites inside dashed circle.
Thus, the dimension of ${\mathcal C}_1$ is bounded by some constant.
However, ${\mathcal C}_2$ is a subalgebra of the algebra generated by $A_\Delta$.  A similar argument for $A_m$ establishes the bound on the difference in dimensions.

\subsection{Extending a germ periodically}
\label{productsection}
In this subsection we use a modification of the Kirby torus trick [K] and \cite{Hastings2013} (see the latter reference for the QCA context) to get a better understanding of the role that germs of QCA play in the subject. The results here answer an initial question but there are natural questions that are still unanswered. For example, we do not know how to construct a deformation of an initial germ to the periodic extension constructed next.

\begin{defin}
  A \emph{quantum system $\mathcal{H}$} on a Riemannian manifold $Y$ is a collection of finite dimensional Hilbert spaces $\{H_i\}$ indexed by a locally finite collection of points $\{y_i\}$ in $Y$. It is not necessary to assume that $y_i \neq y_j$ whenever $i \neq j$. However, we assume that in any ball of radius $r$ only finitely many degrees of freedom are associated to points within that ball.

One might also require that for all $y \in Y$, there is a $y_i$ associated to an $H_i$ of $\dim H_i \geq 2$ so that

\[
\text{dist}(y,y_i) < < \text{inj. rad}(Y, y)
\]

But this is not essential to the proof we give. However, without such an assumption it makes little sense to think of $Y$ as parameterizing the quantum system $\mathcal{H}$.
\end{defin}

\begin{defin}
  The operator algebra $\mathcal{O}$ of $\mathcal{H}$ is

  \[
  \mathcal{O} = \mathrm{End}(\underset{\scriptscriptstyle{i}}{\otimes}H_i) \cong \underset{\scriptscriptstyle{i}}{\otimes} \mathrm{End}(H_i)
  \]

  \noindent
  in the case that $i$ is from a finite index $I$ and

  \[
  \mathcal{O} = \lim_{\to} \mathrm{End}(\underset{\scriptscriptstyle{i \in F}}{H_i})
  \]

  \noindent
  where the direct limit is taken over the finite subsets $F$ of $I$. Several topologies are possible on this limit, for example: weak, trace, and norm; but restricted to local operators, they all agree, up to constants.

  For us, all algebras are $\ast$-algebras and all homomorphisms unital.
\end{defin}

\begin{note}
With an appropriately rescaled trace, the infinite case $\mathcal{O}$ becomes the hyperfinite type II factor.
\end{note}

\begin{defin}
  For us, all homomorphisms carry unit to unit, so all "automorphisms" are implicitly \emph{unital}. An automorphism $\alpha$ of $\co$, or equivalently "over $Y$", is said to be \emph{$R$-local} (= \emph{"range $R$"}) if for every operator $A \in \co$ with $\sup(\alpha(A)) \subset \mathcal{N}_R(\sup(A))$ i.e. $\alpha(A)$ has support within an $R$-neighborhood of $\sup(A)$. $\co^{-1}$ is easily seen to be $R$-local as well. Similarly, a homomorphism $\gamma: \co \ra \co^\pr$, $\co$ supported in $\text{U} \subset Y$, $\co^\pr$ supported in $\text{U}^\pr \subset Y$ is $R$-local iff $\sup(\gamma(A)) \subset \mathcal{N}_R(\sup(A))$, for every $A \in \co$.
\end{defin}

\begin{defin}
  Let $\text{U} \subset Y$ be an open set. Then we say the composition $\beta = \gamma_2 \circ \gamma_1$ is a \emph{ragged} $R$-autmorphism over $\text{U}$ if there are $R$-homomorphisms

  \begin{align*}
    \gamma_1: & \co_{\text{U}} \ra \co_{\text{V}}, \text{ V} \subset \mathcal{N}_R(\text{U}) \\
    \gamma_2: & \co_{\text{V}} \ra \co_{\text{W}}, \text{ U} \subset \text{W}
  \end{align*}

  \noindent
  so that $\gamma_2 \circ \gamma_1$ is the map $\mathrm{Inc}_\#: \co_{\text{U}} \ra \co_{\text{W}}$ induced by inclusion of $\text{U}$ into $\text{W}$.
\end{defin}

\begin{note}
Any automorphism $\alpha$ of $\mathrm{End}(\C^n)$ is "inner" in the sense of being conjugation by a unitary: $\alpha(A) = \text{U} A \text{U}^+$. If $\alpha$ is $R$-local so will be the conjugating unitary.
\end{note}

We typically consider systems, endomorphism algebras, and their local automorphisms up to stabilization. This means that we free add additional finite dimensional Hilbert spaces $H_j$, $j \in J$, located at $y_j \in Y$ and extend $\alpha$ to the identity on these. In atomic physics these would be "inner orbitals," remote from the low energy physics.

The theme of this subsection is to begin to define and prove analogs, in this algebraic setting, of famous theorems in manifold topology from the 1960's: the product structure theorem\cite{browder1965structures} and the uniqueness of collars \cite{brown1962locally}. We prove:

\begin{thm}["Extend a germ"]
\label{extend}
  Let $\beta$ be a ragged $R$-automorphism over $X \times (0,1) \subset X \times \R$. Assume $r << R << 1$ and $R << r_0$, the injectivity radius of the Riemannian manifold $X$. Then $\beta|_{X \times [-1/2, 1/2]}$ may be extended to a 2-periodic $R$-local automorphism $\alpha$ over $X \times \R$, or equivalently over $X \times S_2^1$, $S_2^1$ being the circle of length 2, $[-1,1]/_{-1 \cong 1}$.
\end{thm}

\begin{proof}
  The chief topological ingredient is the existence of the immersion $\overset{\scriptstyle{e}}{\hookrightarrow}$ extending horizontal and vertical maps in the diagram below:

  \begin{figure}[ht]
    \begin{tikzpicture}
      \node at (-3.2,2.5) {$X\times (-\frac{1}{2}, \frac{1}{2}) \hookrightarrow X \times (-1,1)$};
      \node [rotate = 90, scale = 1.5] at (-5,1.6) {$\longhookleftarrow$};
      \node at (-5,0.6) {$((X \times S^1_2)$\textbackslash pt.)};
      \node [rotate = 45, scale = 2] at (-3.2,1.4) {$\longhookrightarrow$};
      \node [rotate = 45] at (-3.3,1.6) {$e$};
      \draw [->] (-3.5, 0.6) -- (-1.45,2.1);
      \node at (-2.6,1.4) {$f^-$};
    \end{tikzpicture}
    \caption{}
  \end{figure}

  There is a "folding map" $f$ which is the identity on the x-coordinate and "folds" $S^1_2$ into $[-\frac{1}{2}, \frac{1}{2}]$, generically a 2-1 map. Removing a single point reduces the homotopy dimension of the source to less than the target putting us within the scope of the Smale-Hirsh immersion theorem, which reduce the question: "can $f \coloneqq f$ restricted to $(X \times S^1_2)$\textbackslash pt., called $f^-$, be homotopic (rel $X \times [-\frac{1}{2}, \frac{1}{2}]$) to an immersion" to homotopy theory. One must find a tangent bundle injection covering $f^-$ (and restricting to the differential of $f$ on $X \times [-\frac{1}{2}, \frac{1}{2}]$). On the level of stable tangent bundles this is immediate since the composition covers the folding map.
  \begin{equation}
    \tau^S(X \times S^1) \hookrightarrow \tau^S(X \times D^2) \xrightarrow{p} \tau^S(X \times I)
  \end{equation}
  $p$ is projection to the first coordinate $p: D^2 \ra D^1$.

  To destabilize we must continuously select a nonzero section in the fiber of $\tau^S(X \times D^2)$, $R^{d+k}$, for $k \geq 2$, $d = \dim(X)$. There is no obstruction as long as

  \[
  d+1 = \dim(\text{base}) < \dim(\text{fiber}) = d+k
  \]

  \noindent
  so $f$ may be covered by a tangent bundle injection. (Actually we only need to cover $f^-$ whose source has homotopy-dimension $d$, so there is room to spare meeting the hypotheses of immersion theory.) In any case, the immersion $e$ exists.

  Concretely, one may supply the tangential information a factor at a time: on the $X$-factor cover the inclusion by the identity map on tangent bundles, $\Id: \tau(X) \ra \tau(X)$. On the second factor $S^1 \ra [-1, 1]$ can be covered in only one way consistent with orientations: $\frac{\de}{\de \te} \ra \frac{+\de}{\de x}$.

\begin{figure}[ht]
  \centering
  \hspace{3em}
  \begin{tikzpicture}[scale = 1.5]
    \draw  (0,1.5) ellipse (0.9 and 0.15);
    \draw (-0.9,1.5) -- (-0.9,-2.1);
    \draw (0.9, 1.5) -- (0.9, -2.1);
    \draw (-0.9, -2.1) arc (180:360:0.9 and 0.25);
    \draw (-0.9, -1) arc (180:250:0.9 and 0.25);
    \draw (0.9, -1) arc (0:-52:0.9 and 0.25);
    \draw (-0.9, 0.3) arc (180:250:0.9 and 0.25);
    \draw (0.9, 0.3) arc (0:-49:0.9 and 0.25);

    \draw (-0.32, 0.06) to [out = 90, in = 180] (0.1,0.8) to [out = 0, in = 90] (0.6,-0.6) to [out = -90, in = 0] (0.1,-1.7) to [out = 180, in = -90] (-0.3, -1.23);
    \draw (-0.15, 0.05) to [out = 90, in = 180] (0.1, 0.6) to [out = 0, in = 90] (0.43, -0.6) to [out = -90, in = 0] (0.1, -1.55) to [out = 180, in = -90] (-0.15, -1.24);
    \draw (0.45, 0.09) arc (-60:-100:0.9 and 0.25);
    \draw (0.42, -1.2) arc (-60:-98:0.9 and 0.25);
    \draw [dashed] (-0.9, -1) arc (180:65:0.9 and 0.25);
    \draw [dashed] (0.9, -1) arc (0:50:0.9 and 0.25);
    \draw [dashed] (-0.9, 0.3) arc (180:100:0.9 and 0.25);
    \draw [dashed] (0.9, 0.3) arc (0:55:0.9 and 0.25);

    \node at (1.2,1.5) {$1$};
    \node at (1.2,0.3) {$\frac{1}{2}$};
    \node at (1.2,-0.5) {0};
    \node at (1.2,-1.5) {$-\frac{1}{2}$};
    \node at (1.2,-2.1) {$-1$};
    \node at (-0.2,-0.4) {image $e$};

    \draw [->] (1.9,-0.2) -- (0.5, -0.4);
    \node at (2.8,0) {band thickness};
    \node at (2.8,-0.4) {large w.r.t. $R$};
    \node at (-2.2,0) {picture of $e$};
    \node at (-2.2,-0.3) {for the case};
    \node at (-2.2,-0.6) {$X = S^1$};
  \end{tikzpicture}
  \caption{}
\end{figure}

  Now we follow \cite{Hastings2013} (see text above and below his line (53)), pull back $\beta$ along $e$ to $\widetilde{\beta}$ and then "heal the puncture." To do this define $\mathcal{B}$ to be the injective image under $\widetilde{\beta}$ of the full matrix algebra $\mathcal{A}$ of d.o.f. not within distance $R$ of the "the edge of the immersion." $\mathcal{A}$ and $\mathcal{B}$ are both simple algebras and consequently (by Wedderburn's Theorem) sit as tensor subfactors inside the entire pulled back algebra $\mathcal{C}$. Now extend $\widetilde{\beta}: \mathcal{A} \ra \mathcal{B}$ by an arbitrary isomorphism $\beta^\infty: \frac{\mathcal{C}}{\mathcal{A}} \ra \frac{\mathcal{C}}{\mathcal{B}}$ and set
  \begin{equation}
    \bar{\beta} = \widetilde{\beta} \otimes \beta^\infty
  \end{equation}
  to be the healed algebra endomorphism over $X \times S^1$. $\bar{\beta}$ and its unwrapping $\alpha$ satisfy the requirement of Theorem \ref{extend}.
\end{proof}

We now make a comment regarding uniqueness of $\alpha$, which is mostly about how little we know. One might hope to prove that any two periodic extensions $\alpha_0$ and $\alpha_1$ can be suitably stabilized, and then deformed (site separations must be stretched or compressed if $\alpha_0$ and $\alpha_1$ have distinct periods) one into the other. Possibly, but we cannot prove this. However, if one restricts attention to $\alpha_0$ and $\alpha_1$ coming only from the construction we have just given, using two different immersions $e_0$ and $e_1$ and perhaps slightly different cut offs at infinity, then we may reference a topological fact leading to uniqueness under the corresponding notion of "deformation."

\textbf{Topological fact}: Immersion theory does not merely produce immersion in specific homotopy classes but when the hypotheses are satisfied produces a weak homotopy equivalence of spaces

\[
\{\text{immersion of $A$ into $B$}\} \xrightarrow[\cong]{\text{differential}} \{\text{maps $A \ra B$ covered by a tangent bundle injection}\}
\]

Since covering $f$ by bundle data was canonical, we were not required to make any choices, any two immersions $e_0$ and $e_1$ will not only be homotopic to $f$ but regularly homotopic to each other. This regular homotopy translates to a 1-parameter family of partially defined algebraic automorphisms $\widetilde{e}_t$, $0 \leq t \leq 1$. We say "partially defined" because of the usual problem of degrees of freedom near the (moving) edge of the immersion. "Healing" the puncture now can be done w.r.t. a cover $\{\text{U}_i\}$ of $[0,1]$ and on intersections $\text{U}_i \cap \text{U}_j$ we will need to trim to smaller algebras $\mathcal{A}_i \cap \mathcal{A}_j$ and $\mathcal{B}_i \cap \mathcal{B}_j$, with $\mathcal{A}_k$ and $\mathcal{B}_k$ associated to $U_k$, $k = i$ or $i$. This leads to some suitable notion of uniqueness for the extension $\alpha$. One would like to know more.

\section{Algebraic Methods and Triviality in Two Dimensions}
\label{algebraicsection}
We now apply algebraic methods in to two dimensional QCA.
In this section, we will often refer to QCA defined by a specific graph $G$ rather than considering an abstract control space.

To classify QCA, we begin in subsection \ref{ba} by constructing a certain algebra that we call a boundary algebra.  This boundary algebra in fact is a special case of the support algebra considered previously, but the alternative definition here makes it simpler to consider certain properties of this algebra.
We then define some basic properties of this algebra, defining properties that we call "visibly simple" and "locally factorizable".  This subsection considers an arbitrary graph $G$.

In subsection \ref{gnvw} we use the properties of the boundary algebra to give a brief definition of the GNVW index; of course, we have already defined this previously but this subsection shows how to define it in terms of the boundary algebra considered here.  This subsection considers one-dimensional QCA only.

Then, in subsection \ref{vsstruct} we further consider some properties of visibly simple algebras, and give a decomposition of such algebras defined on a "linear" one-dimensional graph (i.e., sites are labelled by integers $i$ with ${\rm dist}(i,j)=|i-j|$).
Then, in subsection
\ref{lastruct}, we consider visibly simple, locally factorizable algebras defined on a "circular" one-dimensional graph (i.e., a ring of sites labelled by integers with ${\rm dist}(i,j)=1$ if $i=j\pm 1$ modulo $L$, where $L$ is the number of sites); this is relevant to the case that the algebra is the boundary algebra of a two-dimensional subset.
In subsection \ref{nolo}, we use this to show that in two dimensions we can blend between any QCA $\alpha$ on some (slightly smaller) subset and the identity QCA $\Id$ far away, thus answering the question of blending in two-dimensions.

In subsection \ref{noho} we consider stable equivalence between QCA in two dimensions, showing that any QCA in two-dimensions is stably equivalent to some QCA which acts as the identity everywhere except on some lower dimensional set, answering the question of local equivalence.
\label{algebrasection}
\subsection{Construction of Boundary Algebra}
\label{ba}
We begin by defining a certain {\it boundary algebra} and exploring some of its properties that will hold for an arbitrary graph $G$.

We need a general result.
\begin{lemma} Suppose a QCA $\alpha$ has range $R$.  Then $\alpha^{-1}$ also has range $R$.
\begin{proof}
Let $O_x$ be supprted on a site $x$.  Let $O_y$ be supported on some site $y$ with ${\rm dist}(x,y)>R$.  Then,
$[\alpha(O_y),O_x]=0$.  Hence, $[O_y,\alpha^{-1}(O_x)]=0$.  Since this holds for all sites $y$ with ${\rm dist}(x,y)>R$, $\alpha^{-1}$ has range $R$.
\end{proof}
\end{lemma}

For any site $x$, let $b_R(x)$ be the set of sites within distance $R$ of $x$, i.e., the set of sites $y$ such that ${\rm dist}(x,y)\leq R$.

For any set $S$, let ${\mathcal A}(S)$ be the algebra of operators supported on $S$.

\begin{lemma}
\label{1dalgdef}
Let $F\subset V$ be some set of sites.
Let $\alpha$ be a QCA with range $R$.
Let $\Int(F)$ denote the set of sites $x$ such that $b_R(x)\subset F$; we call this the {\it interior}.  Note that the interior depends on $R$; indeed, in many cases a dependence on $R$ will be implicit in our definitions.
Let $\Ext(F)$ denote the set of sites $x$ such that $b_R(x) \cap F=\emptyset$; we call this the {\it exterior}.
Finally, let $\Bd(F)$ denote $V\setminus (\Int(F)\cup \Ext(F))$; we call this the {\it boundary}.

Then, the algebra $\alpha({\mathcal A}(F))$
is generated by algebras on disjoint sets:
\be
\alpha({\mathcal A}(F)=<{\mathcal A}({\Int(F)}),\cPalphaF>,
\ee 
where $\cPalphaF$ is supported on $\Bd(F)$
and where $<\ldots,\ldots>$ denotes an algebra generated by other algebras.
\begin{proof}
By the definition of the range of $\alpha$, $\alpha({\mathcal A}(F))$ is supported on $\Int(F) \cup \Bd(F)$.

Similarly, $\alpha^{-1}({\mathcal A}(\Int(F)))$ is supported on $F$ so
$\alpha^{-1}({\mathcal A}(\Int(F))) \subset {\mathcal A}(F)$.
So ${\mathcal A}(\Int(F)) \subset \alpha({\mathcal A}(F))$.
Since
${\mathcal A}({\Int(F)})$ has trivial center, this means that $\alpha({\mathcal A}(F))$ is generated by ${\mathcal A}({\Int(F)})$ and by the commutant of
${\mathcal A}({\Int(F)})$ in $\alpha({\mathcal A}(F))$.
(Recall that given two algebras ${\mathcal X},{\mathcal Y}$, with ${\mathcal X}$ a subalgebra of ${\mathcal Y}$, the commutant of ${\mathcal X}$ in ${\mathcal Y}$ is
defined to be the set of elements of ${\mathcal Y}$ that commute with ${\mathcal X}$.)
This commutant is supported on $\Bd(F)$ and is the algebra we call $\cPalphaF$.
\end{proof}
\end{lemma}

We now further characterize $\cPalphaF$ in the next two subsections, defining properties that we call visibly simple
and locally factorizable.
\subsubsection{Locally Factorizable}

\begin{defin}
\label{llf}
An algebra ${\mathcal A}$ is an $l$-locally-factorizable algebra if given any operator $O\in {\mathcal A}$ supported on $T_1 \cup T_2$ for any two disjoints sets $T_1,T_2$ with ${\rm dist}(T_1,T_2)>l$, if $O$ is decomposed using a singular value decomposition
using the Hilbert-Schmidt inner product as $O=\sum_\beta A(\beta) O_1^\beta O_2^\beta$, for some discrete index $\beta$ and complex scalars $A(\beta) \neq 0$
with $O_i^\beta$ supported on $T_i$, then
$O_i^\beta$ is in ${\mathcal A}$ for all $i,\beta$.
\end{defin}
An equivalent definition is that the support algebra of $O$ on $T_1$ and on $T_2$ are both in ${\mathcal A}$.

To clarify this definition, suppose that ${\mathcal A}$ is generated by two algebras ${\mathcal A}_1$ and ${\mathcal A}_2$, with
${\mathcal A}_1$ supported on some set $S_1$ with diameter $l$ and ${\mathcal A}_2$ supported on some set $S_2$ also with
diameter $l$, with $S_1,S_2$ disjoint and ${\rm dist}(S_1,S_2)$ arbitrarily large (for example, the distance may be much larger than $l$).  In this case, ${\mathcal A}$ is $l$-locally factorizable.  Proof: let $O$ in $A$.  Without loss of generality assume $T_1\subset S_1,T_2\subset S_2$.  Then, $O_1^\beta=A(\beta)^{-1} {\rm tr}_{T_2}(O (O_2^\beta)^\dagger)$.
However, $O=\sum_\alpha P_1^\alpha P_2^\alpha$ with $P_1\in {\mathcal A}_1$ and $P_2\in {\mathcal A}_2$.
Hence, $$O_1^\beta=A(\beta)^{-1} \sum_\alpha P_1^\alpha {\rm tr}_{T_2}\Bigl(P_2^\alpha (O_2^\beta)^\dagger\Bigr).$$
The trace is a scalar, so the result follows.

As another example, consider a line of sites labelled by integers $i$ with $1 \leq i \leq L$, and consider the algebra generated by $X_1 X_L$ and $Z_1$.  This is a simple algebra and it is $(L-1)$-locally factorizable but it is not $(L-2)$-locally factorizable, as given $T_1=\{1\}$ and $T_2=\{L\}$ with ${\rm dist}(T_1,T_2)=L-1$, the operator $X_1 X_L$ is supported on $T_1 \cup T_2$ but does not obey the requirements of the definition for $l=L-2$.

\begin{remark}
If ${\mathcal A}$ is $l$-locally factorizable, and an operator $O\in {\mathcal A}$ is supported on $T_1 \cup T_2 \cup \ldots$ for any finite number of disjoint sets
$T_i$ with ${\rm dist}(T_i,T_j)>l$ for $i \neq j$, then $O$ can be written as a sum of products, $O=\sum_\beta O_1^\beta O_2^\beta \ldots$, with each $O_i^\beta$ supported on $T_i$ and $O_i^\beta\in {\mathcal A}$ and pairwise orthogonal as $\beta$ varies.  This follows by induction from the definition.
\end{remark}

\begin{lemma}
\label{llfundercirc}
Let $\alpha$ be a QCA with range $R$.  Let ${\mathcal A}$ be an $l$-locally-factorizable algebra.  Then $\alpha({\mathcal A})$ is an $l+2R$-locally-factorizable algebra.
\begin{proof}
Consider an operator $O$ in $\alpha({\mathcal A})$.
Decompose $O=\sum_\beta A(\beta) O_1^\beta O_2^\beta$ as in definition \ref{llf} with ${\rm dist}(T_1,T_2) > l+2R$.  Then,
$\alpha^{-1}(O)=\sum_\beta A(\beta) \alpha^{-1}(O_1^\beta) \alpha^{-1}(O_2^\beta)$.
Let $S_1,S_2$ be the set of sites within distance $R$ of $T_1,T_2$, respectively, so that $\alpha^{-1}(O_i^\beta)$ is supported on $S_i$.
Then, since ${\rm dist}(S_1,S_2) > l$ and ${\mathcal A}$ is $l$-locally-factorizable, $\alpha^{-1}(O_i^\beta)\in {\mathcal A}$ for all $i,\beta$ and
so $O_i^\beta\in \alpha({\mathcal A})$.
\end{proof}
\end{lemma}

\begin{lemma}
\label{llfsubalg}
Let algebra ${\mathcal A}$ be generated by a finite number of algebras ${\mathcal A}_1, {\mathcal A}_2, \ldots$, with ${\mathcal A}_i$ supported on some set $S_i$,
with $S_i \cap S_j = \emptyset$ for $i \neq j$.  Assume that ${\mathcal A}$ is an $l$-locally-factorizable algebra.  Then, each ${\mathcal A}_i$ is also an $l$-locally-factorizable subalgebra.
\begin{proof}
We consider the case that ${\mathcal A}$ is generated by two algebras, ${\mathcal A}_1, {\mathcal A}_2$; the general case follows by induction.

Consider an operator $O \in {\mathcal A}_1$ supported on sets $T_1 \cup T_2$ with ${\rm dist}(T_1,T_2)>l$.  Without loss of generality, we can assume that $T_1,T_2 \subset S_1$.  Since ${\mathcal A}$ is $l$-locally-factorizable,
$O = \sum_{\beta} A(\beta) O_1^\beta O_2^\beta$ with $O_1^\beta,O_2^\beta$ supported on $T_1,T_2$.  Then, $O_1^\beta,O_2^\beta\in {\mathcal A}_1$
since $T_1,T_2\subset S_1$.
\end{proof}
\end{lemma}

\begin{thm}
\label{thmfac}
The algebra $\cPalphaF$ in theorem \ref{1dalgdef} is a $2R$-locally-factorizable algebra.
\begin{proof}
${\mathcal A}(F)$ is a $0$-locally-factorizable algebra.  By lemma \ref{llfundercirc}, $\alpha({\mathcal A}(F))$ is a $2R$-locally-factorizable algebra.  By lemma \ref{llfsubalg}, the claim follows.
\end{proof}
\end{thm}

\subsubsection{Visibly Simple Algebras}
We now define
\begin{defin}
Let ${\rm dist}(.,.)$ be some metric.  We say that an algebra of operators ${\mathcal A}$ is an $l$-visibly simple algebra
if given any operator $O \in {\mathcal A}$, and given any site $x$ in the support of $O$, there is an operator $P\in {\mathcal A}$ which is supported on the
set of sites within distance $l$ of $x$ such that $[O,P] \neq 0$.

In this definition, when we refer to the "support of $O$", we mean the minimal set such that $O$ is supported on that set; we say that a scalar is supported on the empty set and so no sites are in the support of a scalar.
\end{defin}
Note that any $l$-visibly simple algebra has trivial center.
\begin{remark}
The term "visibly simple" is used because it suggests that one can "see" that $O$ is not a central element of ${\mathcal A}$ by just "looking" at some small set of sites.
\end{remark}

We give three examples of algebras to clarify this definition.
First, let $S$ be any set of sites and use any metric and let ${\mathcal A}={\mathcal A}(S)$.  Then, ${\mathcal A}$ is a $0$-visibly simple algeba, because if $x$ is in the support of an operator $O\in {\mathcal A}$, then there is an operator supported on $x$ which does not commute with $O$ and ${\mathcal A}$ contains all operators supported on $S$ and so it contains that operator supported on $x$.
Second, consider a set $S$ of sites labeled by integers $i$ with $1 \leq i \leq L$.  Let there be a qubit on each site.  Let ${\mathcal A}$ be the algebra generated by operators $X_i X_{i+1}$ for $i$ odd and $Z_i Z_{i+1}$ for $i$ even, where $X_i$ denotes the Pauli $X$-operator on site $i$ and $Z_i$ denotes the Pauli $Z$-operator on site $i$.
Let $L$ be even and use periodic boundary conditions, identifying $Z_{L+1}$ with $Z_1$, and use any metric.
Then, ${\mathcal A}$ is not an $l$-visibly simple algebra for any $l$ because it has non-trivial central elements generated by $X_1 X_2 X_3 X_4 \ldots$ and $Z_1 Z_2 Z_3 Z_4\ldots$
Thirdly, again let $S$ have sites labeled by integers $i$ with $1 \leq i \leq L$.  Let there be a qubit on each site.  Let ${\mathcal A}$ be the algebra generated by operators $X_i X_{i+1}$ for $i$ odd and $Z_i Z_{i+1}$ for $i$ even.
However, let $L$ be {\it odd} and $\geq 3$ and use open (i.e., not periodic) boundary conditions, defining ${\rm dist}(i,j)=|i-j|$.
Then, ${\mathcal A}$ is an $(L-1)$-visibly simple algebra but is not an $(L-2)$-visibly simple algebra.
To see the first statement, note that every operator is supported within distance $L-1$ of every site so the first statement follows from the fact that this algebra has trivial center.
To see the second statement, consider the operator $X_1 X_2 \ldots X_{L-1}$.
This operator commutes with all of the generators except for $Z_{L-1} Z_L$ and indeed it commutes with all operators in the algebra supported on sites $1,...,L-1$.
Hence, taking $x$ to be site $1$ (which is in the support of this operator), there is no operator in ${\mathcal A}$ supported within distance $L-2$ of $x$ which commutes with this operator.

\begin{lemma}
\label{llstrcorr}
Let ${\mathcal A}$ be an $l$-visibly simple algebra.  Then, for any operator $O$, with $O$ not necessarily in ${\mathcal A}$, if there exists an
$X\in {\mathcal A}$ supported on some set $S$ such that
${\rm tr}(XO){\rm tr}(I) \neq {\rm tr}(X){\rm tr}(O)$, then there is some operator $P \in {\mathcal A}$ which is supported on the
set of sites within distance $l$ of $S$ such that $[O,P] \neq 0$.
\begin{proof}
Since ${\mathcal A}$ is $l$-visibly simple it has trivial center.  So, the algebra of all operators is generated by ${\mathcal A}$ and some other algebra ${\mathcal C}$ which also has trivial center
such that ${\mathcal A},{\mathcal C}$ commute with each other.  So, $O$ can be decomposed as sum of products, $O=\sum_\alpha A_\alpha C_\alpha$, with $A_\alpha\in {\mathcal A}$ and $C_\alpha\in {\mathcal C}$.
Then, ${\rm tr}(XO)=\sum_\alpha {\rm tr}(X A_\alpha) {\rm tr}(C_\alpha)/{\rm tr}(I)$ and ${\rm tr}(O)=\sum_\alpha {\rm tr}(A_\alpha) {\rm tr}(C_\alpha)/{\rm tr}(I)$.  Hence, for some $\alpha$, ${\rm tr}(X A_\alpha){\rm tr}(I) \neq {\rm tr}(X){\rm tr}(A_\alpha)$.
Hence, for some $\alpha$, $A_\alpha$ has support on set $S$.
So, by the assumption that ${\mathcal A}$ is $l$-visibly simple, there is a $P\in {\mathcal A} $ which is supported on the
set of sites within distance $l$ of $S$ such that $[A_\alpha,P] \neq 0$ for the given $\alpha$.
Do this decomposition so that ${\rm tr}(C_\alpha^\dagger C_\beta)=\delta_{\alpha,\beta}$, where $\delta$ is the Kronecker $\delta$-function.
Then, since $[A_\alpha,P]\neq 0$, it follows that $[O,P] \neq 0$.
\end{proof}
\end{lemma}

\begin{lemma}
\label{llsundercirc}
Let $\alpha$ be a QCA with range $R$.  Let ${\mathcal A}$ be an $l$-visibly simple algebra.  Then $\alpha({\mathcal A})$ is an $l+2R$-visibly simple algebra.
\begin{proof}
Consider an operator $O$ in $\alpha({\mathcal A})$ with a site $x$ in its support.  Then, there is a site $y$ in the support of $\alpha^{-1}({\mathcal A})$ with ${\rm dist}(x,y) \leq R$.  Then, since ${\mathcal A}$ is $l$-visibly simple, there is an operator $P\in {\mathcal A}$ supported on the set of sites within distance $l$ of site $y$ such that $[\alpha^{-1}(O),P] \neq 0$.  Hence, $[O,\alpha(P)] \neq 0$, and $\alpha(P)$ is supported on the set of sites within distance $l+R$ of site $y$ which is within distance $l+2R$ of site $x$.
\end{proof}
\end{lemma}

\begin{lemma}
\label{llssubalg}
Let algebra ${\mathcal A}$ be generated by a finite number of algebras ${\mathcal A}_1, {\mathcal A}_2, \ldots$, with ${\mathcal A}_i$ supported on some set $S_i$,
with $S_i \cap S_j = \emptyset$ for $i \neq j$.  Assume that ${\mathcal A}$ is an $l$-visibly simple algebra.  Then, each ${\mathcal A}_i$ is also an $l$-visibly simple subalgebra.
\begin{proof}
We consider the case that ${\mathcal A}$ is generated by two algebras, ${\mathcal A}_1, {\mathcal A}_2$; the general case follows by induction.

Consider any operator $O$ in ${\mathcal A}_1$ and consider any site $x$ in the support of $O$.  Then, since ${\mathcal A}$ is $l$-visibly simple, there is an operator $P$ in ${\mathcal A}$ such that $[O,P]\neq 0$ and $P$ is supported on the set of sites within distance $l$ of $x$.  $P$ can be decomposed as $P=\sum_\beta P_1^\beta P_2^\beta$, for some discrete index $\beta$ with $P_1^\beta \in {\mathcal A}_1$ and $P_2^\beta\in {\mathcal A}_2$ and 
where the $P_2^\beta$ are chosen orthonormal using the Hilbert-Schmidt inner product and $$P_1^\beta={\rm tr}_2(P (P_2^\beta)^\dagger).$$
 Then, there is some $\beta$ such that $[P_1^\beta,O]\neq 0$ and since $P$ is supported within distance $l$ of $x$, $P_1^\beta$ is supported within distance $l$ of $x$.
\end{proof}
\end{lemma}

Hence,
\begin{thm}
\label{thmls}
The algebra $\cPalphaF$ in lemma \ref{1dalgdef} is a $2R$-visibly simple algebra.
\begin{proof}
${\mathcal A}(F)$ is a $0$-visibly simple algebra.  By lemma \ref{llsundercirc}, $\alpha({\mathcal A}(F))$ is a $2R$-visibly simple algebra.  By lemma \ref{llssubalg}, the claim follows.
\end{proof}
\end{thm}

\begin{remark}
We will in fact see later as a corollary in theorem \ref{vsislf} that every
$l$-visibly simple algebra is $O(l)$-locally factorizable.
\end{remark}

\subsubsection{Boundary With Disjoint Components}
In some applications, it will be possible to write the boundary $\Bd(F)$ as union of separate components.
More precisely, suppose there are sets $\Bd_i(F)$,
where $i$ is some discrete index taking a finite set of possible values such that
\be
\cup_i \Bd_i(F)=\Bd(F),
\ee
and such that
${\rm dist}(\Bd_i(F),\Bd_j(F)) \geq 2R$ for $i\neq j$.

Then, using local factorizability of $\cPalphaF$, the algebra $\cPalphaF$ is generated by algebras $\cPalphaF_i$, with $\cPalphaF_i$ being the algebra of operators in $\cPalphaF$ which are supported on $\Bd_i$.
Each of these $\cPalphaF_i$ is $2R$-visibly simple and $2R$-locally-factorizable by 
theorem \ref{thmfac} and
lemma
\ref{llfsubalg}.

\subsection{GNVW Index}
\label{gnvw}
We now show how to define the GNVW index for a one-dimensional QCA $\alpha$ using the algebras $\cPalphaF_i$ defined above.  Although we derive it in a slightly different way than in the original reference, the resulting is the same index as previously defined.
Consider a line of sites indexed by an integer $i$, with a Hilbert space dimension $d_i$ on each site.

We use notation $[i,j]$ to denote a set of sites $\{i,i+1,\ldots,j\}$ for $i\leq j$.
Let $F$ be any interval of sites $[i,j]$ which is sufficiently long that we can write the boundary as separate components which we call $\Bd_L(F)$ and $\Bd_R(F)$.
In this case $\Bd_L(F)=[i-R,i+R-1]$ and $\Bd_R(F)=[j-R+1,j+R]$, so we need $j-i>4R-2$.

\begin{remark}
In fact, the GNVW index can be defined using shorter intervals than this.  In the original GNVW paper, by coarse-graining $R$ was taken equal to $1$ and intervals with $j-i=1$ were used.  That required a slightly more careful treatment of commutation of algebras.  The larger intervals that we use here may help make the commutation slightly more clear.
\end{remark}

Since $\cPalphaF_R$ is simple, it is isomorphic to a $d$-by-$d$ matrix algebra.  We define $d_R(\alpha,F)$ to equal this dimension $d$, i.e.,$\cPalphaF_R$ has dimension $d_R(\alpha,F)^2$.
Similarly, we define $d_L(\alpha,F)$ so that $\cPalphaF_L$ has dimension $d_L(\alpha,F)^2$.
We now define the GNVW index to equal
\be
\Ind(\alpha,F)=\log\Bigl(\frac{d_R(\alpha,F)}{\prod_{i \in \Bd_R(F) \cap F} d_i}\Bigr).
\ee
This GNVW index $\Ind(\alpha,F)$ is indeed equal to the logarithm of the index defined in Ref.~\cite{Gross2012}.

Let us explain why, given the algebra $\cPalphaF_R$, we are naturally led to make this choice for the index (up to choices such as whether or not to take the logarithm).
First,
the only property of $\cPalphaF_R$ that is invariant under unitary transformations supported on $\Bd_R(F)$ is the dimension $d_R(\alpha,F)$, so the index must be a function of $d_R(\alpha,F)$.
Since we wish the index to be invariant under stabilization, we divide by $\prod_{i \in \Bd_R(F) \cap F} d_i$; one may verify that then $\Ind(\alpha \otimes \Id,F)=\Ind(\alpha,F)$ after tensoring with any choice of additional degrees of freedom.
Then, the choice of taking the logarithm is done so that
the index to equals $0$ if $\alpha=\Id$ and is additive under tensor product: $\Int(\alpha \otimes \beta,F)=\Ind(\alpha,F) \otimes \Int(\beta,F)$ for
any choice of $\alpha,\beta$.

One may wonder whether the algebra $\cPalphaF_L$ can be used to define a different index.
However, since the dimension of $\alpha({\mathcal A}(F))$ is the same as that of ${\mathcal A}(F)$, we have
\be
d_R(\alpha,F) d_L(\alpha,F) \prod_{i \in \Int(F)} d_i=\prod_{i \in F} d_i.
\ee
Hence,
\be
\Ind(\alpha,F)=\log\Bigl(\frac{\prod_{i \in \Bd_L(F) \cap F} d_i}{d_L(\alpha,F)}\Bigr).
\ee
So, we are led to the same index (up to a sign) using the two different algebras $\cPalphaF_L,\cPalphaF_R$.

Finally, given any two intervals, $F,G$ with $F=[i,j]$ and $G=[j+1,k]$ with $j-i>3R$ and $k-j>3R$ we can define indexes $\Ind(\alpha,F)$,$\Ind(\alpha,G)$.
Then, $\Bd_R(F)=\Bd_L(G)$ and
$d_R(\alpha,F) d_L(\alpha,G)=\prod_{i \in \Bd_R(F)} d_i$ so
\be
\Ind(\alpha,F)=\Ind(\alpha,G),
\ee
showing the equivalence of the GNVW index for different intervals along the chain as shown in Ref.~\cite{Gross2012}.
Since the GNVW index is independent of choice of interval $F$, we may drop the $F$-dependence from the notation and just write the GNVW index as $\Ind(\alpha)$.

\subsection{Structure of Visibly Simple Algebras}
\label{vsstruct}
We now prove a structure theorem \ref{lltoy} for visibly simple algebras on arbitrary control spaces.  By induction, in the case of "linear" graph as a control space,
this gives us theorem \ref{llinduct}.
In the next subsection, we prove theorem \ref{llsllfthm}, where we use a "circular" graph as a control space.

In the following, the metric used to measure distances between sites is arbitrary:
\begin{thm}
\label{lltoy}
Suppose an algebra ${\mathcal A}$ of operators  is $l$-visibly simple.
Let $S$ be any set of sites.
Then, ${\mathcal A}$ is generated by algebras ${\mathcal A}_1,{\mathcal A}_2$ constructed below where ${\mathcal A}_1$ is supported within distance $2l$ of $S$,
where ${\mathcal A}_2$ is supported within distance $l$ of the complement of $S$, and
where ${\mathcal A}_1,{\mathcal A}_2$ commute with each other and have trivial center.
\begin{proof}
Let $\cB_1$ be the algebra of operators in ${\mathcal A}$ which are supported on $S$.

We begin by considering a special case before considering the general case.  The special case that we consider is that $\cB_1$ is a simple algebra.  If so, then define ${\mathcal A}_1=\cB_1$
and define ${\mathcal A}_2$ to be the commutant of ${\mathcal A}_1$ in ${\mathcal A}$.  Then, these algebras obey the claims of the lemma.  The only claim that is non-obvious is that
${\mathcal A}_2$ is supported within distance $O(l)$ of the complement of $S$.  This claim follows from the assumption that ${\mathcal A}_1$ is visibly simple: let $O$ be an operator in ${\mathcal A}_2$.  If $O$ has support on some site which is distance greater than $l$ from the complement of $S$, then there is some operator $P$ supported on $S$ which does not commute with $O$.  Such an operator $P$ is in ${\mathcal A}_1$ by construction, giving a contradiction since ${\mathcal A}_2$ is the commutant of ${\mathcal A}_1$.

Now consider the general case that $\cB_1$ is not a simple algebra.  We will construct a simple algebra ${\mathcal A}_1$ such that $\cB_1\subset {\mathcal A}_1$ and such that
${\mathcal A}_1$ is supported within distance $O(l)$ of $S$.  We then take ${\mathcal A}_2$ to be the commutant of ${\mathcal A}_1$ in ${\mathcal A}$.  Then, the claims of the lemma follow: the only claim that is non-obvious is the support of ${\mathcal A}_2$, but since $\cB_1\subset {\mathcal A}_1$ this follows in the same way as in the special case above.

To construct ${\mathcal A}_1$ in this special case, define $\cB_2$ to be the subalgebra of ${\mathcal A}$ supported on sites within distance $l$ of $S$.
Any operator in the center of $\cB_2$ must then be supported on the complement of $S$: if $O$ is in the center of $\cB_2$ and there is some site $x$ in the support of $O$ which is
in $S$, then there is some operator $P$ which does not commute with $O$ with $P$ supported within distance $l$ of $x$, i.e., so that $P$ is supported within distance $l$ of $S$ so that $P\in \cB_2$ giving a contradiction.
Hence, for any operator $P_2$ in the center of $\cB_2$ and any operator $P_1$ in the center of $\cB_1$, the supports of $P_1,P_2$ are disjoint
and so
${\rm tr}(P_1 P_2)={\rm tr}(P_1) {\rm tr}(P_2)/{\rm tr}(I)$.
Then, the existence of ${\mathcal A}_1$ follows from lemma \ref{inbetweenlemma} below.
\end{proof}
\end{thm}
Remark: the condition ${\rm tr}(P_1 P_2)={\rm tr}(P_1) {\rm tr}(P_2)/{\rm tr}(I)$ may seem mysterious.  It can be informally understood as follows: regard the traces as describing the expectation value of some classical random variables; then, the condition says that these variables are uncorrelated with each other.

From this theorem we prove by induction the following theorem.  
We will use the term "linear graph" throughout to refer to a graph where the vertices can be labelled by integers $1,2,\ldots,J$ with edges $(j,j+1)$; this corresponds
to the case of an interval control space with fixed endpoints and scale.
Remark: in words, the theorem says that a visibly simple algebra on a linear graph is generated by a set of simple algebras, which commute with each other, and which have bounded support.  Clearly the converse holds: given any set of simple algebras, each of bounded support, which commute with each other, the algebra generated by that set of simple algebras is visibly simple.
\begin{thm}
\label{llinduct}
Consider a set of sites labelled by integers $1,2,...,J$ for some $J$, using a "linear" graph metric to measure distance where site $j$ is connected to sites $j \pm 1$ and consider a tensor product Hilbert space with a finite dimensional Hilbert space for each site.  Suppose an algebra ${\mathcal A}$ of operators  is $l$-visibly simple.
Then, ${\mathcal A}$ is generated by algebras ${\mathcal C}_i$ constructed below
where each ${\mathcal C}_i$
is an algebra of operators supported on some set $C_i$ of diameter $O(l)$.
Further, the algebras ${\mathcal C}_i$ all commute with each other and have trivial center.
Further, the support of $\cC_i$ is disjoint from the support of all $\cC_j$ for $i\neq j$ except possibly $\cC_{i\pm 1}$
\begin{proof}
This follows from theorem \ref{lltoy} by induction.  Pick $S$ to consist of the set of sites $1,2,\ldots,2l$, so that ${\mathcal A}_2$ is supported on sites $l+1,l+2,\ldots,J$
and ${\mathcal A}_1$ is supported on sites $1,\ldots,3l$.
Construct algebras ${\mathcal A}_1,{\mathcal A}_2$ from that theorem.  Set ${\mathcal C}_1={\mathcal A}_1$.

Call the original linear graph $G$.  We will construct a new linear graph $G'$.
Now, group sites $l+1,l+2,\ldots,3l$ into a new "supersite" which we label $1$.  Relabel the other sites by subtracting $3l-1$, so that site $3l+1$ becomes site $1$, site $3l+2$ becomes site $2$, and so on.
This gives a new linear graph with sites $1,\ldots,J'$ for some $J'=J-O(l)$.  We claim that ${\mathcal A}_2$ is $l$-visibly simple with respect to $G'$.
Indeed, consider some operator $O\in {\mathcal A}_2$.  Let $x'$ be in the support of $O$.  If $x'>l+1$, then $x=x'+3l-1>4l$ is in the support of $O$ on the original graph $G$.
Hence, since ${\mathcal A}$ is $l$-visibly simple, there is some operator $P\in{\mathcal A}$ which does not commute with $O$, such that $x$ is supported on sites within distance $l$ of $x$.
Since $x>4l$, this implies that the support of $P$ is disjoint from the support of ${\mathcal A}_1$, and so $P\in {\mathcal A}_2$.
Suppose instead that $x'\leq l+1$.  Then, $x=x'+3l-1\leq 4l$ is in the support of $O$ on the original graph $G$.
Hence, since ${\mathcal A}$ is $l$-visibly simple, there is some operator $P\in{\mathcal A}$ which does not commute with $O$, such that $x$ is supported on sites within distance $l$ of $x$.
This operator $P$ may not be in ${\mathcal A}_2$ (i.e., it may be some some of products of operators in ${\mathcal A}_1,{\mathcal A}_2$), but since ${\mathcal A}_1$ commutes with $O$,
this means that there is some operator $Q$ in ${\mathcal A}_2$ which does not commute with $O$ which is supported on the union of the support of ${\mathcal A}_1$ and the support of $P$.
So, $Q$ is supported on the set of sites $\{1,\ldots,5l\}$ in $G$.  On graph $G'$, this corresponds to sites $\{1,\ldots,2l+1\}$, and this set is within distance $l$ of $x$.

We then apply theorem \ref{lltoy} to algebra ${\mathcal A}_2$, calling the resulting algebras ${\mathcal A}'_1,{\mathcal A}'_2$.  We set ${\mathcal C}_2={\mathcal A}'_1$.  Proceeding inductively in this fashion gives the desired decomposition.
\end{proof}
\end{thm}

From this we have the corollary:
\begin{thm}
\label{vsislf}
Let ${\mathcal A}$ be an $l$-visibly simple algebra.  Then, ${\mathcal A}$ is an $O(l)$-locally factorizable algebra.
\begin{proof}
Let $O\in {\mathcal A}$ be supported on $T_1 \cup T_2$.

Define a linear graph by grouping all sites in $T_1$ into supersite $1$, and grouping all sites at distance $r$ from $T_1$ into supersite $r+1$ for each integer $r$.
Apply theorem \ref{llinduct} to construct algebras ${\mathcal C}_i$.  Let ${\mathcal X}$ be the algebra generated by all algebras ${\mathcal C}_i$
whose support includes supersites corresponding to $T_1$ (from the proof of the theorem, in fact this is only algebra ${\mathcal C}_1$) and let
${\mathcal Y}$ be the algebra generated by all
algebras ${\mathcal C}_i$
whose support includes supersites corresponding to $T_2$.

Let $X$ denote the support of ${\mathcal X}$ and let $Y$ denote the support of ${\mathcal Y}$.  Note that $X$ may be larger than $T_1$ and $Y$ may be larger than $T_2$.
Assume that
${\rm dist}(T_1,T_2)$ is sufficiently large compared to $l$ so that $X\cap Y=\emptyset$.
 
Decompose $O=\sum_\gamma B(\gamma) O_X^\gamma O_Y^\gamma$, where $X\in {\mathcal X}$ and $Y\in {\mathcal Y}$ and $B(\gamma)\neq 0$ is a scalar, using a singular value decomposition and Hilbert-Schmidt inner product.
Then, the set of $O_X^\gamma$ generates the support algebra of $O$ on the complement of $Y$.  Since $O$ is supported on $T_1 \cup T_2$ and $T_2\subset Y$, this means that the $O_X^\gamma$ are supported on $T_1$.
Hence, the support algebra of $O$ on $T_1$ is generated by the set of $X^\gamma$ and so it is in in ${\mathcal A}$, and similarly the support algebra of $O$ on $T_2$ is in ${\mathcal A}$.
\end{proof}
\end{thm}

Finally, we prove lemma \ref{inbetweenlemma}.  However, we first recall
 a structure theorem that is an application of the Artin-Wedderburn theorem.
We will use this structure theorem several times, both in this lemma and in the proof of theorem \ref{llsllfthm}.
Let ${\mathcal M}_n$ denote the algebra of $n$-by-$n$ matrices.
Consider any subalgebra ${\mathcal X}$ of the algebra of operators ${\mathcal Y}$ on a finite-dimensional Hilbert space.  Suppose this subalgebra ${\mathcal X}$ is closed under
adjoint and has an $m$-dimensional center, meaning that the center consists of operators of the form $z_1 P_1 + ... + z_m P_m$ for $z_a$ being complex scalars and $P_a$ being projectors
with $\sum_a P_a=I$.
Then, ${\mathcal X}$ is isomorphic to a direct sum of algebras ${\mathcal M}_{d_1},...,{\mathcal M}_{d_m}$.
Further,
up to a unitary transformation (acting on the Hilbert space of the entire system), the operators in this algebra
are of the form
\be
\label{Mform}
M=
\begin{pmatrix}
M_1 \\
& M_1 \\
&& ... \\
& & & M_2 \\
&&&& M_2 \\
&&&&&...\\
&&&&&&...\\
&&&&&&&M_m\\
&&&&&&&&...
\end{pmatrix}.
\ee
That is, there are some number of copies of a matrix $M_1$ which has dimension $d_1$-by-$d_1$, followed by some number of copies of a matrix $M_2$ which has dimension $d_2$-by-$d_2$, and so on.  Let $c_1,c_2,...$ denote the number of copies of $M_1,M_2,...$.  Note that $c_a \geq 1$; the form above shows $2$ or more copies in each case but one may have only a single copy.
That is, the operator is of the form
\be
\label{Mform2}
M=(M_1 \otimes I_{c_1}) \oplus (M_2 \otimes I_{c_2}) \oplus ...,
\ee
where $I_c$ is the $c$-by-$c$ identity matrix.

Now we prove
\begin{lemma}
\label{inbetweenlemma}
Let ${\mathcal X}\subset {\mathcal Z}$ be subalgebras of some matrix algebra.  Suppose that for any $P_X$ in the center of ${\mathcal X}$ and any $P_Z$ in the center of ${\mathcal Z}$ we have
${\rm tr}(P_X P_Z)={\rm tr}(P_X) {\rm tr}(P_Z)/{\rm tr}(I)$, where all traces are taken in the matrix algebra.
Then, there exists a simple algebra ${\mathcal Y}$ with ${\mathcal X}\subset {\mathcal Y} \subset {\mathcal Z}$.
\begin{proof}
First we
apply the structure theorem to ${\mathcal Z}$.  We write $d(Z)_1,d(Z)_2,\ldots$ rather than $d_1,d_2,\ldots$ for the dimensions for clarity, $P(Z)_1,P(Z)_2,\ldots$ for the projectors,
and $c(Z)_1,c(Z)_2,\ldots$ for the number of copies.  We write ${\mathcal M}_a$ to denote the algebra of $d(Z)_a$-by-$d(Z)_a$ matrices.

Since ${\mathcal X} \subset {\mathcal Z}$ we have a homomorphism from ${\mathcal X}$ to ${\mathcal M}_a$ for each $a$ in the obvious way: given an $O\in {\mathcal X}$, apply the structure theorem to write it in form Eq.~(\ref{Mform}), and then the image of $O$ is the matrix $M_a$.

Consider a given $a$.
We now apply the structure theorem to ${\mathcal X}$, writing $d(X)_1,d(X)_2,\ldots$ for the dimensions, $P(X)_1,P(X)_2,\ldots$ for the projectors, and $c(X)_1,c(X)_2,\ldots$ for the number of copies.  Then,
the image of this homomorphism is of the form
$(M_1 \otimes I_{C_{a,1}}) \oplus (M_2 \otimes I_{C_{a,2}}) \oplus \ldots ...$ where $M_b$ has dimension $d(X)_b$ and $C_{a,b}$ are some integers so
that $$\sum_b d(X)_b C_{a,b}=d(Z)_a.$$
We claim that there is some integer $N$ such that
\be
\label{fac}
C_{a,b}=\frac{d(Z)_a c(X)_b}{N}.
\ee
To see this, note that $C_{a,b} d(X)_b c(Z)_a={\rm tr}(P(X)_b P(Z)_a)$.
By assumption, this equals ${\rm tr}(P(X)_b) {\rm tr}(P(Z)_a)/{\rm tr}(I)=
d(X)_b c(X)_b d(Z)_a c(Z)_a/{\rm tr}(I)$, so
 so Eq.~(\ref{fac}) holds for $N={\rm tr}(I)$.

Let $\tilde c(X)_b=c(X)_b/{\rm gcd}(c(X)_1,c(X)_2,\ldots),$ and let $\tilde N=N/{\rm gcd}(c(X)_1,c(X)_2,\ldots)$.  Note that $\tilde N$ is
also an integer.
Then,
\be
\label{fac2}
C_{a,b}=\tilde c(X)_b \frac{d(Z)_a}{\tilde N}.
\ee
So, since $C_{a,b}$ is an integer and the greatest common divisor of the $\tilde c$ is $1$, $\frac{d(Z)_a}{\tilde N}$ is an integer.
So, for each $a$, the image of the homomorphism from ${\mathcal X}$ to ${\mathcal M}_a$ is of the form
$$\Bigl((M_1 \otimes I_{\tilde c(X)_1}) \oplus (M_2 \otimes I_{\tilde c(X)_2}) \oplus \ldots \Bigr) \otimes I_{\frac{d(Z)_a}{\tilde N}}.$$
Hence, taking a direct sum over $a$, each operator in ${\mathcal X}$ can be written, using the structure theorem for ${\mathcal Z}$, as
$$\oplus_a  \Bigl(\Bigl((M_1 \otimes I_{\tilde c(X)_1}) \oplus (M_2 \otimes I_{\tilde c(X)_2}) \oplus \ldots \Bigr) \otimes I_{\frac{d(Z)_a}{\tilde N}}\otimes I_{c(Z)_a} \Bigr).$$
Now consider the algebra of all operators of the form
$$\oplus_a  \Bigl(M \otimes I_{\frac{d(Z)_a}{\tilde N}} \otimes I_{c(Z)_a} \Bigr).$$
where $M$ is a matrix of dimension $d$-by-$d$ for
\be
d=\sum_a d(X)_b \tilde c(X)_b.
\ee
This is a simple subalgebra of ${\mathcal Z}$.  
We will define ${\mathcal Y}$ to be this subalgebra.
\end{proof}
\end{lemma}

\subsection{Structure of Visibly Simple, Locally Factorizable Algebras in One Dimension: Circular Graph}
\label{lastruct}
We now further characterize $\cPalphaF$. 
Our main result is the following structure theorem for visibly simple, locally factorizable algebras in the case that the sites are on a one-dimensional line with a "circular" metric as defined in the theorem.
This will be relevant to the case of QCA in two dimensions, as we explain later when we use this theorem.  This theorem is more complicated than the "linear" metric considered previously.

\begin{thm}
\label{llsllfthm}
Consider a set of sites labelled by integers $1,2,...,J$ for some $J$, using a "circular" graph metric to measure distance where site $j$ is connected to sites $j \pm 1 \, {\rm mod} \, J$ and consider a tensor product Hilbert space with a finite dimensional Hilbert space for each site.  Suppose an algebra ${\mathcal A}$ of operators  is $l$-visibly simple and $l$-locally factorizable.
Then, ${\mathcal A}$ is generated by algebras ${\mathcal C}_i$ constructed below
where each ${\mathcal C}_i$
is an algebra of operators supported on some set $C_i$ of diameter $O(l)$.
Further, the algebras ${\mathcal C}_i$ all commute with each other and have trivial center.
Further, the support of $\cC_i$ is disjoint from the support of all $\cC_j$ for $i\neq j$ except possibly $\cC_{i\pm 1}$; this $i\pm 1$ is taken modulo the number of ${\mathcal C}_i$
so that the last one may have support overlapping with that of ${\mathcal C}_1$.
\begin{proof}
  For any set $S$ of sites, define ${\mathcal A}_S$ to be
the algebra of operators in ${\mathcal A}$ which are supported on $S$.

We define a sequence of sets, $I_1,I_2,\ldots,I_K$, for some integer $K=O(J/l)$ with the following properties.
Each $I_j$ is a set of sites $\{1,2,\ldots,R_j\}$, where $R_j=j\Delta$
where $\Delta=Cl$ for some sufficiently large constant $C$ chosen later (for example, $C=4$ suffices).
We choose $K$ so that $\Delta \leq J-R_K \leq 2 \Delta$.
Note that $I_j \subset I_{j+1}$.
See figure \ref{intervalfig}.

\begin{figure}
\includegraphics[width=3in]{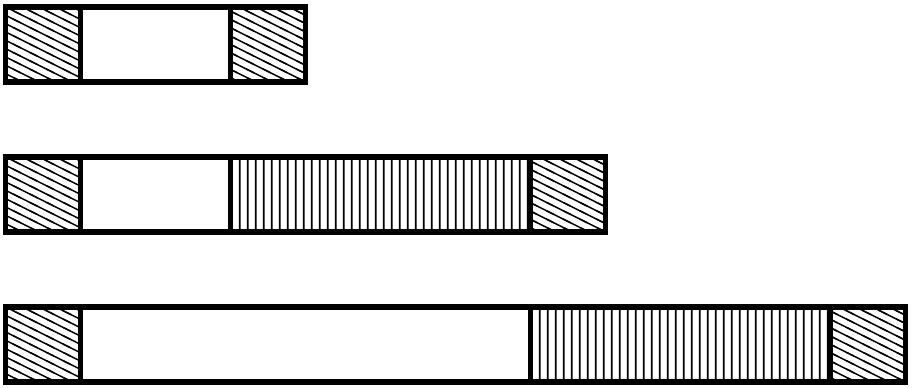}
\caption{Schematic illustration of intervals.  Top line is $I_1$, middle line is $I_2$, bottom line is $I_3$.  Largest rectangular box represents the intervals; the boxes with diagonal lines inside represent the left and right ends.  The extended right end includes the right end as well as the longer box with vertical lines inside.}
\label{intervalfig}
\end{figure}

Let ${\mathcal B}_j={\mathcal A}_{I_j}$ so that ${\mathcal B}_{j-1} \subseteq {\mathcal B}_j$.

{\it Proof in special case when ${\mathcal B}_j$ have trivial center---}
To give the idea of the proof, we consider the case in which all ${\mathcal B}_j$ have trivial center (later we consider the more general case).
In this case, define algebras $\cC_j$ as follows.  Let $\cC_1=\cB_1$.  Let $\cC_j$ for $j>2$ equal the commutant of $\cB_{j-1}$ in $\cB_j$.
That is, $\cC_j$ is equal to the set of elements of $\cB_j$ which commute with all elements of $\cB_{j-1}$.
Then, since all of these algebras $\cB_j$ are assumed to have trivial center, the algebras $\cC_j$ also have trivial center and
$\cC_1,\ldots,\cC_K$ generates $\cB_K$.
Finally, define algebra $\cC_{K+1}$ to equal to commutant of $\cB_K$ in ${\mathcal A}$ so that
$\cC_1,\ldots,\cC_{K+1}$ generate ${\mathcal A}$.

\begin{remark}
It might seem strange to the reader that we use a slightly different notation to define $\cC_{K+1}$ compared to the other algebras $\cC_j$ for $j\leq K$; we could have defined a set $I_{K+1}$ that is the set of all sites and then define $\cB_{K+1}={\mathcal A}_{I_{K+1}}$ and defined $\cC_{K+1}$ as the commutant of $\cB_K$ in $\cB_{K+1}$.  This is the same definition but using a unified notation for all $\cC_j$.  We choose not to do this because in the more general case we will want to treat them slightly differently.
\end{remark}

By construction, then, the algebras ${\mathcal C}_j$ all commute with each other and have trivial center.
It remains then to show the claims about the supports of the algebras. 
For any interval, $I_j$ define the "left end" of that interval to consist of the set of sites $\{1,\ldots,l\}$ and define the "right end" to consist of sites $\{R_j-l+1,\ldots,R_j\}$.  Both the left and right ends consist of $l$ sites.
Define the "extended right end" for $j>1$ to consist of sites
sites $\{R_{j-1}-l+1,\ldots,R_j\}$ so that the extended right end of $I_j$ consists of the right end of $I_{j-1}$ as well as all sites in $I_j\setminus I_{j-1}$.

 Consider algebra $\cC_j$ for $j\leq K$.
Every operator $O\in \cC_j$ is supported on $I_j$.  
We claim that for $j>1$ the operator $O$ is supported on the union of the left end and the extended right end of $I_j$
and that for $j>2$ the operator $O$ is supported on the extended right end of $I_j$.
First left us show that the operator $O$ is supported on the union of the left end and extended right end for $j>1$.
Indeed, if a site $x\in I_j$ is in the support of $O$ with $x$ not in the left or right end, then, by properties of a visibly simple algebra, there is an operator $P$ in ${\mathcal A}$ supported within the set of sites within distance $l$ of $x$ which does not commute with $O$.  However the support of such a $P$ is in $I_{j-1}$ and hence $P$ is in $\cB_{j-1}$, giving a contradiction.
Next, we show that the $O$ is supported on the extended right end of $I_j$ for $j>2$.
By assumption that ${\mathcal A}$ is locally factorizable, since $O$ is supported on the union of the left end and extended right end,
we can decompose
 $O=\sum_\beta A(\beta) O_L^\beta O_R^\beta$ with $O_L,O_R$ supported on the left end and extended right end respectively and $O_L^\beta,O_R^\beta\in {\mathcal A}$, and $A(\beta)$ being a scalar, respectively (here we use that the distance between left end and extended right end is greater than $l$ which requires that $C>3$).
Thus $O_L^\beta$ is in $\cB_1$.
Hence, since $\cB_1$ is simple, if $O_L^\beta$ is not a scalar then there is some operator $P\in \cB_1$ which does not commute with $O_L^\beta$.  Since the extended right end of $I_j$ is disjoint from $I_1$ for $j>2$ for large enough $C$ (again, $C>3$ suffices), this operator $P$ commutes with $O_R^\beta$ and so $P$ does not commute with $O$.
Hence, $O_L^\beta$ must be a scalar and so the support of $O$ is as claimed.

Then, since the extended right ends have diameters $O(l)$, the claim about the support of $\cC_j$ for $j>2$ follows (the same argument also works to bound the support of $\cC_{K+1}$).  For $j=1,2$, the claim about the support of $\cC_j$ follows from the fact that $I_j$ has diameter $O(l)$ for any given $j=O(1)$.

Thus, this proves the theorem under the assumption that all $\cB_j$ have trivial center.

{\it Proof in general case---}
Now we consider the more general case that the $\cB_j$ may have nontrivial center.

Note first that any operator $O$ in the center of $\cB_j$ is supported on the union of the left and right ends of $I_j$.  (Remark: here we indeed mean the right end, not the extended right end.)
Indeed, if not, then there would be some $x$ in the support of $O$ with $x$ not in the union of left and right ends such that the set of sites within distance $l$ of $x$ would be contained within $I_j$ and since ${\mathcal A}$ is visibly simple there would be some operator in ${\mathcal A}$ supported on $I_j$ which does not commute with $O$, contradicting the assumption that $O$ is in the center.
 
\begin{remark}
It seems natural to conjecture that the center of $\cB_j$ is generated by two algebras, one supported on the left end of $I_j$ and one supported on the right end of $I_j$.  In fact, this conjecture will follow from the theorem for sufficiently long intervals.  However, we will {\it not} assume that this conjecture holds.
\end{remark}

Before giving the proof, let us give a rough outline of what we will do.
We will define new algebras $\cD_j$ derived from $\cB_j$ so that the $\cD_j$ have trivial center and we will then use those algebras $\cD_j$ to construct $\cC_j$ by commutants as in the previous case.
We will construct the $\cD_j$ in two different steps, which, roughly, correspond to removing the terms in the center supported on the left and right of $I_j$; of course, as we remarked, we do not assume that the center factorizes into a center on the left and right end.

First we describe the step that, roughly, "removes the terms in the center on the left end".
Let ${\mathcal Q}$ be the algebra of operators in ${\mathcal A}$ which are supported on the set of sites
$\{1,\ldots,l\}$.
Let $\Pi$ be any minimal nonzero projector in ${\mathcal Q}$, i.e., $\Pi$ is any projector in this algebra such that $\Pi\neq 0$ but such that there is no other projector $\Pi'\neq 0$ in ${\mathcal Q}$ with $\Pi'<\Pi$.
Note that $\Pi$ is in $\cB_j$ for all $j$.
For each $j$,
let $\cB'_j$ be the algebra of operators in $\cB_j$ which commute with $\Pi$.
Then, we restrict to the $+1$ eigenspace of $\Pi$, i.e., rather than considering the full Hilbert space which is a tensor product of the Hilbert space on each site, we restrict this full Hilbert space to the $+1$ eigenspace of $\Pi$.

We now apply the structure theorem of Eq.~(\ref{Mform}) to $\cB_j$ for arbitrary $j$; for notational clarity we will add a $j$ in parenthesis to the matrices and projectors in the decomposition, writing $M(j)_a$ and $P(j)_a$.
Recall that $\Pi\in \cB_j$ for all $j$.  The projector $\Pi$, when written in the form (\ref{Mform}), has $M(j)_a$ equal to some projector for each $a$.  If we take a full matrix algebra, consider the subalgebra of operators commuting with some projector, and then restrict to the $+1$ eigenspace of that projector, then the resulting algebra still has trivial center.
Thus, for all $j$ the center of the algebra $\cB'_j$ is generated by operators $\Pi P(j)_a$ for each $a$.  For some $a$, it is possible that $\Pi P(j)_a=0$.

As shown above, any operator $P(j)_a$ in the center of $\cB_j$ is supported on the union of the left and right ends and so can be written as
$P(j)_a=\sum_\beta A(\beta) O_L^\beta O_R^\beta$ with $O_L,O_R$ supported on the left and right ends and $O_L^\beta,O_R^\beta\in \cB_j$, and $A(\beta)$ being a scalar.
The operators $O_L^\beta$ commute with $\Pi$ since $P(j)_a$ is in the center of $\cB_j$
 and the support of $O_R^\beta$ is disjoint from the support of $\cB_1$.
Further, the operators $O_L^\beta$ are in ${\mathcal Q}$.
Restricting to the $+1$ eigenspace of $\Pi$, we claim that $O_L^\beta$ must be proportional to the identity operator for each $\beta$.
If this were not true for some $\beta$, then $\Pi$ would not be minimal as the algebra generated by $\Pi$ and $O_L^\beta$
would contain a nonzero projector $\Pi'$ with $\Pi'<\Pi$.

Hence, the center of $\cB'_j$ is supported on the right end of the interval.
Since we restrict to the $+1$ eigenspace of $\Pi$, note that for any $O$ with $[O,\Pi]=0$ we have $\Pi O = O$.

Now, in the next step, we, roughly speaking, "remove the center on the right".
We remind the reader of the proof of theorem \ref{llinduct} in the linear case; we will use the same technique here to remove the center on the right, in particular using lemma \ref{inbetweenlemma}.
Note that $\cB'_j\subset \cB'_{j+1}$.
We now construct a simple algebra $\cD_j$ such that $\cB'_j \subset \cD_j \subset \cB'_{j+1}$.
Since the right end of $I_j$ is disjoint from the right end of $I_{j+1}$, for any operators $P'(j)$ in the center of $\cB'_j$ and $P'(j+1)$ in the center of $\cB'_{j+1}$, we have that
${\rm tr}(P'(j)_b P'(j+1)_a)={\rm tr}(P'(j)_b) {\rm tr}(P'(j+1)_a)/{\rm tr}(I)$.
Hence, the existence of $\cD_j$ follows from lemma \ref{inbetweenlemma}.

Now define $\cE_j$ to be the commutant of $\cD_{j-1}$ in $\cD_j$.
Since each $\cD_j$ is simple, each $\cE_j$ is also simple, the $\cE_j$ commute with each other, and the set of $\cE_j$ for $j=1,\ldots, K$ generates $\cD_j$.

Let $\cF$ be the algebra generated by $\cE_2,\cE_3,\ldots,\cE_{K-1}$.
Any operator $O$ in ${\mathcal A}$ supported on $I_{K-1}\setminus I_1$ commutes with $\Pi$ and also $\Pi O$ is nonvanishing, i.e., we have an injection $f$ from the set of operators supported on $I_{K-1} \setminus I_1$ to $\cD_K$ given by multiplying the operator by $\Pi$.
Further, $O$ commutes with $\cD_1$, so that $O$ is in $\cF$.

Further, we have an injection $g$ from $\cD_j$ to ${\mathcal A}$, mapping an operator in $\cD_j$ to $\Pi O$ for some $O$ in ${\mathcal A}$ with $[O,\Pi]=0$.
The composition of these two injections $g \circ f$ maps $O$ to $\Pi O$.

Now consider any operator $O$ in $\cE_j$ for $2\leq j \leq K-1$.
This operator commutes with $\cD_{j-1}$.  Hence, $g(O)$ has no support on sites $2l+1,\ldots,C\cdot (j-1)-l$ because if it had support on such sites, by definition of a visibly simple algebra there would be some operator $P$ supported on $l+1,\ldots,C\cdot (j-1)$ which does not commute with $g(O)$ and so
$f(P)$ would not commute with $O$ since the support of $P$ is disjoint from the support of $\Pi$.
So, $g(O)$ is supported on sites $1,\ldots,2l$ and $C\cdot (j-1)-l+1,\ldots,C\cdot(j+1)$.
By the assumption that ${\mathcal A}$ is locally factorizable, we can write $g(O)=\sum_\beta A(\beta) O_L^\beta O_R^\beta$ with $O_L,O_R$ supported on
$1,\ldots,2l$ and $C\cdot (j-1)-l+1,\ldots,C\cdot(j+1)$ respectively.  Then, $O_L^\beta$ commutes with $\Pi$ and indeed $O_L^\beta$ equals the image under $g$ of some operator in $\cD_1$.  However, since $\cD_1$ is simple, this operator must be a scalar.

Hence, $g(O)$ is equal to $\Pi$ times some operator supported on $C\cdot (j-1)-l+1,\ldots,C\cdot (j+1)$.
Hence, every operator in $g(\cF)$ is equal to $\Pi$ times an operator supported on $l+1,\ldots,CK$.

We are now finally ready to define $\cC_j$ for $2\leq j \leq K-1$.
For any $O$ in $\cE_j$, we have that $g(O)$ is equal to $\Pi$ times an operator $P$ supported on 
$C\cdot (j-1)-l+1,\ldots,C\cdot (j+1)$.  We define $\cC_j$ to be the algebra generated by such operators $P$.
The support of the algebras $\cC_j$ is as claimed; the algebras are simple and commute with each other.

Finally, we define $\cC_1$ to be the commutant of the algebra generated by $\cC_2,\ldots,\cC_{K-1}$ in ${\mathcal A}$.
\end{proof}
\end{thm}

\subsubsection{Commutant of Visibly Simple Algebra}
We note that
theorem \ref{llsllfthm} has the following corollary.  This corollary will not be needed later, however, so it may be skipped.
\begin{lemma}
\label{commutatntllsllf}
Consider a set of sites labelled by integers $1,2,...,J$, using a graph metric to measure distance where site $j$ is connected to sites $j \pm 1 \, {\rm mod} \, J$ and consider a tensor product Hilbert space with a finite dimensional Hilbert space for each site.  Suppose an algebra ${\mathcal A}$ of operators  is $l$-visibly simple and $l$-locally factorizable.
Let ${\mathcal B}$ be the commutant of ${\mathcal A}$ in the algebra of all operators on the tensor product Hilbert space.
Then, ${\mathcal B}$ is $O(l)$-locally factorizable and $O(l)$-visibly simple.
\begin{proof}
Define algebras ${\mathcal C}_i$ as in theorem \ref{llsllfthm}.
Suppose $O \in {\mathcal B}$ is supported on sets $T_1 \cup T_2$ $T_1,T_2$ and $O$ is decomposed using a singular value decomposition
using the Hilbert-Schmidt inner product as $O=\sum_\beta A(\beta) O_1^\beta O_2^\beta$, for some discrete index $\beta$ and complex scalars $A(\beta) \neq 0$.  
Then, if ${\rm dist}(T_1,T_2)$ is sufficiently large compared to $l$, there is no set $C_i$ which has nonvanishing intersection with both $T_1$ and $T_2$.
In this case, if $O$ commutes with every element of ${\mathcal A}$, it means that $O_1^\beta$ commutes with every element of ${\mathcal C}_i$ for which $C_i$
has non-vanishing intersection with $T_1$.  Hence, $O_1^\beta$ is in ${\mathcal B}$, and similarly for $O_2^\beta$.  This shows that ${\mathcal B}$ is $O(l)$-locally factorizable.

Next, suppose that $O \in {\mathcal B}$ has support on some site $x$.  Then, there is some operator (in the algebra of all operators) supported on site $x$ which does not commute with $O$.  Call this operator $Z$.  Then, since ${\mathcal A}$ is simple, we can decompose every operator (and hence in particular $Z$) as
a sum of products of operators in ${\mathcal A},{\mathcal B}$ as
\be
\label{Zdecomp}
Z=\sum_\alpha A(\alpha) O_{\mathcal A}^\alpha O_{\mathcal B}^\alpha,
\ee
using a singular value decomposition with Hilbert-Schmidt inner product and complex scalars $A(\alpha)$.

Let $C(x)$ denote the set of $i$ such that $x \in C_i$.
The operator $Z$ commutes with ${\mathcal C}_i$ if $i \not \in C(x)$.
So, we can assume in Eq.~(\ref{Zdecomp}) that $O_{\mathcal A}^\alpha$ is in the algebra generated by the set of ${\mathcal C}_i$ for $i \in C(x)$.
Call this algebra ${\mathcal C}(x)$.
This algebra ${\mathcal C}(x)$ is supported within distance $O(l)$ of $x$.
Then, we have that
\be
\label{reconstruct}
O_{\mathcal B}^\alpha={\rm const.} \times A(\alpha)^{-1} {\rm tr}_{{\mathcal C}(x)}\Bigl( O_{\mathcal A}^\alpha Z \Bigr),
\ee
where the trace is over the algebra ${\mathcal C}$.
The constant in front denotes the fact that the Hilbert-Schmidt inner product used to define the decomposition is defined with a different trace; this constant is equal to ${\rm tr}(I)/{\rm tr}_{{\mathcal C}(x)}(I)$.

The trace over ${\mathcal C}(x)$ can be written as, for an arbitrary operator $Q$,
\be
\label{form}
{\rm tr}_{{\mathcal C}(x)}(Q)={\rm const.} \times \int_{U \in {\mathcal C}(x)} {\rm d}U \; U Q U^\dagger,
\ee
where the integral is over unitaries $U$ in ${\mathcal C}(x)$ using a Haar measure, where now the constant is equal to ${\rm tr}_{{\mathcal C}(x)}(I)$.

So, combining Eqs.~(\ref{Zdecomp},\ref{reconstruct},\ref{form}), and using the fact that the unitaries $U$ in Eq.~(\ref{form}) are supported within distance $O(l)$ of $x$, we find that $O_{{\mathcal B}}^\alpha$ is supported within distance $l$ of $x$ for all $\alpha$.
There must be at least one $\alpha$ such that $O_{{\mathcal B}}^\alpha$ does not commute with $O$, so this shows that ${\mathcal B}$ is $O(l)$-visibly simple.
\end{proof}
\end{lemma}

\subsection{Blending for QCA in $2d$}
\label{nolo}
We now can consider blending using the results above.

Let us also define what we mean by a "one-dimensional boundary".
\begin{defin}
Let $G$ be a graph.  That say that a set $F$ of vertices has a one-dimensional boundary with bounded distortion if
it is possible to define a mapping $f$ from sites in $\Bd(F)$ to $\{1,\ldots,J\}$ for some $J $ such that the following holds.
Define a metric ${\rm dist}(i,j)$ on elements of $\{1,\ldots,J\}$ using a graph metric where site $j$ is connected to sites $j\pm 1$ mod $J$ such that
\be
{\rm dist}(f(i),f(j)) \leq C {\rm dist}(i,j)+C'R
\ee
and
\be
{\rm dist}(i,j) \leq  C{\rm dist}(f(i),f(j))+C'R
\ee
for some constants $C,C'$.
\end{defin}
\begin{remark}
This notion is sometimes made into the equivalence relation "coarse quasi-isometric".
\end{remark}

\begin{remark}
Certainly one might consider more general definitions of "bounded distortion", allowing different ways in which the distance changes.  This definition will be general enough for us.
This definition is useful as if ${\rm dist}(i,j)=O(R)$ for some pair of sites $i,j\in \Bd(F)$, then ${\rm dist}(f(i),f(j))=O(R)$; then, if one finds some $k$ such that ${\rm dist}(f(i),k)=O(R)$, also
${\rm dist}(f^{-1}(f(i)),f^{-1}(k))=O(R)$.  This mapping is not necessarily one-to-one so $f^{-1}(k)$ may be a set of sites.
\end{remark}

A simple example of this is to consider the graph of a discretization of a torus, $H(L_0)$, for some $L_0$ and take $F$ to be the set of sites $(i,j)$ with $1\leq i,j\leq L$, i.e., a smaller square.
Then, the constants $C,C'$ can both be chosen $O(1)$ and $F$ has a boundary with bounded distortion.  Since we refer to this example below, we refer to this choice of $F$ as $F(L)$.

\begin{thm}
\label{extendtotriv}
Let $\alpha$ be a QCA with range $R$.
Suppose there is a set of sites $F$ such that the boundary has bounded distortion.
Then, we can tensor with additional degrees of freedom on $\Bd(F) \cap F$ and construct a QCA $\beta$ with range $O(R)$ such that $\alpha \otimes \Id$ agrees with $\beta$ on $\Int(\Int(F))$ and
$\beta$ agrees with $\Id$ on $\Ext(F)$.  The $O(R)$ notation hides a dependence of the range of $\beta$ on the constants $C,C'$.
\begin{proof}

By theorems \ref{thmfac},\ref{thmls}, $\alpha({\mathcal A}(F))$ is generated by $<{\mathcal A}({\Int(F)}),\cPalphaF>$ where $\cPalphaF$ acts on $\Bd(F)$ and is $O(R)$-visibly simple and $O(R)$-locally factorizable.  
Using that $F$ has a one-dimensional boundary with bounded distortion, we "coarse-grain", defining a new set of "supersites" labelled by $1,\ldots,J$ for some $J$.
A supersite labelled by $i$ is the tensor product of sites in $f^{-1}(i)$.
The algebra ${\mathcal P}$ can be regarded as acting on the supersites and is still $O(R)$-visibly simple and $O(R)$ locally factorizable (here we use that $F$ has a one-dimensional boundary with bounded distortion).
This coarse-graining is done so that
we can apply theorem
\ref{llsllfthm} to show that
$\cPalphaF$ is generated by algebras ${\mathcal C}_i$ where
each ${\mathcal C}_i$
is an algebra of operators supported on some set $C_i$ of diameter $O(R)$ with $C_i \subset \Bd(F)$.  Note that $C_i$ is a set of sites, not of supersites.
The algebras ${\mathcal C}_i$ all commute with each other and have trivial center.

Let ${\mathcal C}_i$ have dimension $D_i^2$.
Note that
\be
\label{dimeq}
\prod_i D_i = \prod_{j \in \Bd(F) \cap F} \overline d_j,
\ee
where $d_j$ is the Hilbert space dimension on site $j$.

We tensor with additional degrees of freedom as follows.  For each algebra ${\mathcal C}_i$, we tensor with an ancilla
degree of freedom with dimension $D_i$ on any (arbitrarily chosen) site $j\in \Bd(F)$ such that $f(j)$ is in on the support of ${\mathcal C}_i$.  
 Let $c_i$ denote the added site corresponding to algebra ${\mathcal C}_i$.

We now define QCA $\beta$.  It will be defined to agree with $\alpha$ on $\Int(\Int(F))$.  We define $\beta$ to agree with $\Id$ on the original degrees of freedom which are in the complement of $F$.

The Hilbert space with the additional degrees of freedom can be written as a tensor product of three Hilbert spaces, corresponding to the original degrees of freedom in the complement of $F$, to the original degrees of freedom in $F$, and to the additional degrees of freedom (which are all supported on $\Bd(F)$.  
Write these three Hilbert spaces as ${\mathcal H}_{\overline F},{\mathcal H}_F,{\mathcal H}_A$, respectively,
and write ${\mathcal A}(F),{\mathcal A}(\overline F),{\mathcal A}(A)$ for the algebra of operators acting on these three spaces, respectively.
We say that an operator is supported on one of these three Hilbert spaces if it is equal to some operator on that Hilbert space tensored with
the identity on the other Hilbert spaces.
We have defined $\beta(O)=O$ for $O$ supported on ${\mathcal H}_{\overline F}$.
We now define $\beta(O)$ for $O$ supported on ${\mathcal H}_F$ or ${\mathcal H}_A$.
For use later, write ${\mathcal H}_{\Int(F)}$ for the Hilbert space of the original degrees of freedom in $F$ and ${\mathcal H}_{F \cap \Bd(F)}$ for
the Hilbert space of the original degrees of freedom in $F \cap \Bd(F)$, and write ${\mathcal A}(\Int(F)),{\mathcal A}(F \cap \Bd(F))$ for the algebras of
operators acting on those Hilbert spaces, respectively.

We define $\beta(O)$ for $O$ supported on ${\mathcal H}_F$ as follows.
For any such $O$,  $\alpha(O)$ 
 can be written as a sum of products $\sum_{\alpha} O_I^\alpha O_B^\alpha$ where
$O_I^\alpha$ is supported on the degrees of freedom in $\Int(F)$ and $O_B^\alpha$ is supported on $F \cap \Bd(F)$ and $O_B^\alpha$ is in $\cPalphaF$.
Define (arbitrary) isomorphisms from operators acting on each ${\mathcal C}_i$ to operators
acting on the corresponding added site $c_i$.
This gives an isomorphism from the tensor product of ${\mathcal C}_i$ to ${\mathcal A}(A)$; call this isomorphism $h(\cdot)$.  
We define $\beta$ so that $\beta(O)=\sum_\alpha O_I^\alpha \otimes h(O_B^\alpha)$.
Thus, while $\alpha$ maps ${\mathcal A}(F)={\mathcal A}(\Int(F)) \otimes {\mathcal A}(F \cap \Bd(F))$ to ${\mathcal A}(\Int(F))$ times the boundary algebra, $\beta$ maps ${\mathcal A}(F))={\mathcal A}(\Int(F)) \otimes {\mathcal A}(F \cap \Bd(F))$ to ${\mathcal A}(\Int(F))$ times ${\mathcal A}(A)$.

For $O$ supported on $\Int(\Int(F))$, $\alpha(O)$ is supported on $\Int(F)$ and so $\alpha$ and $\beta$ agree on $\Int(\Int(F))$.

Finally, we define the action of $\beta$ on the sites $c_i$.
Roughly, the idea is to "use wires to map ${\mathcal A}(A)$ to ${\mathcal A}(F \cap \Bd(F))$".
Here, the algebra of operators on the original boundary degrees of freedom is ${\mathcal A}(F \cap \Bd(F))$ and a "wire" is (roughly) a one-dimensional QCA.
Let us describe a first attempt before giving the actual construction.
Suppose can find some site $c_i$ with an additional degree of freedom with dimension $D_i$ and find some site $j \in \Bd(F) \cap F$ such that the original degree of freedom has dimension $d_j$ such that
$D_i=d_j$, then 
we tensor with additional degrees of freedom, with these degrees of freedom located
at
some sites $s_1,s_2,\ldots,s_k$, in $\Bd(F) \cap F$.  The additional degrees of freedom will all have dimension $D_i$, and the $s_i$ are chosen so that so that ${\rm dist}(s_i,s_{i+1})\leq R$ and ${\rm dist}(c_i,s_1)\leq R$ and ${\rm dist}(s_k,j)\leq R$.
We then define $\beta$ to act as a shift as follows: it maps an operator on $c_i$ to the corresponding operator on $s_1$, maps an operator on $s_i$ to the corresponding operator on $s_{i+1}$, and maps an operator on $s_k$ to the corresponding operator on $j$.
When we say "corresponding", this requires choosing some (arbitrary) isomorphisms between these algebras which all have the same dimension.

We would like to pair up all additional degrees of freedom $c_i$ in this manner.
Unfortunately, this runs into a technical detail.  It is possible that we cannot pair off all sites $c_i$ in this manner.  For example, it may be possible that there is one site $c_1$ with dimension $d_1=15$ and two sites $j=1,2$ with dimensions $d_1=3,d_2=5$, respectively.
However, this problem is simple to solve.  We factor each $D_i$ into prime factor and then write the additional degree of freedom Hilbert space on site $c_i$ as a tensor product of Hilbert space with dimension corresponding to the prime factors of $D_i$, and similarly factor the original degree of freedom for each $j \in \Bd(F) \cap F$
into a tensor product of Hilbert space with dimension corresponding to the prime factors of $d_j$.
We then find a one-to-one matching between tensor factors of sites $c_i$ and tensor factors of sites $j\in \Bd(F) \cap F$.  Such a matching is possible 
by Eq.~(\ref{dimeq}) and by uniqueness of prime factorization.
We then use the same wire construction described in the above paragraph: for each pair of tensor factors we tensor with additional degrees of freedom on
sites $s_1,s_2,\ldots,s_k$, all with dimension equal to the given prime factor of some $D_i$, and then define a shift that maps an operator on the tensor factor of some ${\mathcal C}_i$ to an operator acting on $s_1$, 
maps an operator on $s_i$ to the corresponding operator on $s_{i+1}$, and maps an operator on $s_k$ to the corresponding operator on a tensor factor of some $j$.

Having defined the action of $\beta$ for operators supported on any one of the Hilbert spaces, ${\mathcal H}_{\overline F},{\mathcal H}_F,{\mathcal H}_A$, we
define $\beta$ for an operator which is a product $O_{\overline F} O_F O_A$ of operators supported on those spaces by
$\beta(O_{\overline F} O_F O_A)=\beta(O_{\overline F}) \beta(O_F) \beta(O_A)$ and we define $\beta$ for
a sum of products by linearity.

This completes the definition of $\beta$.
It remains to verify that $\beta$ has range $O(R)$ as claimed and that $\beta$ indeed is a $*$-algebra automorphism.
The range follows immediately from the construction. 

We now verify that $\beta$ is an automorphism, in particular we show that
for any operators $O,P$ we have $\beta(O) \beta(P)=\beta(OP)$.  
The QCA $\beta$ agrees with the identity on ${\mathcal H}_{\overline F}$, i.e., $\beta$ is equal to the identity automorphism on ${\mathcal H}_{\overline F}$ tensored with some map on ${\mathcal H}_F \otimes {\mathcal H}_A$.  So,
it suffices to consider $O,P$ which are supported on ${\mathcal H}_F \otimes {\mathcal H}_A$.
For $O,P$ both supported on ${\mathcal H}_F$, $\beta(O) \beta(P)=\beta(OP)$
follows from the fact that $\alpha$ is a QCA and $h(\cdot)$ is an isomorphism.
For $O,P$ both supported on ${\mathcal H}_A$, one may verify it from the construction.
For $O$ supported on ${\mathcal H}_A$ and $P$ supported on ${\mathcal H}_F$, one may verify that $[\beta(O),\beta(P)]=0$ since $\beta(P)$ is
supported on ${\mathcal H}(\Int(F)) \otimes {\mathcal H}_A$ and $\beta(O)$ acts trivially on those spaces.
So, for any $O=O_F O_A$ with $O_F,O_A$ supported on ${\mathcal H}_F,{\mathcal H}_A$ respectively, and any $P=P_F P_A$ with
$P_F,P_A$ supported on ${\mathcal H}_F,{\mathcal H}_A$ respectively,
$\beta(O) \beta(P)=\beta(O_F O_A) \beta(P_F P_A)=\beta(O_F) \beta(O_A) \beta(P_F) \beta(P_A)=\beta(O_F) \beta(P_F) \beta(O_A) \beta(P_A)=
\beta(O_F P_F) \beta(O_A P_A)=\beta(O_F P_F O_A P_A)=\beta(O P)$.
\end{proof}
\end{thm}

We now give some examples using the graph $H(L_0)$ from before and picking $F$ to be $F(L)$ as above.
\begin{remark}
We emphasize that it might be necessary to connect far separated added sites with wires.  For example, suppose all sites have the same dimension and suppose $\alpha$ acts as the identity everywhere in $S$ except along some line of sites, where this line is the sites with $i=L/2$.  There, $\alpha$ acts as a shift, mapping an operator supported on $(i,j)$ to the corresponding operator supported on $(i,j+1)$.
Along this line, we have a nontrivial GNVW index.
To do the wiring, we must connect some site at the top (i.e., at $(i,L)$ to some site at the bottom, i.e., at $(i,1)$.  This must be done by running the wire all the way around the boundary since we do not let the wires enter the interior.

Indeed, it may be necessary to run many wires.  Suppose that $\alpha$ acts as the shift on all columns, mapping an operator supported on $(i,j)$ to the corresponding operator supported on $(i,j+1)$ for all $i$.  Then, we need to run many wires (a number proportional to $L$ from the top to the bottom.
\end{remark}

\begin{remark}
Fnally, we make an aesthetic remark. As we noted above, it may be necessary to run many wires.  This may seem unaesthetic as the number of sites that we tensor with is proportional to $L^2$, i.e., there are roughly $L$ wires, each requiring roughly $L$ added sites.
Then, locally we may need to add of order $L$ sites even to a small local patch.
If this seems unaesthetic, then one can reduce this number by "running the wires into the bulk": one can move the wire out from $S$, have it move through $\Ext(S)$ in some path, and finally return to $S$.  Since this is purely an aesthetic question, we leave the details of this to the reader.
\end{remark}

\subsection{No Homotopy Invariants}
\label{noho}
In this subsection, we show 
quantum circuit equivalence between various QCAs.  We will often phrase the result as a stable path equivalence, but the given paths will imply a stable equivalence up to quantum circuits.
First we note a general result which holds for an arbitrary graph.

\begin{defin}
Let $S,T$ be any two subsets of the vertex set $V$.  We say that there is a (discrete) deformation from
$S$ to $T$ 
if the following holds for some constants $R,R'$.
There is a map $f(i,t)$ with $i\in S$ and $t$ an integer and with the range of $f$ being vertices of $G$,
such that $f(i,0)=i$, $f(i,T)\in T$ for some integer $T>0$,
${\rm dist}(f(i,t),f(i,t+1))\leq R'$ and for any $i,j$, if ${\rm dist}(i,j)\leq R$ then ${\rm dist}(f(i,t),f(j,t))\leq R'$ for all $t$.

If we wish to make the constants $R,R'$ explicit, we say that there is an $(R,R')$-deformation from $S$ to $T$.
\end{defin}
\begin{lemma}
\label{interiorthm}
Let $\alpha$ be a QCA with range $R$ such that $\alpha$ agrees with $\Id$ on $V\setminus S$ for some set $S$ (i.e., $\alpha$ is nontrivial only on $S$).  
Suppose there is an $(R,R')$-deformation from $S$ to $T$.  Then,
$\alpha$ is stably $3R'$-equivalent to some QCA $\beta$ such that
$\beta$ agrees with $\Id$ on $V \setminus T$.
\begin{proof}
For each site $i$ in $S$, let $i$ have dimension $d_i$.  Tensor with additional degrees of freedom on all sites $f(i,t)$ for $t=1,\ldots,T$ with dimension $d_i$.
Let $(i,t)$ for $t>0$ denote the degree of freedom corresponding to $i$ that we tensor with on site $f(i,t)$ and let $(i,0)$ denote the original degree of freedom on site $i$.
Define $\Swap_{(i,t_1),(i,t_2)}$ to swap degree of freedom $(i,t_1)$ with degree of freedom $(i,t_2)$.  Note that for each site $j$, it is possible that $j$ we tensor with several additional degrees of freedom on $j$ since $j$ may be the image of $f(i,t)$ for several different pairs $i,t$.

\begin{remark}
We slightly overload notation since previously we were using pairs $(i,j)$ to label sites in a two-dimensional lattice and here $(i,t)$ is used differently; no confusion should arise.
\end{remark}

Consider the discrete sequence of QCA
\be
\gamma_t=\Bigl( \prod_{i\in S} \Swap_{(i,t),(i,0)}\Bigr) (\alpha\otimes \Id) \Bigl( \prod_{i \in S} \Swap_{(i,t),(i,0)} \Bigr),
\ee
so that $\gamma_0=\alpha\otimes \Id$ and we define $\beta=\gamma_T$.

From this discrete sequence of $\gamma_t$, we can construct a continuous path of QCA from $\gamma_0$ to $\gamma_T=\beta$ as follows.
We will first construct such a continuous path from $\gamma_{t-1}$ to $\gamma_t$ for each $t=1,\ldots,T$ and then  obtain the path from $\gamma_0$ to
$\gamma_T$ by composing those paths.

Let
$W_1=V_1$ and let $V_t=V_t V_{t-1}^\dagger$.
Then, $W_t \gamma_{t-1} W_t^\dagger=\gamma_t$ for $t\geq 1$.
Each $W_t$ is equal to  $\prod_{i\in S}\Swap_{(i,t),(i,t-1)}$.
Each $\Swap_{(i,t),(i,t-1)}$ acts only on two degrees of freedom which are a distance at most $R'$ apart,
and so $W_t$ is $R'$-path equivalent to $\Id$.
Let
$\delta_{t,s}$ denote the QCA in this path, where $s$ is the path parameter with $\delta_{t,0}=\Id$ and $\delta_{t,s}=W_t$.
Hence, 
the composition of QCA $\delta_{t,s} \circ \gamma_{t-1} \circ \delta_{t,s}$ gives a continuous path from 
from $\gamma_{t-1}$ to $\gamma_t$.
By the assumption that there is an $(R,R')$ deformation from $S$ to $T$, each $\gamma_{t-1}$ has range at most $R'$.
Hence, since $\delta_{t,s}$ and $\gamma_{t-1}$ have range $R'$, the
composition $\delta_{t,s} \circ \gamma_{t-1} \circ \delta_{t,s}$ has range at most $3R'$.
Hence, the result follows.
\end{proof}
\end{lemma}

Now we show an example application of this:
\begin{thm}
\label{applic}
Consider graph $H(L_0)$; recall that this is the graph for a control space which is a $2$-torus.
Let $\alpha$ be a QCA with range $R$.
Then, $\alpha$ is stably $O(R)$-equivalent to a QCA $\gamma$ such that $\gamma$ agrees with the identity everywhere
except the set of sites $(i,j)$ such that $i=1$ or $j=1$ or both; we call this set of sites $Q$.
\begin{proof}
\begin{remark}
We can regard $\beta$ as a QCA that acts trivially everywhere except for two one-dimensional lines.
\end{remark}

Pick some $L$ roughly equal to $L_0/2$.  Then $G$ contains a subgraph $F(L)$ which has a one-dimensional boundary with bounded distortion. Apply theorem \ref{extendtotriv}
to find a QCA $\beta$ which agrees with $\alpha$ on $\Int(\Int(S))$ and $\beta$ agrees with $\Id$ on $\Ext(S)$.

By lemma \ref{interiorthm}, $\beta$ is $O(R)$-stably equivalent to $\Id$ (the set of sites on which $\beta$ acts nontrivially has an $(R,O(R))$  
deformation to a single site) and also so is $\beta^{-1}$.
Hence, $\alpha$ is  $O(R)$-stably equivalent to $\alpha \otimes \beta^{-1}$.
Let $\Swap_S$ be the product over sites in $S$ of a unitary that swaps original and ancilla degrees of freedom; this is a QCA of range $0$.
Since $\Swap_S$ is a product over sites of a unitary that acts only on that site, $\Swap_S$ is $0$-path equivalent to $\Id$.
Let $\Swap_S(t)$ describe such a continuous path, with $\Swap_S(0)=\Id$ and $\Swap_S(t)=\Swap_S$.
Consider the path of QCA
$$\Bigl(\alpha \otimes \Id\Bigr) \circ \Swap_{S}(t)^{-1} \circ \Bigl( \Id \otimes \beta^{-1} \Bigr) \circ \Swap_S(t).$$
This shows that
$\alpha \otimes \beta^{-1}$ is $O(R)$-stably equivalent to
$$\Bigl(\alpha \otimes \Id\Bigr) \circ \Swap_{S} \circ \Bigl( \Id \otimes \beta^{-1} \Bigr) \circ \Swap_S,$$
which agrees with $\Id$ on $\Int(\Int(S))$ and agrees with $\alpha\otimes \Id$ on $\Ext(S)$.

(The trick that we use here, considering a path from $\Swap_S$ to $\Id$ is a slight modification of an unpublished idea of A. Kitaev.  There, a swap QCA was defined which swapped degrees of freedom on {\it all} cites, and a path from that QCA to the identity was used to show that, for any QCA $\alpha$ on any graph, there is a continuous path from $\alpha \otimes \alpha^{-1}$ to $\Id$.  See also \cite{arrighi2011unitarity}.)

Let $T=V\setminus \Int(\Int(S))$.
Then, note that there is an $(R,O(R))$-deformation from $T$ to $Q$.
Hence, the result follows by theorem \ref{interiorthm}.
\end{proof}
\end{thm}

Theorem \ref{applic} is just a special case of a general result for any two dimensional manifold.  
\begin{thm}
\label{twodfull}
Consider a (sequence of) metrics on some manifold with bounded local geometry.  Given any (sequence of) QCAs on such a manifold, with range $R=O(1)$, each QCA is stably $O(R)$-equivalent to a shift QCA acting on degrees of freedom on a set of cycles which form a basis for the first homology of the manifold.
\end{thm}
This theorem should be compared to theorem \ref{thmcohomology} which shows that a QCA $\alpha$ determines a class $[\alpha]^\wedge \in H_1^\textup{lf}(X;M)/\textup{Torsion}$.  Theorem \ref{twodfull} implies that (if torsion is not present), all QCA in the same class are stably $O(R)$-equivalent.

\begin{proof}
Give a cell decomposition of the manifold.  For some $R'=O(R)$,
the QCA is stably $R'$-equivalent to a QCA that agrees with the identity on degrees of freedom on the interior of the $2$-cells.
Take the length of the $1$-cells to be much larger than $R'$ but still $O(R)$.  Also take the distance between any two $1$-cells which are not attached to the same $0$-cell to be much larger than $R'$ but still $O(R)$.  Take the diameter of the $2$-cells to be $O(R)$.
The resulting QCA is a shift QCA on the degrees of freedom on the $1$-cells up to a quantum circuit, since the length of the $1$-cells is large compared to $R,R'$.

On each $1$-cell, we can define the GNVW index.   Since the range of the QCA is short compared to the length of the $1$-cell is possible to define a GNVW index for each edge independently.
We pick an arbitrary orientation on each edge to define the sign of this index.
Define a $1$-chain with coefficient on each edge equal to this GNVW index.
The coefficients are logarithms of positive rationals, with addition as the group multiplication rule.
One may verify that the boundary of this $1$-chain is equal to zero, i.e., it is a $1$-cycle.
Let us call this $1$-cycle the "flow" and write it $\vec v$.

We claim that any two flows which represent the same homology class define QCA which are equivalent up to a 
a quantum circuit of depth $O(1)$ and range $O(R)$.
To see this, note that if $\vec v$ is homologous to $\vec v'$, then $\vec v=\vec v'+\partial \vec w$ for some $2$-chain $\vec w$.
Call the $2$-cells $F_a$, labelled by some index $a$.
Given this $2$-chain $\vec w$, for each $a$,
we tensor with ancilla degrees on each site on the edges around the boundary of the $2$-cell $F_a$ and find some quantum circuit that gives a flow on the edges of the boundary equal to the gradient of the $\vec w_a$.
This can be done by a quantum circuit of depth $O(1)$ and range $O(R)$ because we assume that each $2$-cell has diameter $O(R)$.
We can do these quantum circuits for each square in parallel; call the resulting circuit $\delta$.
This gives the equivalence between two QCA with homologous flows.
The Eilenberg swindle that we discussed in the introduction is one specific realization of this construction.
\end{proof}

\section{Real Projective Spaces}
\label{RPn}
We finally consider QCA on $X=RP^n$ for various $n$.  We have $H_1(X;\Z)=\Z_2$ for $n\geq 2$.  Hence, the index of section \ref{topologysection} is not useful in classifying such QCA as the first homology group modulo torsion is trivial.

This motives the following question: consider an $RP^1$ inside $RP^n$.  Consider a system with degrees of freedom (for example, qubits, but any other degree of freedom may be consider) on this $RP^1$, and consider a QCA that implements a shift of such degrees of freedom.  Can such a QCA be realized a quantum circuit?  To be explicit about what we mean by a shift QCA, since $RP^1$ is diffeomorphic to $S^1$, we consider equally spaced degrees of freedom arranged on a circle, and imagine a shift QCA acting on those degrees of freedom.

Note first that with qubit degrees of freedom, any two such QCAs whose fluxes differ by an {\it even} multiple of $\log(2)$ are related by a quantum circuit.  More generally, any two QCAs whose fluxes differ by {\it even} multiples of $\log(d)$ (with qudit degrees of freedom) are related by a quantum circuit.  For $RP^2$, this follows because the $RP^1$ inside $RP^n$ has $D^2$ attached to it by a two-to-one covering map.  One can construct a quantum circuit on $D^2$ which generates a flux on the boundary which is an arbitrary multiple of $\log(d)$.  The two-to-one covering map doubles this flux.  For higher $n$, one considers the $RP^2$ inside $RP^n$.

For $RP^2$, the results of section \ref{algebraicsection} show indeed that such a shift QCA can {\it not} be realized by a finite depth quantum circuit.
However, we can also give a more direct argument:
\begin{lemma}
\label{rp2}
The shift QCA on $RP^1$ inside $RP^2$ cannot be realized by a fdqc.
\begin{proof}
Suppose there were such a fdqc.
Take a sufficiently large radius tubular neighborhood ( a Mobius band) around $RP^1$ inside $RP^2$ and run the fdqc on this neighborhood, i.e., include only the gates of the circuit supported on this neighborhood.
It will produce single site hoping on $RP^1$ and an "error" QCA on near the boundary of the Mobius band. The GNVW index of the error term is an integer and when we dimensionally reduce everything to $RP^1$ that integer gets multiplied by two. The result is a flux which is an odd multiple of $\log(d)$ on a circle created by a fdqc, giving a contradiction.
\end{proof}
\end{lemma}

Now consider $RP^3$.  We will also be able to show that
\begin{lemma}
\label{rp3}
The shift QCA on $RP^1$ inside $RP^3$ cannot be realized by a fdqc.
\begin{proof}
Let $c$ be such a circuit.  We will show a contradiction.  Delete all gates in $c$ lying in a ball $B^3\subset RP^3$, with the ball disjoint from $RP^1$, calling the result $c^-$.

The circuit $c^-$ products a shift on $RP^1$.  The circuit $c^-$ has some uncontrolled behavior (it is still local, but we do not know what it is) near the $S^2=\partial B^3$.  We refer to such uncontrolled behavior as "garbage".
Finally, $c^-$ acts as the identity on all sites in $RP^3\setminus B^3$ which are sufficiently far from $\partial B^3$ and are not in $RP^1$.  

Let $E$ be the submanifold of $RP^3\setminus B^3$ consisting of sites which are not near $S^2=\partial B^3$, so that on $E$, the circuit $c^-$ acts as the identity on most sites except as the shift on those sites in $RP^1$.  $E$ is a twisted interval bundle over $RP^2$ with projection $\pi:E\rightarrow RP^2$ and $\partial E=S^2$.

The circuit $c^-$ defines a QCA $U(c^-)$ near $E$.  So, $U(c^-)$ can be decomposed as
\be
U(c^-)=g \otimes \Id \otimes h,
\ee
where $g$ is the "garbage" near $S^2$, $h$ is the shift on $RP^1$ and $\Id$ is the identity on remaining sites in $E$.

$g$ acts on a collar $S^2 \otimes [0,1]$ which can be dimensionally reduced to $S^2$.  By the results of section \ref{algebraicsection}, there exists a fdqc $d$ which realizes $g^{-1}$.
So, the concatenation $d \circ c^-$ is a fdqc on $E$ producing the shift $h$.  However, the projection $\pi$  dimensionally reduces $d \circ c^-$ to a fdqc on $RP^2$ producing the shift on $RP^1$, contradicting lemma \ref{rp2}.
\end{proof}
\end{lemma}

It remains open whether or not the shift QCA on $RP^1$ inside $RP^n$ for $n>3$ can be realized by a fdqc.
We remark that there is an unpublished result of A. Kitaev that any QCA on $d$ dimensions can be realized as the boundary of a fdqc in $d+1$ dimensions.  Precisely, let $\alpha$ be any QCA on any control space $X$.  The product $\alpha \otimes \alpha^{-1}$ can be realized by a fdqc.  Hence, $(\alpha \otimes \alpha^{-1}) \otimes (\alpha \otimes \alpha^{-1}) \otimes \ldots$ can be realized by a fdqc on $X \times \N$, where $\N=\{0,1,2,\ldots\}$  Follow this fdqc by $\Id \otimes (\alpha^{-1} \otimes \alpha) \otimes (\alpha^{-1} \otimes \alpha) \otimes \ldots$, which can also be realized by a fdqc.  The result is $\alpha \otimes \Id \otimes \Id \ldots$, realizing the QCA $\alpha$ on the boundary of control space $X \times \N$.
Similarly, we can realize an arbitrary QCA $\alpha$ on $S^n$ as the boundary of a quantum circuit on $B^{n+1}$: use a similar construction to
realize$\alpha\otimes \Id \otimes \Id \ldots \otimes \Id \otimes \Id \otimes \alpha^{-1}$ on $S^n \times \{0,1,\ldots,L\}$ for some odd $L>0$.  Then, let $S^n\times \{0\}$ be on the boundary of $B^{n+1}$, let $S^n\times \{1\}$ be a shell slightly inside the boundary, let $S^n\times \{2\}$ be a shell slightly further inside, and so on, with $S^n \times \{L\}$ being a single point in the center of $B^{n+1}$ so that the QCA $\alpha^{-1}$ on $S^n\times \{L\}$ can be realized by a fdqc (indeed by a single gate).
So, with this result, one can realize arbitrary $n-1$ dimensional QCA on the boundary of a $B^n$ inside $RP^n$.  Since for $n\geq 3$, the $n-1$ dimensional QCA might not itself be a quantum circuit in $(n-1)$-dimensions\cite{FHH}, this prevents direct application of the
proof of lemma \ref{rp3} to the cases $n>3$.

\bibliography{qcac-ref}
\end{document}